\documentclass[a4paper,twoside]{ThesisStyle}

\usepackage{amsmath,amssymb}             

\usepackage{lmodern}
\usepackage[T1]{fontenc}

\usepackage[left=0.75in,right=0.75in,top=1.1in,bottom=1.1in,includefoot,includehead,headheight=22.6pt]{geometry}

\usepackage[nottoc, notlof, notlot]{tocbibind}
\usepackage{minitoc}
\usepackage{aecompl}

\usepackage[intoc]{nomencl}

\makenomenclature

\usepackage{ifpdf}

  \usepackage[pdftex]{graphicx}
  \DeclareGraphicsExtensions{.jpg}

\graphicspath{{.}{images/}}
\usepackage{color}
\definecolor{linkcol}{rgb}{0,0,0.4} 
\definecolor{citecol}{rgb}{0.5,0,0} 

\usepackage{rotating}                    
\usepackage{fancyhdr}                    

\let\headruleORIG\headrule
\renewcommand{\headrule}{\color{black} \headruleORIG}

\usepackage{colortbl}
\arrayrulecolor{black}

\fancypagestyle{plain}{
  \fancyhead{}
  \fancyfoot{}
  
}

\usepackage{algorithm}
\usepackage[noend]{algorithmic}

\makeatletter

\def\cleardoublepage{\clearpage\if@twoside \ifodd\c@page\else%
  \hbox{}%
  \thispagestyle{empty}
  \newpage%
  \if@twocolumn\hbox{}\newpage\fi\fi\fi}

\makeatother
 
{%

\hrulefill
\vspace*{0.5cm}%
\end{minipage}
}

\let\minitocORIG\minitoc
\renewcommand{\minitoc}{\minitocORIG \vspace{1.5em}}

\usepackage{multirow}
\usepackage{slashbox}

{ \begin{list}%
	{$\bullet$}%
	{\setlength{\labelwidth}{25pt}%
	 \setlength{\leftmargin}{30pt}%
	 \setlength{\itemsep}{\parsep}}}%
{ \end{list} }

\renewcommand{\epsilon}{\varepsilon}

\usepackage{longtable}
\usepackage{lscape}
\usepackage{stmaryrd}
\usepackage{subfigure} 
\usepackage{booktabs}
\usepackage{pdfpages}
\usepackage{alltt}
\usepackage{amsthm}

\usepackage{epstopdf}

\newtheorem{thm}{Theorem}[section]

\newcommand{\dif}[2]{\frac{\text{d}#1}{\text{d}#2}}
\newcommand{\diff}[2]{\frac{\partial #1}{\partial #2}}
\newcommand{\difff}[2]{\frac{\partial^2 #1}{\partial #2^2}}
\newcommand{\difffTwo}[3]{\frac{\partial^2 #1}{\partial #2 \partial #3}}
\newcommand{\avar}[1]{\left\langle #1 \right\rangle}
\newcommand{\klamm}[1]{\left( #1 \right)}
\newcommand{\norm}[1]{\left\| #1 \right\|}

\newcommand{\istobe}{{\overset{!}=}}

\newcommand{\ofR}{\klamm{\bf r}}
\newcommand{\ofRD}{\klamm{\bf r'}}
\newcommand{\NewStuff}[1]{#1} 

\newcommand{\Deltan}{\Delta {\bf n}^k}
\newcommand{\nk}{{\bf n}^k}
\newcommand{\nkpe}{{\bf n}^{k+1}}

\newcommand{\defi}{{\mathrel{\mathop:}=}}

\usepackage{MnSymbol}%
\usepackage{wasysym}

\begin{document}

\pagenumbering{alph}
\begin{titlepage}
\begin{center}
\vspace*{0.3cm}
\vspace*{0.5cm}
\noindent \Huge \textbf{Statistical Mechanics of a Thin Film on a Solid Substrate} \\
\vspace*{0.3cm}
\noindent {\Large \textbf{Diploma Thesis, revised version}} \\
\vspace*{0.8cm}
\noindent \LARGE Andreas \textsc{Nold} \\
\vspace*{0.8cm}
\noindent \large submitted at the Technische Universit\"at Darmstadt \\
\noindent \large written at Imperial College London\\
\vspace*{0.8cm}
\noindent \large revised version: \today\\
\noindent \large defended on May 31, 2010 \\
\vspace*{1.2cm}
\end{center}

\begin{center}
\begin{tabular}{llcl}

      \textit{External Supervisor :}	& Prof. Serafim \textsc{Kalliadsis}		& - & Imperial College London\\
	\textit{Internal Supervisor :}	& Prof. Martin \textsc{Oberlack}		& - & Technische Universit\"at Darmstadt\\
\end{tabular}
\end{center}
\vspace{2cm}
{\bf Abstract:}
We study the behavior of very thin liquid films wetting homogeneous planar and spherical substrates. In order to describe 
a simple fluid at very small scales, we employ a classical density functional theory (DFT). Here, we model a fluid with a local density approximation
(LDA) for its hard-sphere contribution and assume that the intermolecular attractive forces are long-range.
In particular, we first introduce the basic concept of DFT, and then present computations for fluid films on planar and spherically symmetric walls.
We present equilibrium density profiles and adsorption isotherms. We also compare our results to predictions from a sharp-interface approximation (SIA) and 
suggest a piecewise function approximation (PFA), which assumes that the density profile at the wall-liquid and the liquid-vapor interfaces varies smoothly. 

\end{titlepage}
\sloppy

\titlepage

\minitoc
\dominitoc

\setcounter{tocdepth}{1}
\tableofcontents
\mainmatter

\chapter{Nomenclature}

\begin{longtable}{p{3cm}m{12cm}}
$\alpha$ & interaction energy per unit density in a uniform fluid (Eq. (\ref{eq:Uniform_def_alpha}))\\
$\beta$ & $1/(k_B T)$ where $k_B$ is the Boltzmann constant and $T$ is the temperature\\
$\delta$ & width of the wall-liquid interface\\
$\bar \delta$ & Tolman length\\
$\delta_{ij}$ & Kroenecker-Delta\\
$\delta\klamm{x}$ & Dirac delta function\\
$\varepsilon$ & depth of the Lennard Jones potential of the fluid-fluid interaction\\
$\varepsilon_w$ & depth of the Lennard Jones potential of the wall-fluid interaction \\
$\gamma$ & surface tension (Eq. (\ref{eq:StatMech_DefSurfaceTension}))\\
$\gamma_{lg,\infty}$ & surface tension of a planar liquid-gas interface \\ 
$\gamma_{lg,R}$ & surface tension of a liquid-gas interface of a droplet with radius $R$ \\
$\kappa$ & width of the liquid-gas interface\\
$\ell$ & thickness of a film on a solid substrate\\
$\ell^\ast$ & maximal thickness of a film on a spherical substrate\\
$\mu$ & chemical potential of a system (Eq. (\ref{eq:StatMech_DefChemPot}))\\
$\mu_{sat}(T)$ & chemical potential at which bulk liquid and bulk gas phase are equally stable\\
$\mu_{HS}\klamm{n}$ & hard-sphere chemical potential, defined in dimensionless  form in Eq. (\ref{eq:Mu_HS_Def})\\
$\Delta \mu$ & deviation of the chemical potential from its saturation value $\mu-\mu_{sat}$\\
$\phi(r)$ & Lennard-Jones-Potential for the fluid-fluid interaction (Eq.(\ref{eq:LennardJonesPotential}))\\
$\phi_{attr}(r)$ & attractive interaction potential between two particles at distance $r$ (Eq. (\ref{eq:PerturbationPotential_dimless}))\\
$\phi_{attr,W}(r)$ & attractive interaction potential by Weeks (Eq. (\ref{eq:WeeksPerturbation}))\\
$\phi_w(r)$ & Lennard-Jones potential for the wall-fluid interaction (Eq.(\ref{eq:Wall_LennardJonesPotential}))\\
$\xi_{I,V}$ & typical deviations from the Gibbs dividing surface (Eq. (\ref{eq:WettingSphere_IntermediateResult})\\ 
$\varrho$ & probability density distribution\\
$\varrho_0$ & equilibrium  probability density distribution\\
$\sigma$ & soft-core parameter of the LJ-potential of the fluid-fluid interaction (Eq.(\ref{eq:LennardJonesPotential}))\\
$\sigma_w$ & soft-core parameter of the LJ-potential of the wall-fluid interaction\\
\caption{Lower-Case Greek}
\end{longtable}

\begin{longtable}{p{3cm}m{12cm}}
$\Omega_0$ & grand canonical potential (Eq.(\ref{eq:DefGrandCanoncialPotential}))\\
$\Omega_B$ & binding potential \\
$\Omega_{ex}$ & excess grand potential (Eq.(\ref{eq:Def_ExcessGrandPotential}))\\
$\Omega^{SIA}$ & sharp-interface approximation of the grand potential \\
$\Omega^{PFA}$ & piecewise function approximation of the grand potential \\
$\Phi_{\text{Pla}}(z)$ & attractive interaction potential between a point in the fluid and a plane at distance $z$\\
$\Phi_{sph}(r,r')$ & attractive interaction potential between a point in the fluid at distance $r$ from the origin and the surface of a sphere with radius $r'$\\
\caption{Upper-Case Greek}
\end{longtable}

\begin{longtable}{p{3cm}m{12cm}}
$d$ & hard-sphere diameter\\
$h$ & Planck's constant \\
$h(r)$ & pair correlation function (Eq. (\ref{eq:PairCorrelationFunction}))\\
$k_B$ & Boltzmann constant $1.3806504 × 10^{-23} J/K$\\
$f_{HS}(n)$ & local hard sphere free energy\\
${\bf g}$ & vector function for the discretized minimization problem (Eq. (\ref{eq:g_Discretized}))\\
$g\klamm{r}$ & pair distribution function\\
$g_{HS}\klamm{r}$ & pair distribution function of a hard-sphere fluid (Eq. (\ref{eq:StatMech_PairDistributionFunction}))\\
$n\ofR$ & particle density (Eq.(\ref{eq:DefEqdensity})) \\
$n_l$ & density of the liquid bulk phase\\
$n_g$ & density of the gas bulk phase\\
$\Delta n$ & difference of the liquid bulk and gas bulk densities $n_l -n_g$\\
$n^{(2)}\klamm{{\bf r}_2,{\bf r}_2}$ & two-particle distribution \\
$n_{HS}^{(2)}\klamm{{\bf r}_1,{\bf r}_2}$ & two-particle distribution of a hard-sphere fluid (Eq.(\ref{eq:HStwoparticleDistr}))\\
$p$ & pressure of a system \\
${\bf p}$ & momentum vector\\
${\bf r}$ & position vector\\
$d{\bf r}$ & infinitesimal volume element\\
$y$ & packing fraction $=\frac{\pi}{6}nd^3$\\
\caption{Lower-Case Roman}
\end{longtable}

\begin{longtable}{p{3cm}m{12cm}}
$A$ & area of an interface\\
$C\klamm{{\bf r}_1,{\bf r}_2}$ & direct correlation function (Eq. (\ref{eq:StatMech_OrnsteinZernicleEquation}))\\
$E_k\klamm{{\bf r}^N}$ & kinetic energy of a system\\
$F$ & Helmholtz Free Energy (Eq. (\ref{eq:DefHelmholtzFreeEnergy}))\\
$F_{id}$ & ideal gas contribution to the Helmholtz free energy\\
$H_N\klamm{{\bf r}^N,{\bf p}^N}$ & Hamiltonian of a system (Eq.(\ref{eq:StatisticalMech_DefHamiltonian}))\\
${\bf J}$ & Jacobi matrix of ${\bf g}$ (Eq.(\ref{eq:DenProf_Jacobi_Def}))\\
$N$ & number of particles of a system\\
$R$ & radius of the spherical wall\\
$S$ & entropy of a system\\
$T$ & temperature of a system\\
$T_c$ & critical temperature  \\
$T_{cw}$ & complete wetting temperature, separates complete wetting ($T > T_{cw}$) from prewetting ($T < T_{cw}$)\\
$T_w$ & wetting temperature, separates partial wetting ($T< T_w$) from prewetting ($T > T_w$)\\
$U\klamm{{\bf r}^N}$ & particle interaction energy of a system (Eq. (\ref{eq:StatMech_DefInteractionEnergy}))\\
$U_{attr}\klamm{{\bf r}^N}$ & attractive particle interaction energy of a fluid (Eq. (\ref{eq:StatMech_DefUAttr}))\\
$U_{HS}\klamm{{\bf r}^N}$ & particle interaction energy of a hard-sphere fluid\\
$V$ & volume of a system\\
$V_f$ & volume of the film close to a solid substrate\\
$V_B$ & bulk volume\\
$V_{ext}\klamm{{\bf r}^N}$ & external potential energy of a system (Eq.(\ref{eq:StatMech_ExtPot}))\\
$V\ofR$ & external potential induced by a wall\\
$V_{\text{Pla}}(z)$ & external potential induced by a planar wall $W = \mathbb{R}^2\times \mathbb{R}^-$\\
$V_{sph,R}(r)$ & external potential induced by a spherical wall $W = \{ {\bf r} \in \mathbb{R}^3:  |{\bf r}| < R \}$\\
$V_{cav,R}(r)$ & external potential induced by a cavity $W = \{ {\bf r} \in \mathbb{R}^3:  |{\bf r}| > R \}$\\
$W$ & volume occupied by the solid substrate: $W \subset \mathbb{R}^3$ \\
$Z_C$ & canonical partition function\\
$Z_G$ & grand canonical partition function\\
\caption{Upper-Case Roman}
\end{longtable}

\chapter{Introduction}

The behavior of fluids (liquids or gases) at small scales, in
particular in the vicinity of solid substrates, is of paramount
significance in numerous technological applications and natural
phenomena. It is also of relevance to several fields, from
engineering to chemistry and biology. As a consequence, it has
received considerable attention, both experimentally and
theoretically, for several decades.

One of the most widely used methods for the study of the microscopic
structure of fluids is density-functional theory (DFT). It offers an
increasingly popular compromise between computationally costly
molecular dynamics simulations and phenomenological
approaches~\cite{Wu-DFT}. The basic idea of classical DFT is to
describe the microscopic properties of a fluid in terms of its
density distribution. The method can be derived consistently from
equilibrium statistical mechanics \NewStuff{of} fluids and is thus based on first
principles. It has been used successfully to study interfacial
phenomena, including wetting transitions on substrates.

In the present study we examine the equilibrium of a liquid film on
an attractive solid substrate, where we focus our attention on simple monatomic liquids. There also have been recent developments in the DFT-modeling of systems including chemical and hydrogen bonds and polymer systems~\cite{Wu-DFT}. However, here we focus our attention on the basic properties of the wetting behavior of a liquid film. For this, a nonlocal mean-field DFT
approach is adopted in which the grand potential as a thermodynamical potential is split into a repulsive
hard-sphere part and an attractive component.

We consider both planar and spherical substrates, thus
restricting our attention to 1D configurations. 2D problems of nanodrops and three-phase contact lines are adressed in \cite{RuckensteinNanorough,RuckensteinInclined,Antonio2010}. However, we consider a typical system
made of a planar/spherical wall in contact with a Lennard-Jones (LJ)
gas below the critical temperature. The wall exerts an attractive
force on the fluid molecules so that a thin liquid film can usually
form between the wall and the gas. The density of the fluid in the
presence of the wall is obtained by solving numerically an integral
equation resulting from the minimization of the grand potential.

\NewStuff{A comprehensive review of wetting phenomena on substrates is given} in Ref.~\cite{Dietrich}. In an earlier study, Napi\'orowski and Dietrich showed that the
so-called sharp-interface approximation (SIA), in which the liquid-gas interface is approximated by a step-function, simplifies the investigation of wetting phenomena on solid
substrates considerably, as with this approximation computations of the full
density profile are
avoided~\cite{DietrichNapiorkowski_BulkCorrelation}. This
approximation was then often used to investigate wetting transitions
on both planar~\cite{Dietrich} and curved
substrates~\cite{BykovZheng,DietrichWettingOnCurvedSubstrates} as
well as for the computation of contact angles~\cite{Dietrich2D}. 

However, the sharp interface approximation fails to give the correct liquid-gas surface tension. This leads to systematic errors in the prediction of the wetting behavior on curved substrates. Hence, we introduce a piecewise function approximation (PFA), for which the density is assumed to be everywhere constant except in the wall-liquid and the liquid-gas interface where it varies. In this work, the PFA as well as the SIA are introduced for general geometries, not being restricted to the planar or the spherical case. 

Beyond analytical approaches, rapid progress in computational power over the last few
years has allowed us to solve the DFT equations for the
full density profile in different wetting problems. Here, we give particular emphasis on the bifurcation
diagrams for \NewStuff{the excess density as a function of
chemical potential at a given temperature (adsorption isotherms)}. Such diagrams
are obtained from a pseudo arc-length continuation scheme. They are
typically multi-valued S-type curves often with a value of the
chemical potential above which no solutions exist and with three
branches of solutions from which the middle one is always unstable.

We first focus on a thin film in contact with a planar wall, which
is essential to understanding the substantially more involved
spherical case. The isotherms approach infinity as the deviation of the chemical
potential $\Delta \mu$ from saturation tends to zero from negative
values. We also examine in
detail the prewetting transition, a first-order phase transition occurring at a
specific value of the chemical potential where two equally stable
films, a thin one and a thick one, coexist.

We subsequently examine the case of a liquid film on a sphere. Analytically, applying the PFA allows for a simple way of examining the effects of the liquid-gas
surface tension in the wetting behavior of curved substrates. In a
somewhat related approach followed by Dietrich and Napi\'orowski for
the planar case \cite{DietrichNapiorkowski_BulkCorrelation}, the effects of the smooth liquid-gas interface are directly accounted for by the coefficients of an asymptotic expansion of the grand potential in inverse powers of the film thickness. Here, we introduce instead a number of auxiliary parameters, such as typical deviations  $\xi_{I,V}$ from the Gibbs dividing surface, which lead to an exact expression for the grand potential as a function of the film thickness and the radius of the substrate.
This allows separate asymptotic expansions in both the radius of the wall and the film thickness. 

These analytical results are compared with the numerical results obtained from the continuation procedure. For small film thickness, the bifurcation diagrams are similar to the planar case, while shifted towards values of higher chemical potential. We give analytical and numerical evidence that in the limit of zero curvature, this shift is directly related with the Laplace pressure. As a consequence, the bifurcation diagrams cross the $\Delta \mu = 0$ line such that, additionally to the prewetting transition at $\Delta \mu < 0$, a second first-order phase transition occurs. Hence, the film thickness does not go to infinity as saturation is approached but instead is limited to a maximal film thickness which exhibits a leading order $R^{1/3}$-dependence as a function of
the radius of the substrate~\cite{DietrichWettingOnCurvedSubstrates}. For $\Delta \mu > 0$, the  isotherms exhibit a second unstable branch
compared to the planar case one. This branch approaches the saturation line as
$(\Delta \mu)^{-1}$. These numerical results are found to be in excellent agreement with
the analytical predictions obtained by the PFA.

The thesis is organized as follows: In chapter~\ref{sec:StatMechDFT} we give a brief introduction to statistical thermodynamics, before presenting the basic mathematical theorems of classical DFT. In the following, we introduce several methods to model the free energy of the system in Sec.~\ref{sec:StatMech_ModelFreeEnergy}. One of the models, the perturbation approach, makes use of a hard-sphere fluid as a reference system, which we discuss in Sec.~\ref{sec:ModelsHardSphere}. In Sec.~\ref{sec:StatMech_OurModel}, we introduce the non-dimensionalization of the model used in this work. Phase diagrams for the homogeneous case are presented in Sec.~\ref{sec:UniformLiquid}, while non-homogeneous effects of the grand potential are linked with the surface energy and surface tension in Sec.~\ref{sec:SurfaceTensionExcessGP}.

In chapter \ref{sec:DensityProfiles} we treat the computation of single density profiles. In this context, we give details about numerical methods to solve the minimization problem (Sec.~\ref{sec:DenProf_Numerics}). In Sec.~\ref{sec:DenProf_PlanarWall}, we present analytical expressions for the case of a planar wall. We also give results for a pure liquid-gas interface and compare them with the SIA. Furthermore, the behavior of the density profile close to the wall is studied. In Sec.~\ref{sec:DenProf_Sphere}, we give analytical expressions for a spherical wall and compare density profiles with the planar case. 

In chapter \ref{sec:WettingBehavior} wetting on solid substrates is studied. For this, we introduce in Sec.~\ref{sec:Wetting_Analytical} formalisms for the SIA and the PFA which are not restricted to special geometries. In Sec.~\ref{sec:ContinuationMethod}, the pseudo-arc length continuation method is introduced. It is applied in Sec.~\ref{sec:Wetting_PlanarWall} to the case of a planar wall, where we compare the numerical results with analytical results from the SIA. In Sec.~\ref{sec:Wetting_Sphere}, wetting on a sphere is studied, where it is compared with the planar case and with the analytical prediction from the PFA.

\chapter{Statistical Thermodynamics and Density Functional Theory \label{sec:StatMechDFT}}
\section{Fundamentals of Statistical Thermodynamics \label{sec:StatThermodynamics}}

Statistical Mechanics deals with the description of systems with a large number of particles. Here, we want to describe a fluid with $N$ molecules, where $N$ is a very large number. However, we are not interested in the precise motion of each molecule, but instead want obtain relations between thermodynamic quantities such as pressure, temperature or density. Statistical Thermodynamics closes this gap between the microscopic behavior of a fluid and the macroscopic quantities. Gibbs described the link between both levels of description as follows\cite{Gibbs}:

\begin{quote}
\it
The laws of thermodynamics, as empirically determined,
express the approximate and probable behavior of systems of
a great number of particles, or, more precisely, they express
the laws of mechanics for such systems as they appear to
beings who have not the fineness of perception to enable
them to appreciate quantities of the order of magnitude of
those which relate to single particles, and who cannot repeat
their experiments often enough to obtain any but the most
probable results.
\end{quote}

It is our aim to describe the probability distribution of each microstate - characterized by one specific microscopic configuration of particles - as a function of macroscopic quantities. In other words: We want to know the probability of finding a macroscopic system with average energy $\avar{E}$ in a certain microstate at a certain point in time. 

In the sequel, we consider a canoncial ensemble, i.e. a collection of systems which is characterized by its number of particles $N$, its volume $V$ and its temperature $T$ \cite[p.20]{Hansen}. It is in contact with a heat reservoir of temperature $T$, with which it can exchange energy. However, it is closed, i.e. there is no exchange of particles between the system and the reservoir. 

J.W. Gibbs first introduced the idea of dealing with the specific microstates as identical copies of the same macroscopic state of a system~\cite{Schroedinger}. Each copy has the given temperature, volume and number of particles and is in thermal equilibrium with a heat bath. This means that each copy is exchanging energy with its environment. Now, assume that the number of copies $M$ is large, i.e. $M \to \infty$. Furthermore, we assume that the set of available microstates is discrete and that each microstate is equally probable. This is the fundamental postulate of statistical mechanics~\cite{Stowe}: 

\begin{quote}
An isolated system in equilibrium is equally likely to be in any of its accessible states, each of which is defined by a particular configuration of the system's elements.
\end{quote}

We say that the system is in equilibrium, if it attains its most probable distribution of microstates over the available energy levels. For a comprehensive proof of the method of the most probable distribution in the case of discrete energy levels $E_i$, see also Schr\"{o}dinger~\cite{Schroedinger}. The probability $p_i$ of being in a microstate at energy level $E_i$ is then equal to 
\begin{align}
p_i =& \frac{1}{Z} e^{- \beta E_i},\label{eq:StatisticalMech_PDiscrete}\\
\text{where}\qquad Z \defi& \sum_i e^{- \beta E_i}. \notag
\end{align}
$Z$ is the {\it partition function} and $\beta = 1/(k_B T)$ with Boltzmann constant $k_B$. The partition function will be used as a generator for all kinds of macroscopic properties. As an example, the average energy is given by
\begin{align}
\avar{E} &= \sum_i  E_i p_i
= \frac{1}{Z} \sum_i E_i e^{-\beta E_i}
= - \frac{1}{Z} \sum_i \dif{}{\beta} e^{-\beta E_i}
= - \frac{1}{Z} \dif{}{\beta} Z = \notag \\ 
&= - \dif{}{\beta} \ln Z. \label{eq:StatMech_EnergyCan}
\end{align}
The partition function is also directly connected to the entropy $S$ of the system. The statistical mechanical definition of entropy was formulated by Boltzmann for a microcanonical ensemble. A microcanonical ensemble is a closed isolated system, i.e. there is no exchange of energy or particles with its environment. In Boltzmann's formulation, the entropy is proportional to the logarithm of the number of possible microstates which a system can occupy. Hence, the entropy is a measure for the uncertainty inherent to a system: The less entropy a system has, the less microstates are available and consequently the more probable it is to find the system in one of the given microstates. Furthermore, we expect the entropy to be an extrinsic property, i.e. the entropy of two identical systems is the sum of the entropy of the two separate systems. However, the number of possible microstates of the two systems is squared compared to the entropy of the single systems. This property is accounted for by employing the logarithm. In the case of a canonical ensemble, this relation can be written as 
\begin{align*}
\avar{S} &= - k_B \sum_i p_i \ln p_i
\end{align*}
Inserting (\ref{eq:StatisticalMech_PDiscrete}) in the equation above leads to an expression of the entropy in terms of the partition function $Z$ and the average energy $\avar{E}$:
\begin{align}
\avar{S} &= - \frac{k_B}{Z} \sum e^{-\beta E_i} \klamm{ - \beta E_i - \ln Z }\notag\\
&= k_B\klamm{ \ln Z + \beta \avar{E} } \label{eq:StatMech_EntropieCan}
\end{align}
This relation leads to the introduction of the statistical mechanical definition of the Helmholtz free energy $F$:
\begin{align}
F &\defi - \beta^{-1}\ln Z,\label{eq:DefHelmholtzFreeEnergy}
\end{align}
which is in the literature often also denoted by $A$. It corresponds to the thermodynamic definition in an average sense:
\begin{align*}
F &= \avar{E} - T \avar{S}. 
\end{align*}
The equivalence of the thermodynamic and statistical mechanical definition can be shown easily by inserting (\ref{eq:StatMech_EnergyCan}) and (\ref{eq:StatMech_EntropieCan}) in the equation above. 

\paragraph{The Classical Limit}
Here, we assume that the difference between two energy levels is infinitely small ($\Delta E \to 0$). Hence, the set of microstates is continuous and each microstate is uniquely defined by the positions $\{{\bf r}_i:i=1\ldots N\}$ of its particles and their momentum $\{{\bf p}_i:i=1\ldots N\}$. The equilibrium probability function can be written as
\begin{align*}
\varrho_C\klamm{{\bf r}^N,{\bf p}^N} = \frac{1}{h^{3N}N!} \frac{1}{Z_C} e^{-\beta {H}_N\klamm{{\bf r}^N,{\bf p}^N}},
\end{align*}
where $Z_C$ is the canonical partition function in the classical limit:
\begin{align}
Z_C = \frac{1}{h^{3N}N!} \iint  e^{-\beta H\klamm{{\bf r}^N,{\bf p}^N}} d{\bf r}^N d{\bf p}^N. \label{eq:StatMech_Can_PartitionFunction}
\end{align}
${\bf r}^N$ and ${\bf p}^N$ are the arrays of position and momentum vectors for all particles, $\{{\bf r}_1,\ldots,{\bf r}_N\}$ and $\{{\bf p}_1,\ldots,{\bf p}_N\}$, respectively. It is assumed that the system consists of $N$ interchangeable particles. The division by $N!$ takes the number of permutations of the identical particles into account and assures a correct counting of the states. Plack's constant $h$ ensures that both the probability density function as well as the canonical partition function $Z_C$ are dimensionless. $H$ is the Hamiltonian function. It is the energy of the system as a function of the position and momentum of the particles. It is defined as 
\begin{align}
H_N\klamm{{\bf r}^N,{\bf p}^N} = E_k\klamm{{\bf p}^N} + U\klamm{{\bf r}^N} + V_{ext}\klamm{{\bf r}^N}, \label{eq:StatisticalMech_DefHamiltonian}
\end{align}
where $E_k$ is the kinetic energy of the system, $U$ is the interaction energy and $V_{ext}$ is the potential energy. Here, we say that the kinetic energy of the system is the sum of the kinetic energy of each particle, whereas the external energy does only depend on the position of the particles:
\begin{align}
E_k\klamm{{\bf p}^N} &\defi \sum_{i=1}^N \frac{|{\bf p_i}|^2}{2m}\notag\\
V_{ext}\klamm{{\bf r}^N} &\defi \sum_{i=1}^N V\klamm{{\bf r_i}}. \label{eq:StatMech_ExtPot}
\end{align}
$m$ is the mass of each particle, and $V_{ext}\klamm{\bf r}$ is an arbitrary external potential. For a more detailed introduction to the topic, see also Hansen and McDonald~\cite{Hansen}, or Stowe~\cite{Stowe}.

\paragraph{The Grand Canonical Ensemble - Legendre Transform}
\begin{figure}[ht]
\centering
\includegraphics{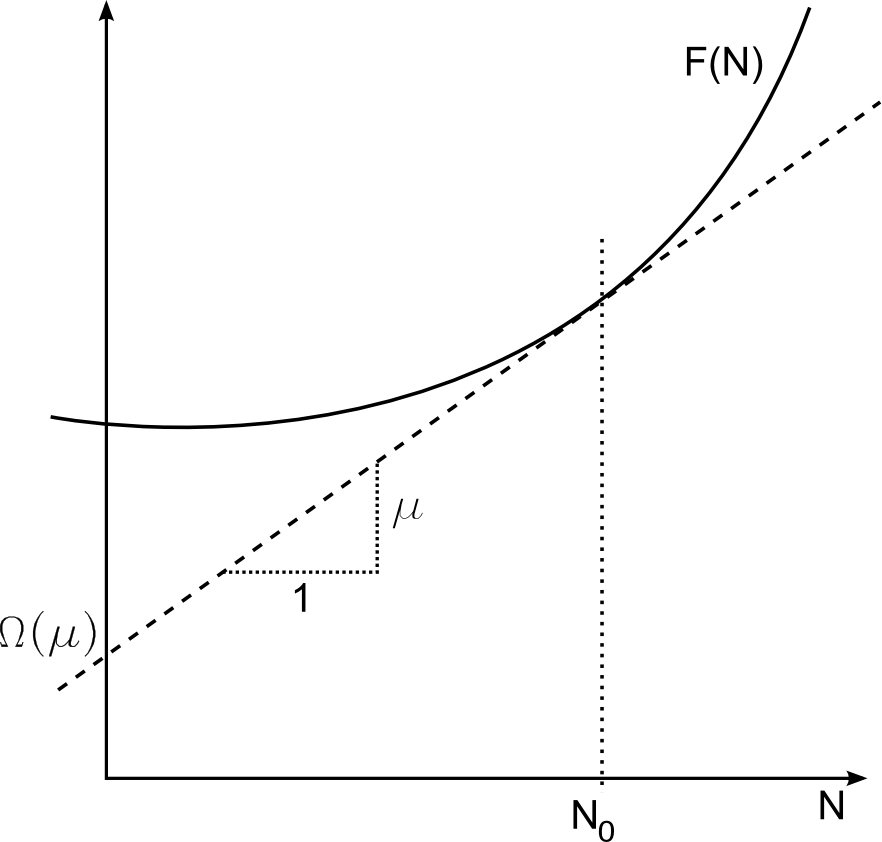}
\caption{Geometric interpretation of a Legendre transformation from the Helmholtz free energy $F$ as a function of the particle number $N$ to the grand potential function $\Omega$ as a function of the chemical potential $\mu$. At the intersection $N = N_0$, one gets $\mu = \dif{F}{N}$. $\Omega(\mu)$ is defined by the intersection of the tangent with the vertical axis.}
\label{fig:LegendreTransform}
\end{figure}
In the present work, we consider an open system which is in thermal equilibrium with a heat reservoir and for which the number of particles is not known. Instead, the chemical potential $\mu$ is known, defined as the derivative of the free energy of the system with respect to its number of particles:
\begin{align}
\mu \defi \klamm{\diff{F}{N}}_{T,V}. \label{eq:StatMech_DefChemPot}
\end{align}
This change of variables can be interpreted as follows: In the canonical ensemble, the potential of the system is given by the Helmholtz free energy $F\klamm{N,T,V}$. We now release $N$ and say that the derivative $\diff{F}{N}$ has to be equal to $\mu$. A variable transformation from a control variable $x$ to the derivative of the function with respect to this variable $f'$, is usually performed by means of the Legendre transformation $\{ f,x \} \to \{ f^\ast, p \}$ defined by:
\begin{align*}
f^\ast(p) \defi \max_x \klamm{ px - f(x) }
\end{align*}
For a pedagogical introduction to the applications of the Legendre transform in physics, see also Zia, Redish and McKay~\cite{Zia_LegendreTransform}. A geometric interpretation of the Legendre transform is given in Fig.\ref{fig:LegendreTransform}. Here, we introduce the grand potential $\Omega$ by the negative Legendre transform of the Helmholtz free energy:
\begin{align}
\Omega_0\klamm{\mu,T,V} &\defi \min_{N} \klamm{ F\klamm{N,T,V} - \mu N} \label{eq:DefGrandCanoncialPotential}
\end{align}
In the sequel, the values of $F\klamm{N,T,V} - \mu N$ and $N$ at the minimum will be called {\it equilibrium} values. They will be denoted by a subscript ''0''. Whenever $\Omega$ without subscript ''0'' is used, it will denote the value of $F\klamm{N,T,V} - \mu N$ at an arbitrary number of particles $N$. 

The thermodynamic interpretation of the grand potential in the homogeneous case can be derived from the Gibbs-Duhem equation\footnote{For one-component systems, the Gibbs-Duhem equation is $E = TS - pV + \mu N$}~\cite{Stowe} such that 
\begin{align}
\Omega_0 = - pV, \label{eq:StatMech_GrandCanonicalPotenial}
\end{align}
where $p$ is the pressure of the system.

\paragraph{Statistical Mechanical Definitions for an Open System}
For a grand canonical system where the number of particles is not known, the equilibrium probability function can be written as
\begin{align}
\varrho_0\klamm{{\bf r}^N,{\bf p}^N,N} = \frac{1}{Z_G} e^{-\beta \klamm{{H}\klamm{{\bf r}^N,{\bf p}^N} - \mu N} },\label{eq:Def_PDF_R0}
\end{align}
where $H$ is the Hamiltonian as defined in (\ref{eq:StatisticalMech_DefHamiltonian}) and $Z_G$ is the grand canonical partition function:
\begin{align*}
Z_G = \sum_{N=0}^\infty \frac{1}{h^{3N}N!} \int \int  e^{-\beta \klamm{H\klamm{{\bf r}^N,{\bf p}^N}-\mu N}} d{\bf r}^N d{\bf p}^N.
\end{align*}
In the literature, one often finds the symbol $\Xi$ for the grand partition function. By means of simplicity, we define the average over the probability distribution $\varrho_0$ by
\begin{align}
\avar{\cdot} \defi \text{Tr}\klamm{\varrho_0 \cdot}, \label{eq:Def_Average}
\end{align}
where $\text{Tr}$ is the trace. It is defined by
\begin{align*}
\text{Tr}\klamm{\cdot } \defi \sum_{N=0}^\infty \frac{1}{h^{3N}N!} \iint \cdot  d{\bf r}_i d{\bf p}_i.
\end{align*}
Again, the partition function can be used as a generator for all kind of macroscopic quantities. For an open system , the average entropy is defined as
\begin{align*}
\avar{S} &= - k_B \avar{\varrho_0}\\
&= k_B \klamm{ \ln Z_G + \beta \klamm{ \avar{H} - \mu N } }.
\end{align*}
The equation above leads directly to an expression of the grand potential $\Omega$ as a function of the partition function. This relation corresponds to the thermodynamic definition in an average sense:
\begin{align*}
\Omega_0 = - \beta^{-1} \ln Z_G = \avar{H_N} - \mu \avar{N} - T\avar{S}.
\end{align*}

\section{Basic Theorems of Density Functional Theory \label{sec:StatMech_DFT}}
The statistical mechanical formalism establishes a way of computing a probability density function over the microstates of a system. One microstate is defined by the number of particles $N$, their position ${\bf r}^N$ and momentum ${\bf p}^N$. However, computing the full probability density function leads to unnecessary high computational costs. In fact, we are only interested in obtaining the particle density as a function of space. In other words, we want to know the probability of finding a particle at a given position ${\bf r}$ of the volume. Mathematically, this can be written as
\begin{align}
n_0({\bf r}) &\defi \langle {\sum_{i=1}^N \delta\klamm{{\bf r}- {\bf r_i}}}\rangle. \label{eq:DefEqdensity}
\end{align}
DFT reformulates the Helmholtz free energy in terms of the particle density $n\ofR$, thus avoiding the computation of the full probability density function.

In order to do so, we have to show that the equilibrium probability density distribution $\varrho_0\klamm{{\bf r}^N,{\bf p}^N,N}$ is uniquely defined by the equilibrium particle density $n_0\ofR$. First, we introduce the functional
\begin{align}
\Omega[\varrho] \defi \text{Tr} \klamm{ \varrho \klamm{H_N - \mu N + \beta^{-1} \ln \varrho} } \qquad \klamm{\forall \varrho}  \text{Tr}\klamm{\varrho} = 1.\label{eq:DefOmega}
\end{align}
In the equilibrium case $\Omega[\varrho=\varrho_0]$ corresponds to the grand potential $\Omega_0$. We now show that the definition of the functional given above is consistent, i.e. we show that the equilibrium probability density $\varrho_0$ minimizes $\Omega[\varrho]$.

\begin{thm}{Minimumprinciple}
\begin{align}
\forall \varrho \neq \varrho_0 \qquad \Omega[\varrho_0] < \Omega[\varrho] \label{eq:MinimalPrinciple}
\end{align}
\end{thm}
\begin{proof}
With the definition of the equilibrium probability density $\varrho_0$ (see Eq.(\ref{eq:Def_PDF_R0})), it can be shown that 
\begin{align}
\beta^{-1} {\text{Tr} \klamm{\varrho\ln \varrho_0}} = \Omega_0 - \text{Tr}\klamm{ \rho \klamm{ H_N - \mu N } }
\end{align}
Hence, $\Omega[\varrho]$ can also be written as
\begin{align}
\Omega[\varrho] = \Omega[\varrho_0] + \beta^{-1} \klamm{\text{Tr} \klamm{\varrho\ln \varrho} - \text{Tr} \klamm{\varrho\ln \varrho_0 }}.
\end{align}
Now, it is our aim to show that the second term is strictly positive. Using $\text{Tr}\klamm{\varrho} = \text{Tr}\klamm{\varrho_0} = 1$, it can be rewritten as follows:
\begin{align}
\text{Tr} \klamm{\varrho\ln \varrho} - \text{Tr} \klamm{\varrho\ln \varrho_0} = \text{Tr}\klamm{ \varrho_0 \klamm{ \frac{\varrho}{\varrho_0}\ln \frac{\varrho}{\varrho_0} - \frac{\varrho}{\varrho_0} +1 } }
\end{align}
If we can show that the inner part of the right hand side is strictly positive for any positive nonunity value of $\frac{\varrho}{\varrho_0}$, then we are done. For this purpose, the problem is reformulated. We want to show, that for every $x> 0$ and $x \neq 1$, $f(x) \defi x\ln x - x + 1 > 0$. This can be easily shown by taking the derivative of $f$. Indeed, the minimum of $f(x)$ is at $x=1$, where $f(1) = 0$.
\end{proof}

The Minimumprinciple leads to the conclusion that the probability density $\varrho_0\klamm{{\bf r}^N,{\bf p}^N,N}$ is uniquely determined by the particle density $n\klamm{{\bf r}}$. The way to proof this is through the external potential.
\begin{figure}[ht]
\centering
\includegraphics{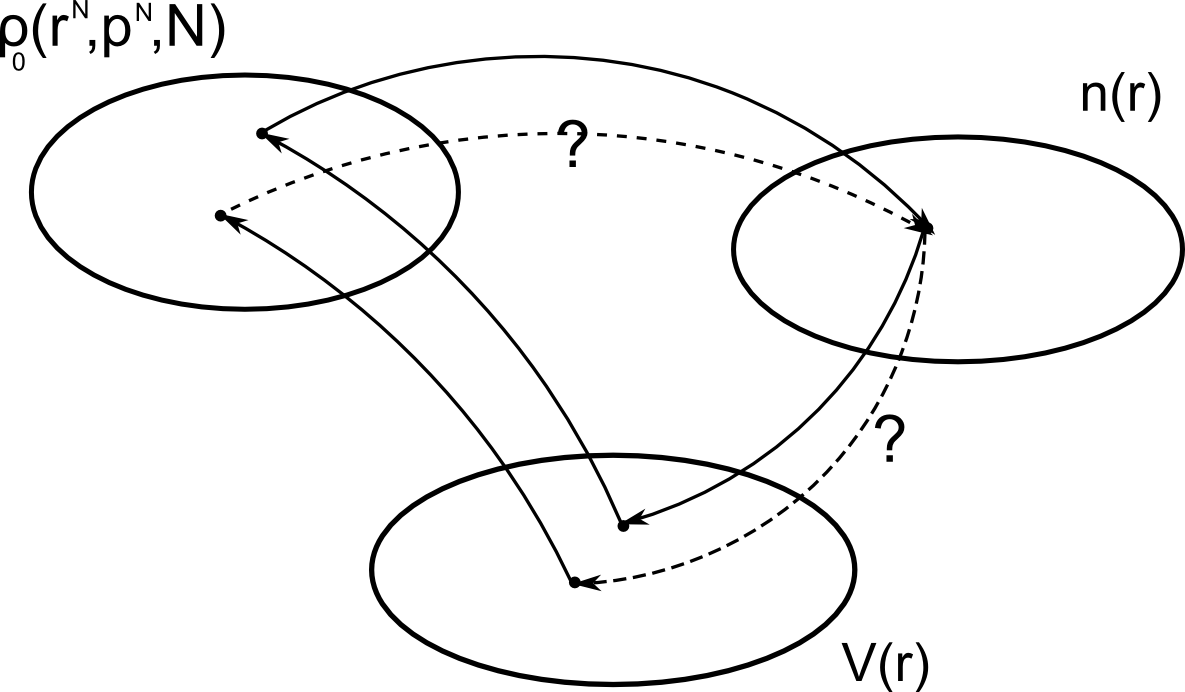}
\caption{Sketch of the mappings between equilibrium probability density $\varrho_0\klamm{{\bf r}^N,{\bf p}^N,N}$, particle density $n\ofR$ and external potential $V\ofRD$. By definition, the probability density is uniquely defined by the external potential $V\ofR$. We want to show that for each particle density $n\ofR$, there is not more than one equilibrium probability density $\varrho_0\klamm{{\bf r}^N,{\bf p}^N,N}$. This is done by showing that the mapping from the set of external potentials $V\ofR$ to the set of particle densities $n\ofR$ is injective.}
\label{fig:}
\end{figure}

\vspace{0.5cm}

\begin{thm}
The mapping which assigns an equilibrium particle density to a given external potential is injective. This means that for a given particle density, there is not more than one external potential.
\end{thm}
\begin{proof} Assume that for one particle density $n \ofR$, there are two different external potentials $V_{1}\klamm{\bf r},V_{2}\klamm{\bf r}$. Different external potentials map to different equilibrium probability distributions $\varrho_{0,1},\varrho_{0,2}$ (see also definition (\ref{eq:Def_PDF_R0})). We say that $\Omega_{1,2}$ is equal to $\Omega[\varrho]$ as defined in Eq.(\ref{eq:DefOmega}) with respect to the external energies $V_{1,ext}\klamm{{\bf r}^N}$ and $V_{2,ext}\klamm{{\bf r}^N}$, respectively. We obtain:
\begin{align}
\Omega_{01}= \Omega_1[\varrho_{01}] &{\overset {(\ref{eq:MinimalPrinciple})} < } \Omega_1[\varrho_{02}]\notag\\
 &{\overset {(\ref{eq:DefOmega})} =} \Omega_2[\varrho_{02}] +  \text{Tr} \klamm{ \varrho_{02} \klamm{V_{1,ext}\klamm{{\bf r}^N} - V_{2,ext}\klamm{{\bf r}^N}}}. \label{eq:Proof:1}
\end{align}
We have a closer look at the difference of the potential energies $V_{1,ext}\klamm{{\bf r}^N} - V_{2,ext}\klamm{{\bf r}^N}$. As defined in Eq.(\ref{eq:StatMech_ExtPot}), the external potential energy $V_{ext}\klamm{{\bf r}^N}$ can be written as a sum of the potentials over all the particles $i = 1\ldots N$:
\begin{align}
V_{1,ext}\klamm{{\bf r}^N} - V_{2,ext}\klamm{{\bf r}^N} &= \sum_{i=1}^N V_{1}\klamm{\bf r_i} - V_{2}\klamm{\bf r_i}\notag\\
&= \sum_{i=1}^N \int \delta\klamm{{\bf r}-{\bf r_i}} \klamm{ V_{1}\klamm{\bf r} - V_{2}\klamm{\bf r}} d{\bf r}. \label{eq:FormDFT_1}
\end{align}
In this context, the expression using the $\delta$-function is particularly important, as it allows the use of the particle density function $n\ofR$ in the sequel.
\begin{align}
\text{Tr} \klamm{ \varrho_{02} \klamm{V_{1,ext}\klamm{{\bf r}^N} - V_{2,ext}\klamm{{\bf r}^N}}}
&{\overset {(\ref{eq:FormDFT_1})} =} \text{Tr} \klamm{ \varrho_{02} \sum_{i=1}^N \int \delta\klamm{{\bf r}-{\bf r_i}} \klamm{ V_{1}\klamm{\bf r} - V_{2}\klamm{\bf r}}} d{\bf r}\notag\\
&= \int \text{Tr} \klamm{ \varrho_{02} \sum_{i=1}^N \delta\klamm{{\bf r}-{\bf r_i}} } \klamm{ V_{1}\klamm{\bf r} - V_{2}\klamm{\bf r}}d{\bf r} \notag\\
&{\overset {(\ref{eq:DefEqdensity})}  =} \int n_2 \ofR \klamm{ V_{1}\klamm{\bf r} - V_{2}\klamm{\bf r}} d{\bf r}, \label{eq:Proof:2}
\end{align}
where $n_2\klamm{\bf r}$ is the particle density with respect to $\varrho_{02}$ as defined in Eq.(\ref{eq:DefEqdensity}). Here, we used the linearity of the trace. Before taking the last step, we go back to (\ref{eq:Proof:1}). The last term can be replaced by (\ref{eq:Proof:2}). Furthermore, by symmetry arguments the same has to hold for $\Omega_{02}$. Then, we get:
\begin{align}
\Omega_{01} &< \Omega_{02} +  \int n_2 \ofR \klamm{ V_{1}\klamm{\bf r} - V_{2}\klamm{\bf r}} d{\bf r} \\
\Omega_{02} &< \Omega_{01} -  \int n_1 \ofR \klamm{ V_{1}\klamm{\bf r} - V_{2}\klamm{\bf r}} d{\bf r}
\end{align}
Obviously, this leads to a contradiction, if $n_1\ofR = n_2\ofR$.
\end{proof}

Consequently, for an existing equilibrium particle density function $n_0\ofR$, there is not more than one external potential $\varrho_0\klamm{{\bf r}^N,{\bf p}^N,N}$. Furthermore, for every external potential there is not more than one equilibrium probability density. So, starting with one equilibrium density profile $n_0\klamm{\bf r}$, there is not more than one probability density $\varrho_0\klamm{{\bf r}^N,{\bf p}^N,N}$ that corresponds to this profile. 

Consequently, $\varrho_0$ is uniquely defined by $n_0\ofR$. We conclude that $\Omega$ can be written as a functional of the particle density $n\ofR$. Obviously, the minimum principle (\ref{eq:MinimalPrinciple}) for $\Omega$ as a functional of the probability density function translates to $\Omega$ as a functional of the particle density $n\ofR$:
\begin{align}
\Omega[n_0] < \Omega[n] \qquad \klamm{\forall n} n \neq n_0.
\end{align}
Rewriting (\ref{eq:DefOmega}) in terms of the particle density $n\ofR$ leads to the expression
\begin{align}
\Omega[n] = F[n] + \int n\klamm{\bf r}\klamm{V\klamm{\bf r} - \mu}  d{\bf r}\label{eq:Def_Potential_Fct},
\end{align}
where $F[n]$ is the free energy of the system as a function of the particle density $n\ofR$. The exact expression for $F[n]$ is not known. Finding an expression for this functional is part of the fluid modelling. 

\subsection{Applying the Variational Principle}
We now want to find the particle density $n\ofR$ which minimizes the functional $\Omega[n]$. This is done using functional derivatives. For a comprehensive review into the topic, see also Parr and Yang~\cite{ParrYang} or Courant and Hilbert~\cite{CourantHilbert}.

The functional derivative  $\frac{\delta \Omega}{\delta n\ofR}$ of the functional $\Omega[n]$ at the point ${\bf r}$ is defined as 
\begin{align}
\lim_{\varepsilon \to 0} \frac{ \Omega[n + \varepsilon \eta] - \Omega[n] }{\varepsilon}  &=
\left.
 \dif{}{\varepsilon} \Omega[n + \varepsilon \eta]
\right|_{\varepsilon = 0}
= \int \frac{\delta \Omega}{\delta n\ofR} \eta\ofR d{\bf r},
\label{eq:StatMech_FunctionalDerivative}
\end{align}
where $\eta\ofR$ is an arbitrary twice continuously differentiable function which vanishes at the boundaries of the domain. Furthermore, we suppose that $\frac{\delta \Omega}{\delta n\ofR}$ is continuous and that $n\ofR$ minimizes the functional $\Omega[n]$. Then, the expression above has to vanish. By the {\it fundamental lemma of the calculus of variations}~\cite{CourantHilbert}, it follows that
\begin{align}
\frac{\delta \Omega}{\delta n\ofR} = 0. \label{eq:StatMech_EulerLagrange}
\end{align}
(\ref{eq:StatMech_EulerLagrange}) is also known as the {\it Euler-Lagrange Equation}. It is a necessary condition for an extremum of $\Omega$. There are several ways of computing the functional derivative $\frac{\delta \Omega}{\delta n\ofR}$. 

\paragraph{Application for Gradient Expansions of the free energy}
First assume that $\Omega[n]$ can be written in integral form as
\begin{align}
\Omega[n] = \int g(n\ofR,\nabla n\ofR, {\bf r}) d{\bf r},
\end{align}
where $g$ is a scalar function of the particle density, its gradient and the position ${\bf r}$. $g$ is supposed to be twice continuously differentiable in all of its arguments. Assume that $n_0\ofR$ is the desired function which minimizes $\Omega[n]$. We now introduce a variation $\varepsilon \eta\ofR$, for a twice continuously differentiable scalar function $\eta\ofR$, which vanishes at the boundaries of the volume. With a calculation similar to (\ref{eq:StatMech_FunctionalDerivative}), we get as a necessary condition:
\begin{align}
\left.\dif{}{\varepsilon}  \Omega[n+ \varepsilon \eta] \right|_{\varepsilon = 0} &= \int \eta\ofR g_n(n\ofR,n'\ofR, r) + \nabla\eta\ofR \cdot g_{\nabla n}(n\ofR,\nabla n\ofR, r) d{\bf r} \istobe 0,
\end{align}
where $g_{\nabla n}$ is the gradient of $g$ with respect to $\nabla n$. Remark that for the second term, the divergence theorem can be applied. One obtains
\begin{align}
\int_V \eta\ofR \klamm {g_n(..) - \text{div} g_{\nabla n}(..) } d{\bf r} + \int_{\partial V} \eta\ofR g_{\nabla n}(..) \cdot  d{\bf S} \istobe 0,
\end{align}
where $V$ is the volume of the system and $\partial V$ is its boundary. ${\bf S}$ is used in shorthand for the product of the normal vector of the boundary times an infinitesimal element of the surface. Due to the boundary conditions imposed on $\eta\ofR$, the surface term vanishes. As a result, the expression
\begin{align}
\int_V \eta\ofR \klamm {g_n(..) - \text{div} g_{\nabla n}(..) } d{\bf r}
\end{align}
has to vanish for all functions $\eta\ofR$. Comparing this to (\ref{eq:StatMech_FunctionalDerivative}) and (\ref{eq:StatMech_EulerLagrange}), we get that a necessary condition for an extremum is given by the {\it fundamental differential equation of Euler}~\cite{CourantHilbert}:
\begin{align}
\text{div} g_{n'} - g_n = 0 \label{eq:DFT_FundamentalDiffEq_Euler}.
\end{align}

\paragraph{Application for integral formulation of the free energy}
If $\Omega$ includes multiple integrals of the form
\begin{align}
\Omega[n] = \iint n\ofR n\ofRD h\klamm{{\bf r},{\bf r}'} d{\bf r} d{\bf r}',
\label{eq:StatMech_OmegaNonLocal}
\end{align}
then we rather make use of another approach. Remark that if in (\ref{eq:StatMech_FunctionalDerivative}), we replace $\eta\ofR$ by the Delta-function $\delta_{\bf r}\klamm{{\bf r}'} \defi \delta\klamm{|{\bf r}-{\bf r}'|}$, this gives a defining equation for $\frac{\delta \Omega}{\delta n\ofR}$ (see also Plischke and Bergersen~\cite{Plischke}):
\begin{align}
\frac{\delta \Omega[n]}{\delta n\ofR}
=
\left.
 \dif{}{\varepsilon} \Omega\klamm{
n + \varepsilon \delta_{\bf r}
}\right|_{\varepsilon = 0}
\end{align}
Applying this to (\ref{eq:StatMech_OmegaNonLocal}) gives:
\begin{align}
\frac{\delta \Omega[n]}{\delta n\ofR}
=
\int n\ofRD \klamm{h\klamm{{\bf r},{\bf r}'} + h\klamm{{\bf r}',{\bf r}} } d{\bf r}'.
\end{align}

\paragraph{Minimal Condition for the Equilibrium Particle Density}
Applying the above calculations on (\ref{eq:Def_Potential_Fct}) yields the variational equation
\begin{align}
\frac{\delta F[n]}{\delta n\ofR} + V\klamm{\bf r} - \mu = 0 \qquad \klamm{\forall {\bf r}}, \label{eq:StatMech_MinimalCondition}
\end{align}
where $\frac{\delta F[n]}{\delta n\ofR}$ is the functional derivative of $F[n]$ at ${\bf r}$.

\section{Models for the Free Energy \label{sec:StatMech_ModelFreeEnergy}}

One drawback of DFT is, that the exact expression for the Helmholtz free energy is lost when changing from the probability density function space to the particle density function space. Instead, an appropriate model for the free energy as a functional of the particle density $n\ofR$ has to be found. 

\subsection{Local Theory: The Square-Gradient Approximation}

In the sequel, we will follow the arguments presented by Cahn and Hilliard~\cite{CahnHillard_SquareGradient} in 1958. The basic assumption of this approach is that the free energy of a system does only depend on the local density and on the density of the immediate environment. The impact of the latter will be accounted for by the local density derivatives. The density and the local derivatives will be treated as independent variables of the local free energy:
\begin{align}
F[n\ofR] = \int f\klamm{n\ofR, \diff{n}{x_i},\difffTwo{n}{x_i}{x_j},\ldots } d{\bf r} 
\end{align}
Now, the local free energy $f$ is expanded in a taylor series around a system of uniform density $n$, denoted by the subscript ''0''.
\begin{align}
 f\klamm{n, \diff{n}{x_i},\difffTwo{n}{x_i}{x_j},\ldots }
=
f_0(n) 
+ \sum_i L_i   \diff{n}{x_i}
+ \sum_{ij} \kappa_{ij}^{(1)} \difffTwo{n}{x_i}{x_j} 
+ \frac{1}{2} \sum_{ij} \kappa_{ij}^{(2)} \diff{n}{x_i} \diff{n}{x_j}
+ \ldots
\end{align}
where the coefficients are obtained by 
\begin{align}
L_i \defi \left.\diff{f}{ \klamm{\diff{n}{x_i}} } \right|_0
\quad,\quad
\kappa_{ij}^{(1)} \defi \left.\diff{f}{\klamm{ \difffTwo{n}{x_i}{x_j}  }} \right|_0
\quad \text{and} \quad
\kappa_{ij}^{(2)} \defi \left.\difffTwo{f}{ \klamm{ \diff{n}{x_i} } }{ \klamm{ \diff{n}{x_j} } } \right|_0.
\end{align}
It is assumed that the system under consideration is isotropic. Hence, it must be invariant under rotation ($x_i \to x_j$) and reflection ($x_i \to -x_i$). The expression for $f$ can thus be simplified to
\begin{align}
 f\klamm{n, \diff{n}{x_i},\difffTwo{n}{x_i}{x_j},\ldots }
=
f_0(n) +
\kappa_1 \nabla^2 n + \kappa_2 (\nabla n)^2 + \ldots,
\end{align} 
where $\kappa_1 \defi \left.\diff{f}{\klamm{\nabla^2 n}} \right|_0$ and $\kappa_2 \defi \difff{f}{(\nabla n)}$. With the help of the divergence theorem, the volume integral over the second term can be transformed into one surface term and one term containing  $(\nabla n)^2$:
\begin{align}
\int_V \kappa_1 \nabla^2 n dV = - \int_V \dif{\kappa_1}{n} (\nabla n)^2 dV + \int_S \kappa_1 \nabla n \cdot N dS.
\end{align}
The volume is chosen such that the density gradient is orthogonal to the normal vector of the surface. Neglecting terms of higher order gives us the {\it square-gradient approximation}:
\begin{align}
f\klamm{n, \nabla n }
=
f_0(n) + \kappa(n) \klamm{\nabla n}^2
\end{align}
Applying the variational principle (\ref{eq:StatMech_MinimalCondition}) we get the fundamental differential equation of Euler:
\begin{align}
f_0'(n\ofR) - \kappa '(n) |\nabla n|^2 - 2\kappa(n) \nabla^2 n - \mu + V_{ext}\ofR \istobe 0.
\end{align}
This is a partial differential equation for the particle density distribution. For further details, see also Evans~\cite[p.157]{Evans}.

\subsection{Approximation for Van-der-Waals-type Approaches \label{sec:DFT_PerturbationTheory}}
Based on Van der Waals approach to introduce attractive cohesion forces, we split the free energy of the system into repulsive\footnote{Zwanzig~\cite{Zwanzig} argued that at high temperatures, the behaviour of a gas is widely characterized by the repulsive part. The same holds for dense fluids.} and attractive contributions. The repulsive reference fluid will be represented by a hard-sphere system. We will present models to approximate such a system in the next section. The attractive contribution to the free energy will be treated in a perurbation approach, which we will present in the following. In particular, we will first split the particle interaction energy $U$ into two contributions:

\begin{align}
U\klamm{{\bf r}^N} = U_{HS}\klamm{{\bf r}^N} + U_{attr}\klamm{{\bf r}^N},
\end{align}
where $U_{HS}$ is the particle interaction energy of a hard-sphere fluid and $U_{attr}$ is the attractive particle interaction energy. Furthermore, we will assume that the particle interaction energy $U\klamm{{\bf r}^N}$ can be written as a sum of pair potentials 
\begin{align}
U\klamm{{\bf r}^N} = \frac{1}{2}\sum_{i\neq j} \phi\klamm{|{\bf r}_i - {\bf r}_j|},  \label{eq:StatMech_DefInteractionEnergy}
\end{align}
where $\phi(r)$ is the $6$-$12$ Lennard-Jones potential defined by:
\begin{align}
\phi\klamm{r} = 4 \varepsilon \klamm{\klamm{\frac{\sigma}{r}}^{12} - \klamm{\frac{\sigma}{r}}^6}.
\label{eq:LennardJonesPotential}
\end{align}
$\varepsilon$ is the depth of the Lennard-Jones-Potential and $\sigma$ defines the distance at which the LJ-potential vanishes. Analogously to the split of the particle interaction energy $U$, we approximate $\phi$ by the sum of a repulsive hard-sphere component and one attractive component:
\begin{align}
\phi(r) \approx \phi_{HS}(r) + \phi_{attr}(r)
\end{align}
A model proposed by Barker and Henderson~\cite{BarkerHenderson2} sets the attractive particle interaction potential to the negative part of the LJ-Potential:
\begin{align}
\phi_{attr} (r) = \left\{\begin{array}{ll}
0 & \text{ if } r \leq \sigma\\
4 \varepsilon\klamm{\klamm{\frac{\sigma}{r}}^{12} - \klamm{\frac{\sigma}{r}}^6} & \text{ if } r > \sigma
\end{array}\right. \label{eq:LennardJones_Attractive_Approx_1}
\end{align}
Another approach has been made by Weeks, Chandler and Andersen~\cite{Weeks} who  split the LJ-potential into repulsive and attractive parts rather than a sum of positive and negative parts:
\begin{align}
\phi_{attr,W} (r) = \left\{\begin{array}{ll}
- \varepsilon & \text{ if } r \leq 2^{1/6}\sigma\\
4 \varepsilon\klamm{\klamm{\frac{\sigma}{r}}^{12} - \klamm{\frac{\sigma}{r}}^6} & \text{ if } r > 2^{1/6}\sigma
\end{array}\right. \label{eq:WeeksPerturbation}
\end{align}

\begin{figure}[ht]
\begin{center}
\includegraphics[width=10cm]{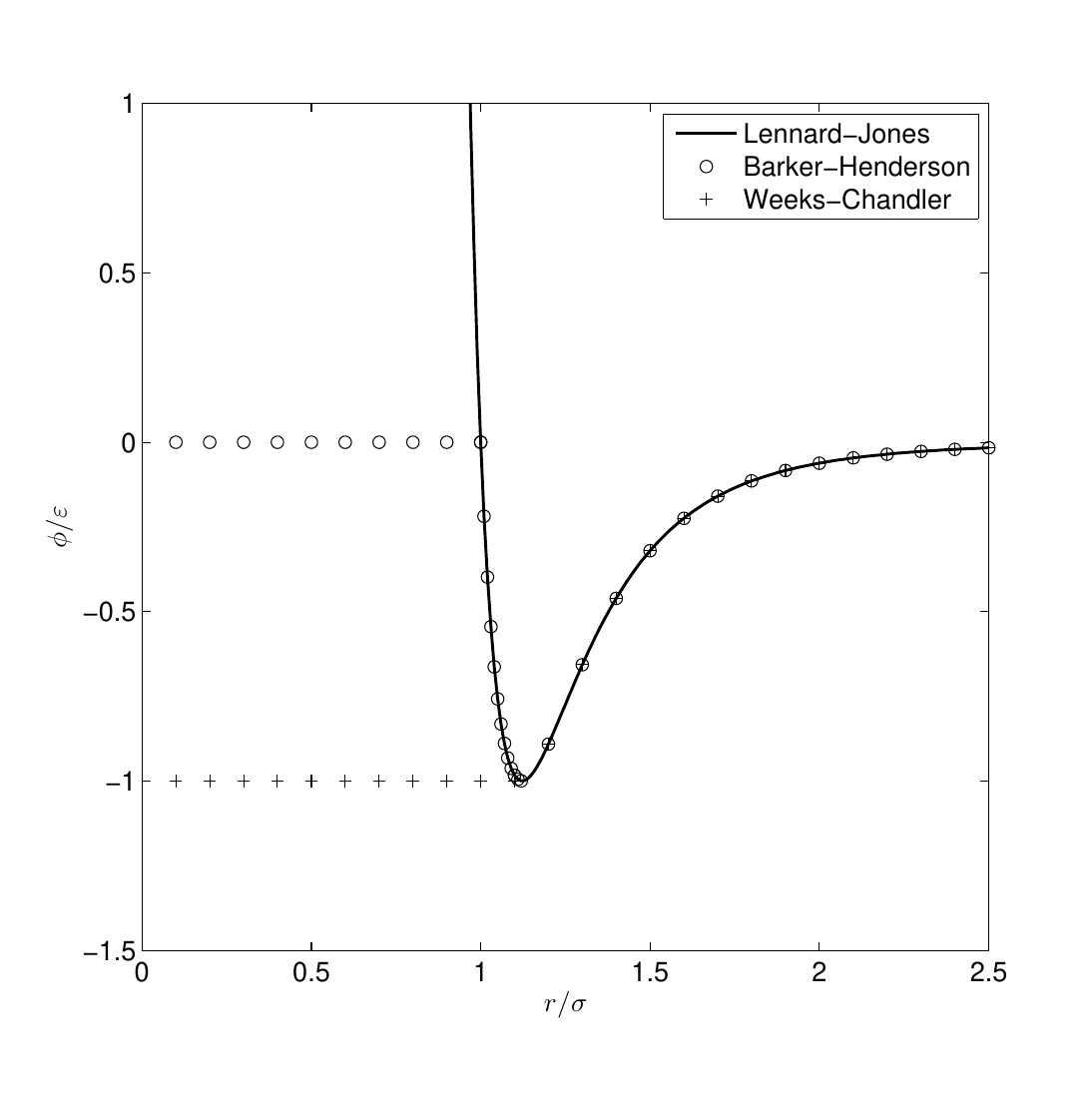}
\caption{Lennard-Jones-Potential $\phi$ as a function of the particle distance $r$. $\sigma$ is the distance at which the LJ-potential vanishes whereas $\varepsilon$ is the depth of the potential.}
\label{fig:WallPotentials1}
\end{center}
\end{figure}

\paragraph{The $\lambda$-Expansion}
 
This expansion is based on an expansion of the rate of change of the free energy with respect to the attractiveness of the fluid. For a detailed review, see also Plischke and Bergersen~\cite{Plischke} or Hansen and McDonald~\cite{Hansen}, whose arguments we sketch in the sequel. An alternative expansion of the free energy around a reference fluid, based on the work from Barker and Henderson in 1967~\cite{BarkerHenderson1}, is presented in Sec. \ref{sec:BarkerHendersonApproach} in the Appendix.

As in the previous section, we say that the particle interaction energy $U$ of the system can be split into one reference part $U_{HS}$ and one perturbative part $U_{attr}$. In order to gradually increase the attractive interaction, we introduce a parameter $\lambda \in [0,1]$ such that:
\begin{align}
U(\lambda) \defi U_{HS} + \lambda \cdot U_{attr},
\end{align}
is the interaction energy of the fluid characterized by $\lambda$. The free energy as well as the canonical partition function $Z_{C}$ of this somewhat imaginary fluid depend on the parameter $\lambda$. 

It is our aim to calculate the derivative of $F$ with respect to $\lambda$. For this, we make use of the statistical mechanical definition of the Helmholtz free energy (\ref{eq:StatMech_BarkerHenderson_FreeEnergyZC}) which yields
\begin{align}
\dif{F}{\lambda} = - \frac{1}{\beta Z_C} \dif{Z_C(\lambda)}{\lambda},
\end{align}
where the canonical partition function $Z_C$ as a function of $\lambda$ is given by
\begin{align}
Z_C\klamm{\lambda} \defi \frac{1}{h^{3N}N!} \iint e^{-\beta H_N(\lambda)} d{\bf p}^N d{\bf r}^N.
\end{align}
Here, $H(\lambda) = E_k + U(\lambda)$ similar to (\ref{eq:StatisticalMech_DefHamiltonian}).
\begin{align}
\dif{F}{\lambda} &= \frac{1}{Z_C(\lambda)} \frac{1}{h^{3N}N!} \iint U_{attr} e^{-\beta H_N(\lambda)}  d{\bf p}^N d{\bf r}^N \\
&=: \avar{U_{attr}}_{\lambda},
\end{align}
where we have taken the average with respect to a fluid with interaction energy $U(\lambda)$. Now, we expand $\diff{F}{\lambda}$ around the reference fluid ($\lambda = 0$):
\begin{align}
\dif{F}{\lambda} = \avar{U_{attr}}_{HS} + \lambda \left.{ \diff{\avar{U_{attr}}_{\lambda}}{\lambda}}\right|_{\lambda = 0} + O\klamm{\lambda^2}.
\end{align}
It can be shown that the second term can be written as
\begin{align}
\left.\diff{\avar{U_{attr}}_{\lambda}}{\lambda}\right|_{\lambda = 0} = - {\beta} \klamm{ \avar{U_{attr}^2}_{HS} - \avar{U_{attr}}_{HS}^2}.
\end{align}
We integrate this rate of change of the Helmholtz free energy from $\lambda=0$ to $1$. As a result, we expect to obtain the difference between the Helmholtz free energy of the fully perturbed fluid and the Helmholtz free energy of the reference hard-sphere fluid:
\begin{align}
F - F_{HS} &=  \int_0^1 \avar{U_{attr}}_{\lambda}d\lambda \notag \\
&= \avar{U_{attr}}_{HS} - \frac{\beta}{2} \klamm{ \avar{U_{attr}^2}_{HS} - \avar{U_{attr}}_{HS}^2} + O(\beta^2),
\label{eq:LambdaEpansion_FreeEnergy}
\end{align}
 The first term of this expansion can be written in terms of the intermolecular energy of pairs of molecules:
\begin{align}
F- F_{HS} &= \avar{U_{attr}}_{HS} + O(\beta)\notag\\
&= \frac{1}{2} \avar{\sum_{i \neq j} \phi_{attr}\klamm{|{\bf r}_i - {\bf r}_j|}}+ O(\beta)\notag\\
&= \frac{1}{2} \iint \phi_{attr}\klamm{|{\bf r} - {\bf r}'|} \avar{  \sum_{i \neq j} \delta\klamm{{\bf r}_i - {\bf r}} \delta \klamm{ {\bf r}_j - {\bf r}' }} d{\bf r}' d{\bf r}+ O(\beta)\notag\\
&= \frac{1}{2} \iint n^{(2)}_{HS}\klamm{{\bf r},{\bf r}'} \phi_{attr}\klamm{|{\bf r}-{\bf r'}|} d{\bf r}' d{\bf r}+ O(\beta),
\label{eq:StatMech_PerturbationHelmholtzFreeEnergy}
\end{align}
where $n^{(2)}_{HS}$ is the two particle distribution of the hard-sphere fluid. In contrast to the expansion introduced in the previous section, Zwanzig showed that the second-order term in (\ref{eq:LambdaEpansion_FreeEnergy}
) includes up to fourth-order correlation functions \cite{Zwanzig}, \cite{Hansen}. In the sequel, we will just consider the expansion up to the first order. 

\section{Models for the Hard Sphere Fluid \label{sec:ModelsHardSphere}}
The perturbative models for the Helmholtz free energy of the fluid are all based on an expansion around a reference hard-sphere fluid. It is thus of essential interest to find simple ways of formulating the two-particle distribution for the hard-sphere fluid as functions of the particle density.

\paragraph{The Hard-Sphere Pair Distribution Function}
One approach to describe the behavior of a hard-sphere fluid is to find approximate expressions for the distribution function $g_{HS}(r)$, defined in (\ref{eq:StatMech_PairDistributionFunction}) for a homogeneous fluid. We define the {\it pair correlation function} as
\begin{align}
h(r) \defi g(r) -1.  \label{eq:PairCorrelationFunction}
\end{align}
In an ideal gas, the particles are completely uncorrelated, hence we get $h(r) = 0$. At large distances, one can assume that the particles are uncorrelated, which leads to $h(r) \to 0$ as $r \to \infty$. In order to find approximate quantities for $g(r)$, we introduce the {\it direct correlation function} $C\klamm{{\bf r}_1,{\bf r}_2}$, defined by the Ornstein-Zernicke-Equation
\begin{align}
h({\bf r}_1,{\bf r}_2) = C\klamm{{\bf r}_1,{\bf r}_2}
+n
\int h({\bf r}_1,{\bf r}_3) C({\bf r}_3,{\bf r}_2) d{\bf r}_3.
\label{eq:StatMech_OrnsteinZernicleEquation}
\end{align}
By definition, the direct correlation function thus excludes effects of three or more particles, which are absorbed in the second term of (\ref{eq:StatMech_OrnsteinZernicleEquation}). In order to find an expression for $h(r)$, a closure is needed. The most famous closure is the Percus-Yevick approximation \cite{Plischke}
\begin{align}
C(r) = \klamm{1 - e^{\beta \phi(r)}} g(r).
\end{align}
It was solved analytically by Wertheim~\cite{Wertheim} for one-component systems (see also Fig.~\ref{fig:PYSol}) and by Lebowitz~\cite{Lebowitz} for mixtures of hard spheres. As a result, one obtains the following equation of state~\cite{Wertheim}:
\begin{align}
\frac{\beta p}{n} = \frac{1+y+y^2}{(1-y)^3}, 
\end{align}
where $y = \frac{\pi}{6} n d^3$ is the packing fraction with the hard-sphere diameter $d$.
\begin{figure}[ht]
\centering
\includegraphics[width=10cm]{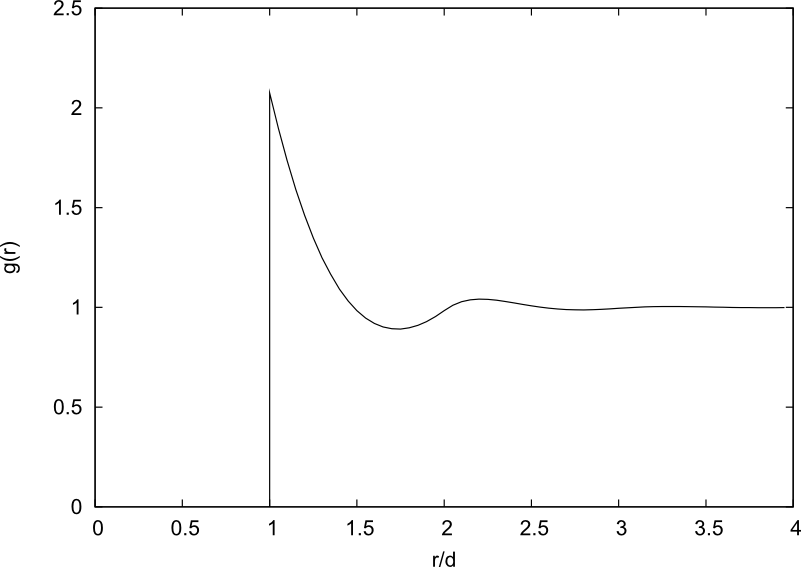}
\caption{Numerical values for the pair distribution function for a hard-sphere fluid of density $n=0.5$ over the distance $r$. The values are the numerical evaluation of the analytical solution of the Percus-Yevick equation by Wertheim~\cite{Wertheim}. The numerical evaluation of the solution was done by Throop and Bearman~\cite{ThroopBearman}.\label{fig:PYSol}}
\end{figure}

\paragraph{Carnahan Starling Approximation for the Helmholtz Free Energy - A Local Density Approximation (LDA)}

Carnahan and Starling~\cite{CarnahanStarling} followed a different approach by approximating the coefficients of a virial expansion by Ree and Hoover~\cite{ReeHoover} with an infinite series. They obtained a slightly modified equation of state
\begin{align}
\frac{\beta p}{n} = \frac{1+y+y^2-y^3}{(1-y)^3}.
\end{align}
Making use of the Helmholtz free energy of an ideal gas, this leads to the following local approximation:
\begin{align}
f_{HS}\klamm{n} = \beta^{-1} \klamm{\ln \klamm{\Lambda^3 n} - 1 + \frac{y\klamm{4-3y}}{\klamm{1-y}^2}},
\end{align}
where $\beta^{-1} \klamm{\ln \klamm{\Lambda^3 n} - 1}$ is the ideal-gas contribution to the Helmholtz free energy:
\begin{align}
f_{id}\klamm{n} = \beta^{-1} \klamm{\ln \klamm{\Lambda^3 n} - 1},\\
\text{and}\qquad
F_{id}[n] = \int f_{id}\klamm{n\ofR} n\ofR d{\bf r}.
\end{align}
For more details on this derivation, see also Hansen and McDonald~\cite{Hansen}. For the total Helmholtz free energy of the fluid, this yields
\begin{align}
F_{HS}[n] = \int f_{HS}\klamm{n\ofR} n\ofR d{\bf r}.
\end{align}


\paragraph{Rosenfeld Fundamental Measure Theory}
The Carnahan-Starling approximation is a local approach, which excludes the possibility of layering in the density profile near a hard wall~\cite[p.849]{TarazonaEvans}. Rosenfeld~\cite{Rosenfeld} derives a non-local approximation which is not a-priori restricted to small nonuniformities. In the homogeneous case, this theory regains the results of the Percus-Yevick-theory. In the inhomogeneous case, considering short-range correlations allows the appearance of oscillatory density profiles. 

Rosenfelds measure theory is based on the assumption that the hard sphere free energy can be written as  a sum of functions of weighted densities $n_\alpha$. 
\begin{align}
F_{HS}[n] &= F_{id}[n] +  \int \Upsilon \klamm{n_\alpha} d{\bf r} \label{eq:FHS:Def}\\
\text{ where } \qquad \Upsilon \klamm{n_\alpha} &\defi 
\underbrace{
- n_0 \ln \klamm{1-n_3} + \frac{n_1n_2}{1-n_3} + \frac{n_2^3}{24\pi \klamm{1-n_3}^2}}
-
\underbrace{
\frac{{\bf n}_1\cdot{\bf n}_2}{1-n_3} - \frac{n_2\klamm{{\bf n}_2 \cdot {\bf n}_2}}{8\pi \klamm{1-n_3}^2}
}.\\
&\qquad\qquad \text{terms for uniform mixture} \qquad\qquad \text{terms for non-uniform mixture}\notag
\end{align}
We note that the sign of the contribution of the vector-weighted densities is positive in Rosenfeld's original papers \cite{Rosenfeld,Rosenfeld1990}, but has been corrected in his subsequent publications~\cite{Rosenfeld1994,Rosenfeld1998}.

In (\ref{eq:FHS:Def}), $F_{id}[n]$ is the ideal gas contribution to the Helmholtz free energy.  $n_i$ are weighted densities:
\begin{align}
n_\alpha\ofR \defi \int n\ofRD w_\alpha \klamm{{\bf r}-{\bf r'}} d{\bf r}'.
\end{align}
The characteristic weight functions $w_\alpha$ are defined as follows:
\begin{align}
w_3\ofR &\defi \Theta\klamm{d/2 - |{\bf r}|} \quad , \quad w_2\ofR \defi \delta\klamm{d/2 - |{\bf r}|}\\
w_1\ofR &\defi \frac{w_2\ofR}{2\pi d} \quad,\quad w_0\ofR \defi \frac{w_2\ofR}{\pi d^2}
\end{align}
where $d$ is the hard-sphere diameter and $\Theta$ is the unit step function $\Theta\klamm{x > 0} = 1$, $\Theta\klamm{x < 0} = 0$. The vector valued weight functions are defined as 
\begin{align}
{\bf w}_1\ofR \defi \frac{{\bf w}_2\ofR}{2 \pi d}
\quad,\quad
{\bf w}_2\ofR \defi \frac{\bf r}{|{\bf r}|} \delta\klamm{ d/2 - |{\bf r}| }.
\end{align}

\paragraph{The Barker-Henderson-Diameter}
Barker and Henderson~\cite{BarkerHenderson2} defined a modified potential function depending on an inverse-steepness parameter and the depth of the potential. The modified potential is defined such that if these parameters are zero, one regains the hard-sphere interaction potential. The Helmholtz free energy of a fluid with the modified potential can then be expanded around zero. Barker and Henderson showed that the first-order term of the inverse-steepness parameter of this expansion vanishes, if the hard-sphere diameter $d$ is chosen as the following temperature-dependent term
\begin{align}
d = \int_0^\sigma \klamm{1-e^{-\beta \phi(r)}} dr,  \label{eq:StatMech_DFT_HS_Diameter}
\end{align}
where $\sigma$ is such that $\phi(\sigma) = 0$.

The evaluation of this integral involves some extra numerical work, as $\phi(r) \to \infty$ as $r \to 0$. Hence, we make use of the approximation $d = \sigma$. Consider that using the hard-sphere potential $\phi(r < d) \defi \infty$ and $\phi(r>d) \defi 0$, and setting $d=\sigma$ satisfies the equation above. This approximation is valid for low temperatures \cite{Tang_BarkerHendersonDiameter}. However, there are more sophisticated models such as the approximation by Cotterman, Schwarz and Prausnitz~\cite{CottermannPrausnitz}:
\begin{align}
d = \frac{1 + 0.2977 \cdot \tilde T}{1 + 0.33163 \cdot \tilde T + 1.047710 \cdot 10^{-3} \cdot \tilde T^2} \sigma, \label{eq:StatMech_CSPApprox}
\end{align}
where $\tilde T = k_B T/\varepsilon$ is the reduced temperature.
Tang~\cite{Tang_BarkerHendersonDiameter} shows that this expression matches perfectly with the Barker Henderson expression (\ref{eq:StatMech_DFT_HS_Diameter}) for the hard-sphere diameter for temperatures $\tilde T<15$. For subcritical temperatures, one gets a correction factor of $0.97\ldots 1.0$, as shown in Fig.~\ref{fig:CottermannDiameter}.

\begin{figure}[ht]
\centering
\includegraphics{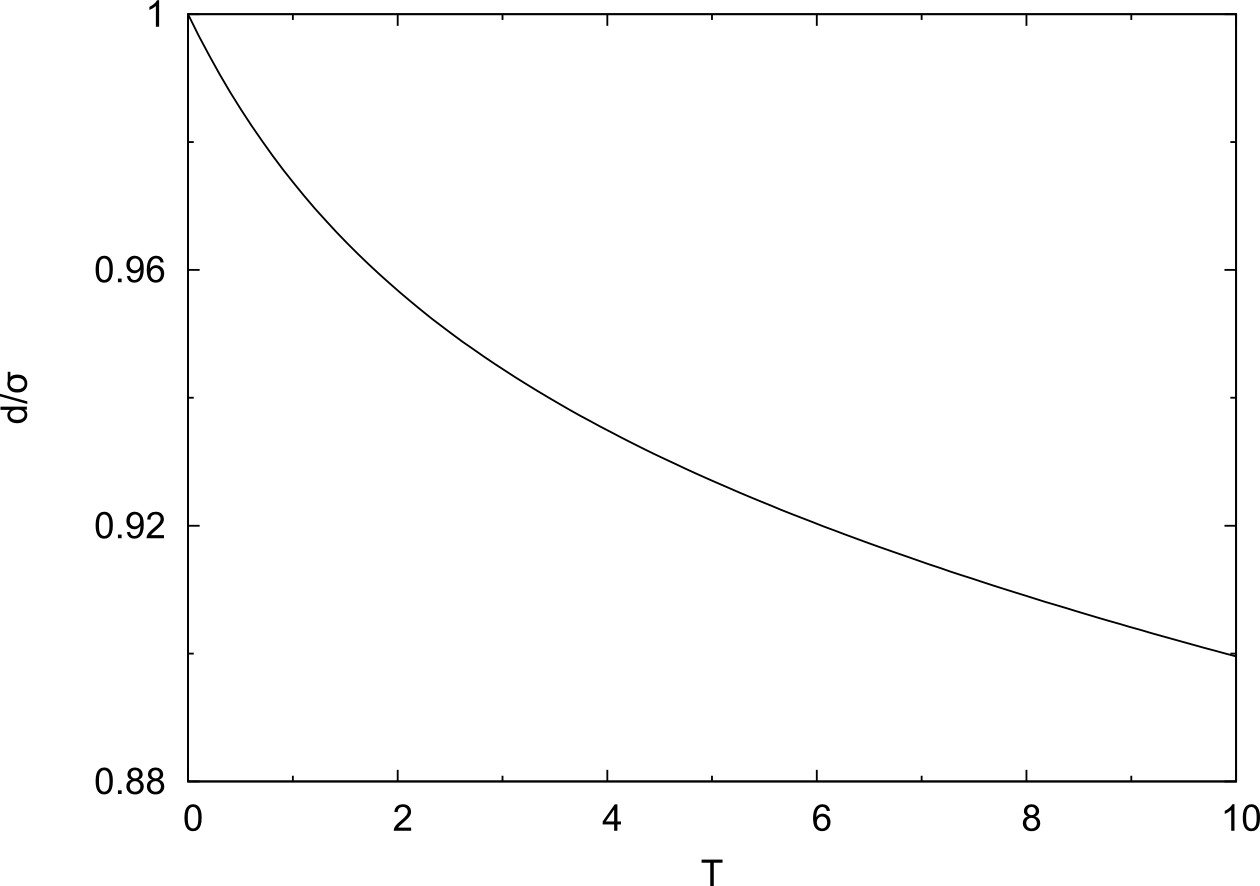}
\caption{Approximation by Cotterman, Schwarz and Prausnitz~\cite{CottermannPrausnitz} for the ratio of hard-core diameter $d$ to soft-core diameter $\sigma$ as in (\ref{eq:StatMech_CSPApprox}). For low temperatures below the critical temperature, the fraction $d/\sigma$ is between $0.97$ and $1$.}
\label{fig:CottermannDiameter}
\end{figure}

\section{The Model used in the Present Work \label{sec:StatMech_OurModel}}

In the present work, we employ a perturbation approach, based on Eq.(\ref{eq:StatMech_PerturbationHelmholtzFreeEnergy}). For the two-particle distribution $n_{HS}^{(2)}$ of the hard-sphere fluid, we neglect all particle-particle correlations by applying a simple {\it Random Phase Approximation}:
\begin{align}
n^{(2)}_{HS}\klamm{{\bf r},{\bf r}'} = n\ofR n\ofRD.
\end{align} 
For the attractive interaction potential $\phi_{attr}(r)$, we employ the Barker-Henderson approach (\ref{eq:LennardJones_Attractive_Approx_1}). Remark that by doing so, the two particle distribution $n_{HS}^{(2)}$ is artificially set to zero for particles with a distance of $|{\bf r}-{\bf r}'| < \sigma$. This is because in (\ref{eq:StatMech_PerturbationHelmholtzFreeEnergy}), $\phi_{attr}(r)$ is multiplied with $n_{HS}^{(2)}$. At the same time, $\phi_{attr}(r) = 0$ for $r < \sigma$. This corresponds to the property of the Percus-Yevick solution as in Fig.~\ref{fig:PYSol}.

The hard-sphere diameter $d$ is set to $\sigma$ for simplicity. Finally, we use the Carnahan-Starling approximation from the class of Local Density approximations (LDA) for the hard-sphere Helmholtz free energy. This leads to the following equation for $\Omega$ as a functional of the particle density:
\begin{align}
\Omega[n] = 
\int f_{HS}(n) n\ofR d{\bf r}+ 
\frac{1}{2} \iint n\ofR n\ofRD \phi_{attr} \klamm{|{\bf r'}-{\bf r}|} d{\bf r}' d{\bf r} + 
\int n\ofR \klamm{V \ofR - \mu }d{\bf r}. \label{eq:OmegaHS}
\end{align}
The non-dimensionalization inherently given in this equation is:
\begin{align}
\tilde n &\defi n\sigma^3 \quad,\quad 
\tilde {\bf r} \defi \frac{\bf r}{\sigma} \quad , \quad 
\tilde T \defi \frac{k_B T}{\varepsilon} \quad,\quad 
\tilde \mu \defi \frac{\mu}{\varepsilon} - \frac{k_B T}{\varepsilon} \ln\klamm{\frac{6\Lambda^3}{\pi d^3}} \quad,\quad
\tilde y = y,\\
\tilde V({\tilde{\bf r}}) &\defi \frac{V\ofR}{\varepsilon} \quad, \quad
\tilde \phi_{attr}(\tilde r) \defi \frac{\phi_{attr}(r)}{\varepsilon} \quad, \quad
\tilde \Omega[\tilde n] \defi \frac{\Omega[n]}{\varepsilon}\quad,\quad
\tilde f_{HS}(\tilde n) \defi \frac{f_{HS}(n)}{\varepsilon} - \frac{k_B T}{\varepsilon} \ln\klamm{\frac{6\Lambda^3}{\pi d^3}}.
\end{align}
In contrast to the non-dimensionalization with the critical temperature, here we do not make use of any additional parameters.
In the sequel, we will omit the tilde and just use dimensionless  variables. Consequently, we get for the fimensionless local hard-sphere free energy
\begin{align}
 f_{HS}\klamm{n} = T \klamm{\ln \klamm{y} - 1 + \frac{y\klamm{4-3y}}{\klamm{1-y}^2}}, \label{eq:DefHardSphereFE_dimless}
\end{align}
where $y = \frac{\pi}{6} n$. The attractive interaction potential (\ref{eq:LennardJones_Attractive_Approx_1}) reduces to 
\begin{align}
\phi_{attr} (r) = \left\{\begin{array}{ll}
0 & \text{ if } r \leq 1\\
4 \klamm{{\frac{1}{r^{12}}} - {\frac{1}{r^6}}} & \text{ if } r > 1
\end{array}\right.. \label{eq:PerturbationPotential_dimless}
\end{align}
The variational equation (\ref{eq:StatMech_MinimalCondition}), where $F[n]$ equals the first two terms in (\ref{eq:OmegaHS}), provides a necessary condition for the minimum of the grand potential $\Omega$:
\begin{align}
\frac{\delta \Omega[n]}{\delta n \klamm{{\bf r}}} = \mu_{HS}\klamm{n\ofR} +
 \int  n\klamm{ {\bf r}'} \phi_{attr}\klamm{| {\bf r}-  {\bf r}'|} d {\bf r}' 
+  V\klamm{ {\bf r}} -  \mu {\overset ! =} 0  \qquad \klamm{\forall {\bf r}} , \label{eq:minCond_dimless}
\end{align}
where the first term is the hard-sphere chemical potential
\begin{align}
\mu_{HS}\klamm{n} \defi \dif{\klamm{nf_{HS}(n)}}{n} = T \klamm{ \ln\klamm{y} + \frac{y\klamm{8-9y+3y^2}}{(1-y)^3}}. \label{eq:Mu_HS_Def}
\end{align}

\section{The Uniform Liquid \label{sec:UniformLiquid}}
In order to set proper boundary conditions, it is important to calculate the particle densities of uniform fluids. The second term of (\ref{eq:OmegaHS}) can be modified such that the integration is done over $\klamm{n\ofR - n\ofRD} \phi_{attr}\klamm{|{\bf r}-{\bf r'}|}$:
\begin{align}
\Omega[n] = 
\int \bar f(n\ofR) n\ofR d{\bf r}+ 
\frac{1}{2} \iint n\ofR \klamm{n\ofRD - n\ofR } \phi_{attr} \klamm{|{\bf r'}-{\bf r}|} d{\bf r}' d{\bf r}
 + \int n\ofR \klamm{V \ofR - \mu } d{\bf r} , \label{eq:Uniform_Omega}
\end{align}
where $\bar f$ defined in (\ref{eq:BarFDef}). This formulation has the advantage of capturing all effects due to non-uniformity of the liquid in the second term, which can be verified by seeing that the second term in Eq. (\ref{eq:Uniform_Omega}) vanishes for a uniform particle density $n\ofR = n$. Consequently, we get a local density approach for uniform liquids in a volume $V$ by $F[n] = {V} n \bar f(n)$, with
\begin{align}
\bar f \klamm{n} &\defi f_{HS}(n) + \alpha n, \label{eq:BarFDef}\\
\text{where} \qquad \alpha &\defi \frac{1}{2}\int \phi_{attr} \klamm{|{\bf r}|} d{\bf r} = 2\pi \int_0^\infty \phi_{attr} \klamm{r} dr=  -\frac{16}{9} \pi. \label{eq:Uniform_def_alpha} 
\end{align}
Here, we made use of the attractive potential as in (\ref{eq:PerturbationPotential_dimless}).
The uniform grand canonical potential per volume element equals the pressure (see also Eq.(\ref{eq:StatMech_GrandCanonicalPotenial})). Inserting (\ref{eq:BarFDef}) into (\ref{eq:Uniform_Omega}) for a fluid with uniform density $n$ and without external potential yields
\begin{align}
{\Omega\klamm{n}}/{{V}} = -p(n) = n f_{HS}\klamm{n} + \alpha n^2 - \mu n. \label{eq:Uniform_Omega_def}
\end{align}
Consider that here, the grand potential $\Omega$ is no longer a functional, but a function of the particle density $n$, which is a scalar in the uniform case. The equilibrium particle densities of a uniform fluid are obtained by solving the minimal condition
\begin{align}
-p'(n) = \mu_{HS}\klamm{n} + 2\alpha n - \mu = 0. \label{eq:Uniform_Omega_prime}
\end{align}
We search for liquid and gas densities such that both phases are in equilibrium. This means that we search for two minima of equal depth of the negative pressure. Hence, we have three equations:
\begin{alignat}{2}
p'(n_l) &= 0 &\qquad &\text{ Variational principle for $n_l$ } \notag\\
p'(n_g) &= 0 &\qquad &\text{ Variational principle for $n_g$}  \label{eq:Uniform_eq}\\
p(n_l) &= p(n_g) &\qquad &\text{ Mechanical equilibrium}
\notag
\end{alignat}
with the three unknowns $n_l,n_g$ and $\mu$ for a given temperature $T$. For the solution $\mu = \mu_{sat}$ of this system of equations, we get two minima of equal depth of the negative pressure as a function of the density (see Fig.~\ref{plot:1COEX}). For a chemical potential $\mu < \mu_{sat}$, the global minimum of $-p(n)$ is at the gas density (see Fig.~\ref{plot:1GAS}). In this case, we say that the gas phase is preferred, or more stable. For $\mu > \mu_{sat}$, the liquid phase is preferred (see Fig.\ref{plot:1LIQ}). In the Appendix, plots of the negative pressure $-p(n)$ and its derivative over $n$ are shown for the temperatures $T = 0.4,0.7,0.9$ and $1.003$. 

\begin{figure}[ht]
\begin{center}
\subfigure[Gas phase, $\mu~=~3.36$]{
\includegraphics[width=4.5cm]{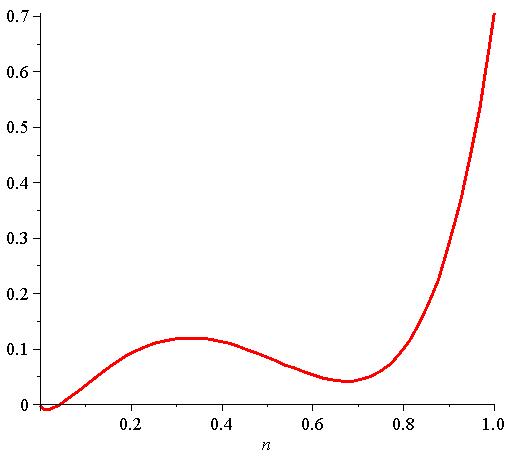}
\label{plot:1GAS}
}
\subfigure[Saturation, $\mu~=~-3.44$]{
\includegraphics[width=4.5cm]{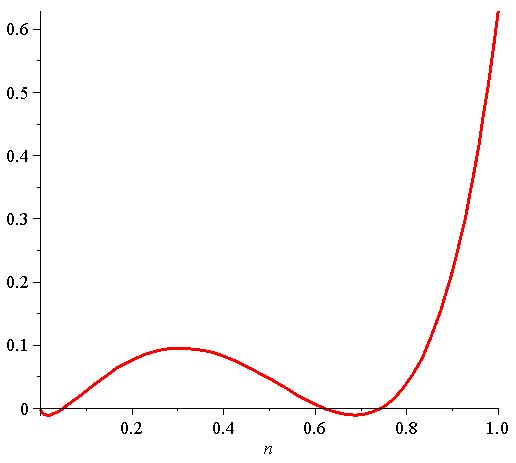}
\label{plot:1COEX}
}
\subfigure[Liquid Phase, $\mu~=~3.52$]{
\includegraphics[width=4.5cm]{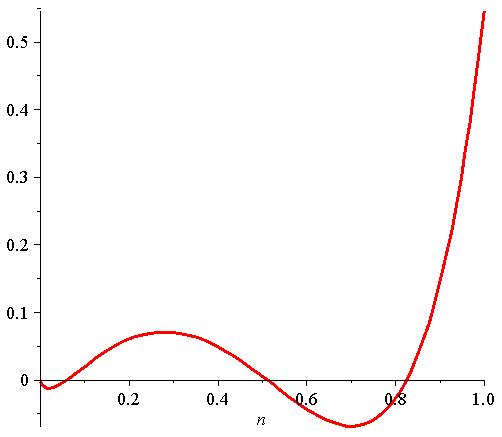}
\label{plot:1LIQ}
}
\caption{Plots of the negative pressure over the density for a uniform fluid for different chemical potentials at $T = 0.7$. The negative pressure corresponds to the grand canonical potential per volume in the uniform case. The preferred density of the system is given by the minimum of the negative pressure. At equilibrium of two phases, there are two minima of equal depth, each one of which indicates the density of one phase (see (b)). For a negative deviation of the chemical potential from the equilibrium value, the gas phase is more stable (see (a)), whereas positive deviations give a more stable liquid phase (see (c)).}
\label{fig:UniformOmegaMu}
\end{center}
\end{figure}
Figure~\ref{fig:PhaseDiagrams} depicts the phase
diagrams for the temperature and pressure as functions of density
and temperature, respectively. These are
contrasted to results of molecular dynamical simulations by
Trokhymchuk and Alejandre~\cite{Trokhymchuk} and experiments by
Michels, Levelt, and De Graeff~\cite{Michels} for argon. There is
qualitative agreement between the model used in this work and the
molecular dynamics simulations/experiments but not a quantitative one.
Indeed, the model seems to underestimate the critical temperature.
The gas density is adequately approximated for the relevant
temperatures close to $0.7$, but the liquid
densities are lower than expected. One can also see that the
saturation pressure is larger compared to the actual one. However,
such deviations appear to be common in DFT/mean field approaches. In
the homogeneous limit they have been analyzed in detail by Tang and
Wu~\cite{TangWu} who pointed out that the deviations are due to
neglecting higher-order correlations in such approaches. However, it
was also shown that in the non-homogeneous case including an attractive wall, the deviations of
density profiles from results of molecular dynamical simulations are
less than expected.

\begin{figure}[ht]
\begin{center}
\subfigure[$\mu_{sat}$ over the reduced temperature $T$]{
\includegraphics[width=7cm]{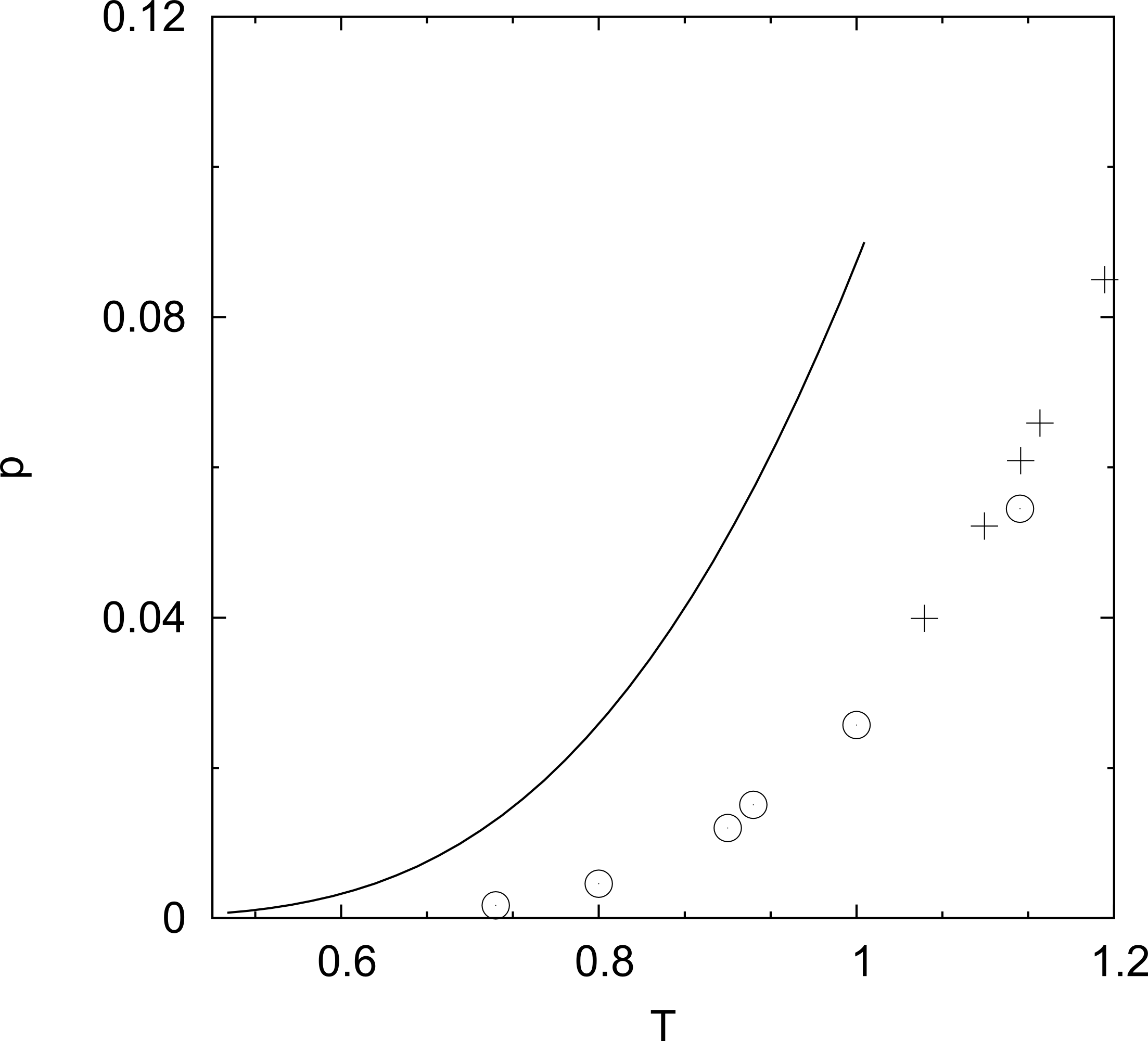}
\label{fig:saturationDataChemPotTemp}
}
\subfigure[Particle densities at $\mu_{sat}$ for liquid and gas]{
\includegraphics[width=7cm]{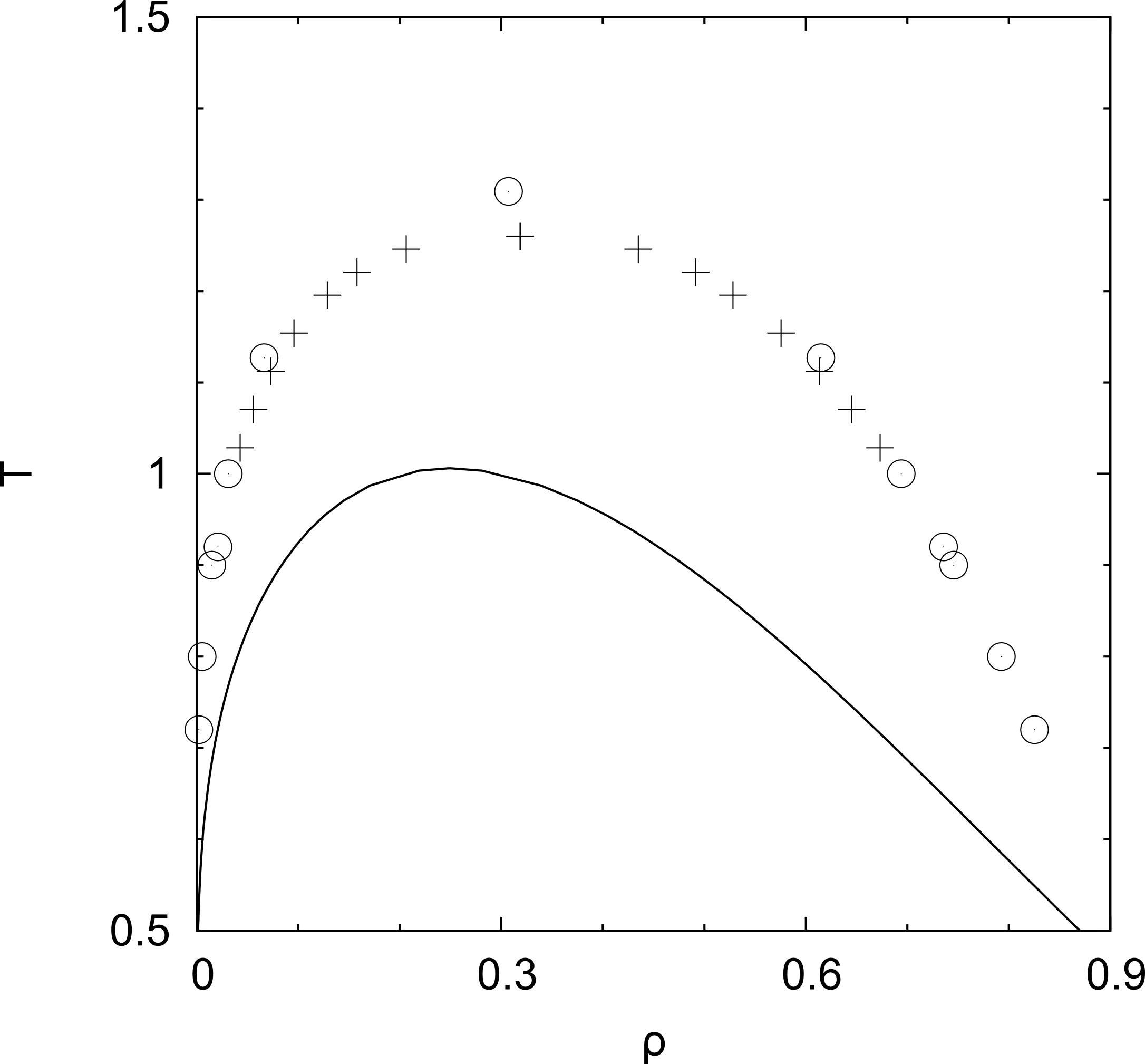}
\label{fig:saturationDataDensityTemp}
}
\caption{Phase diagrams at saturation for the temperature as a function of
density and pressure as a function of temperature. The solid lines are obtained from the model used in this work. Crosses: canonical molecular
dynamics simulations by Trokhymchuk and Alejandre~\cite{Trokhymchuk}; circles: experimental results for argon ($\sigma = 3.405\cdot 10^{-8} \text{cm}$ and $\varepsilon
= 165.3 \cdot 10^{-16} \text{erg}$) by Michels, Levelt, and De Graeff~\cite{Michels}.}
\label{fig:PhaseDiagrams}
\end{center}
\end{figure}

\subsection{The Critical Point}
We now want to calculate the values at the critical point. For this, remark that Eq.~(\ref{eq:Uniform_Omega_prime}) for the equilibrium defines the chemical potential as a function of the density and temperature:
\begin{align}
\mu(n,T) = \mu_{HS}\klamm{n} + 2\alpha n.
\end{align}
Hansen and McDonald~\cite{Hansen} showed that in the critical point, the chemical potential satisfies
\begin{align}
\left.\diff{\mu}{n}\right|_{T=T_c} = \left.\difff{\mu}{n}\right|_{T=T_c} = 0.
\end{align}
These equations can be reduced to
\begin{align}
-1 + 5y_c + 20y_c^2 +4y_c^3-5y_c^4+y_c^5=0,
\end{align} 
where $y_c= \frac{\pi}{6}n_c$ and $n_c$ is the critical density.
This equation has the three real solutions $\klamm{-1.284180231,-0.4245763099,0.1304438842}$. The critical temperature is 
\begin{align}
T_c = - 12 \frac{\alpha}{\pi} \frac{y_c\klamm{1-y_c}^4}{1+4y_c+4y_c^2-4y_c^3 +y_c^4}.
\end{align}
Now choose the only positive solution for $y_c$. This leads to 
\par
\begin{center}
\fbox{
\centering
\parbox{14cm}{
\centering
\begin{align}
n_c = 0.2491294675\quad,\quad
T_c = 1.006172833 \quad,\quad
\mu_c = -3.459392667,
\end{align}}}
\end{center}
which is in excellent agreement with the critical temperature obtained by solving (\ref{eq:Uniform_eq}) numerically (see also Fig.~\ref{fig:saturationDataDensityTemp}). 
In contrast to these results, Monte-Carlo simulations of the Lennard-Jones fluid using mixed-field finite scaling analysis by Caillol~\cite{Caillol} estimate the critical values to be $T_c = 1.326 \pm 0.002$, $n_c = 0.316\pm 0.002$ and $\mu_c = -2.676 \pm 0.005$. We conclude that our model fails to predict the behavior of the fluid close to the critical point. 

\subsection{The Point of Maximal Chemical Potential}

\begin{figure}[ht]
\centering
\includegraphics[height=6cm]{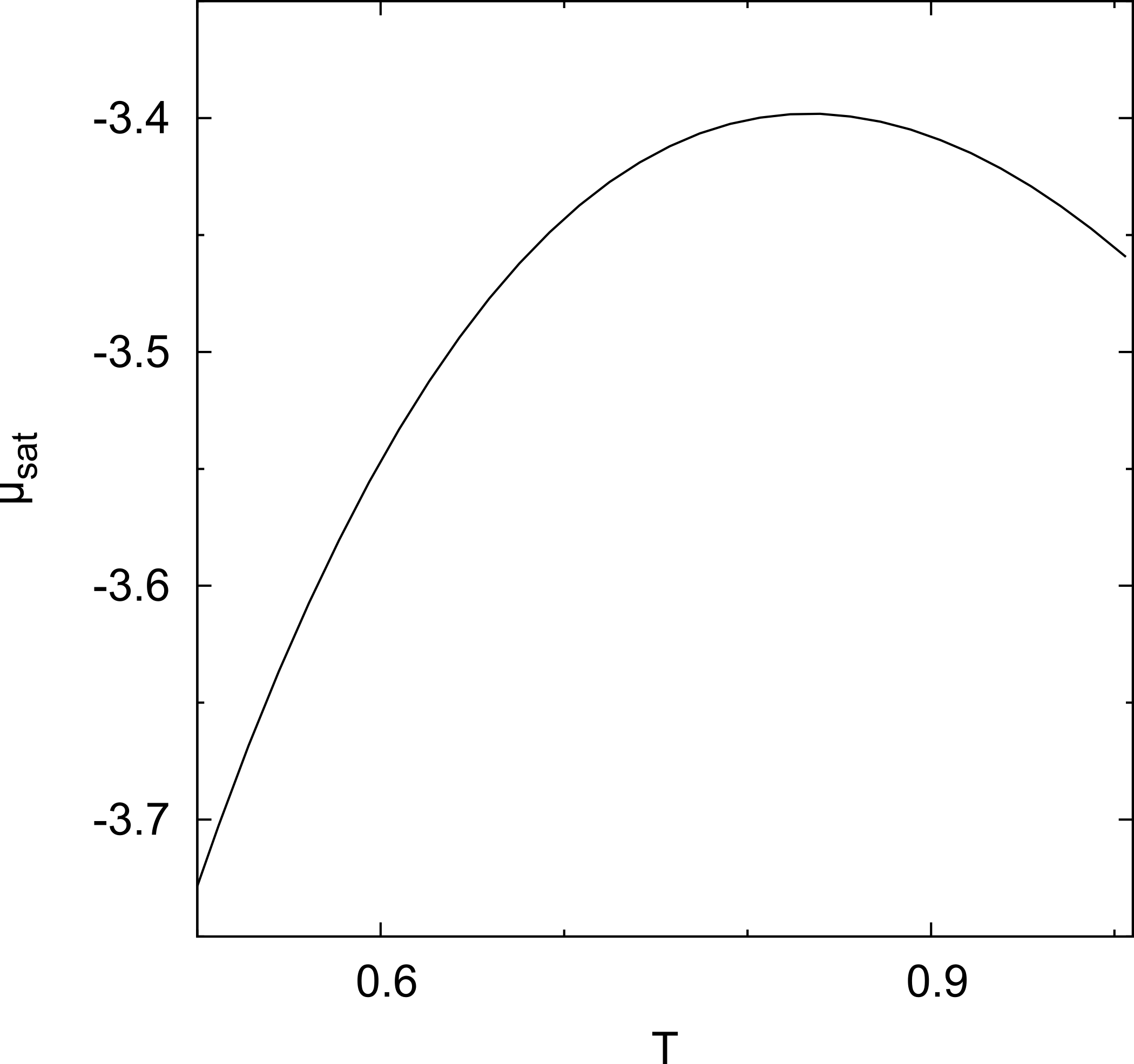}
\caption{
Chemical potential at saturation as a function of the temperature.}
\label{fig:SatPlot_MuT}
\end{figure}

In Fig.~\ref{fig:SatPlot_MuT}, it can be observed that the chemical potential at saturation attains a maximum close to $T= 0.8$. To calculate the values at this extremum, it is necessary to consider the defining equations (\ref{eq:Uniform_eq}). These equations can be rewritten by the function ${\bf k}: \mathbb{R}^4 \to \mathbb{R}^3$,
\begin{align}
{\bf k} \klamm{n_g,n_l,\mu,T} \defi 
\left(
\begin{array}{c}
p'(n_g)\\
p'(n_l)\\
p\klamm{n_l}-p\klamm{n_g}
\end{array}
\right). \label{eq:Uniform_Def_k}
\end{align}
Every configuration at saturation corresponds to one solution of the equation ${\bf k} \klamm{n_g,n_l,\mu,T} = {\bf 0}$. Now, we parametrize the solutions of this equation with the temperature. This means, that the liquid and gas densities as well as the chemical potential at saturation are written as functions of the temperature $T$. The value of ${\bf k}$ is equal to zero on the path  $\klamm{n_g(T),n_l(T),\mu(T),T}$. Hence, $d{\bf k}/dT = 0$, i.e. 
\begin{align}
\dif{{\bf k}}{T} = 
\diff{{\bf k}}{n_l} \dif{n_l}{T}+ 
\diff{{\bf k}}{n_g} \dif{n_g}{T}+ 
\diff{{\bf k}}{\mu} \dif{\mu}{T}+
\diff{{\bf k}}{T} = 0
\end{align}
The point with maximal chemical potential is characterized by $\dif{\mu}{T} = 0$. Setting the respective term in the previous equation to zero and defining the two variables $n_{lT} \defi \dif{n_l}{T}$ and $n_{gT} \defi \dif{n_g}{T}$ leads to the system of six equations
\begin{align}
{\bf k} &= 0\\
\diff{{\bf k}}{n_l} n_{lT}+ 
\diff{{\bf k}}{n_g} n_{gT}+ 
\diff{{\bf k}}{T} &= 0
\end{align}
with the unknowns $\klamm{n_l,n_g,\mu,T,n_{lT},n_{gT}}$. Applying a Newton method on this nonlinear system of equations gives the solution:
\begin{center}
\fbox{
\centering
\parbox{12cm}{
\centering
\begin{alignat*}{3}
n_l &= 0.5570137060& \qquad
n_{lT} &= -1.064612550& \qquad
T &= 0.8334676569 \\
n_g &= 0.05142195929&
n_{gT} &= 0.3804118128&
\mu &= -3.398145808 
\end{alignat*}}}
\end{center}

\section{Surface Tension and the Excess Grand Potential \label{sec:SurfaceTensionExcessGP}}

In this work, we are interested inhomogeneous systems, i.e. in systems with varying density $n\ofR$. In this context, remark that in (\ref{eq:StatMech_GrandCanonicalPotenial}), we introduced the grand potential for homogeneous systems as the negative product of the volume times the pressure of the system. This relation is no longer true for inhomogeneous systems. 

In the presence of a dividing surface between two volumes, the surface itself has a certain contribution to the grand potential. This contribution will be measured by means of the change of the grand potential per unit area of the interface:
\begin{align}
\gamma = \left. \diff{\Omega}{A} \right|_{T,V,\mu},
\label{eq:StatMech_DefSurfaceTension}
\end{align}
which is called the {\it surface tension} or {\it surface energy}. Hence, it becomes necessary to introduce a more general form of Eq.(\ref{eq:StatMech_GrandCanonicalPotenial}): 
\begin{align}
\Omega &= \Omega_{Bulk} + \Omega_{ex} = \Omega_{Bulk} + \gamma A,
\label{eq:DefExcessOmega_Gamma}
\end{align}
where $\Omega_{Bulk}$ is the contribution to the grand potential from the bulk fluid and $\Omega_{ex}$ is the contribution from the interface (see also Plischke \cite[p.164]{Plischke} or Landau \cite[p.455]{Landau}). 

In order to find an expression for $\Omega_{Bulk}$, we divide the three-dimensional space into a partition of three sets: One bulk volume $V_A$, a film $V_f$ and a bulk volume $V_B$ (see also Fig.~\ref{fig:AreasInterface}). The set of possible density distributions $n\ofR$ is restricted by assuming that the density of the fluid in $V_A$ and in $V_B$ is equal to the uniform bulk densities $n_A$ and $n_B$, respectively. Resuming, this leads to
\begin{align}
n\ofR =
\left\{
\begin{array}{lll}
n_A & \text{ if } & {\bf r} \in V_A\\
n_f\ofR & \text{ if } & {\bf r} \in V_f\\
n_B & \text{ if } & {\bf r} \in V_B
\end{array}
\right. . \label{eq:DFTApplication_DensityDistributionForm}
\end{align}

\begin{figure}[ht]
\centering
\includegraphics[width=10cm]{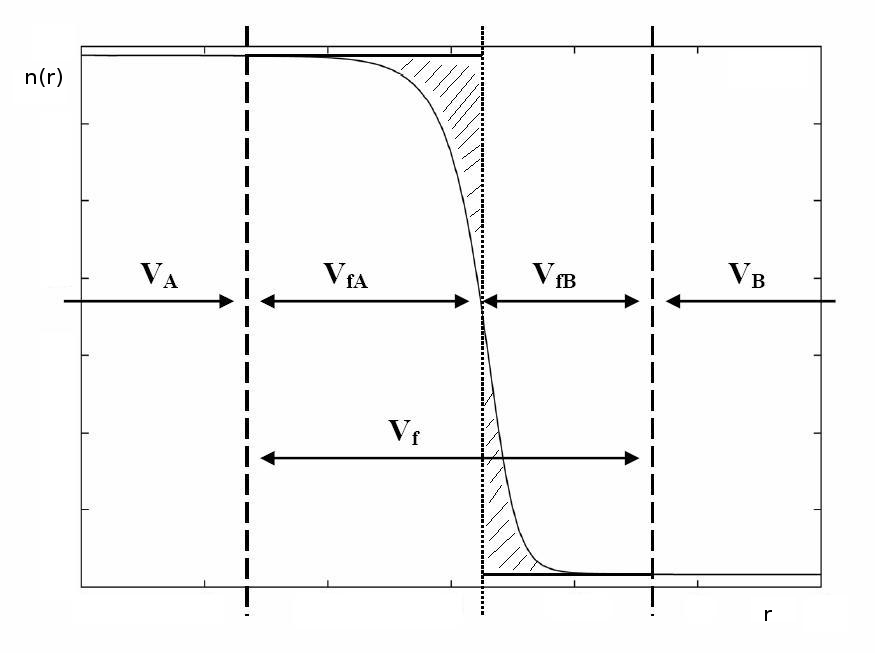}
\caption{An example of a partition of the space for a liquid-gas interface. $V_A$ and $V_B$ are the volumes in which uniform density is assumed. $V_f$ is the volume in which computations are executed. The interface divides the film volume into the two parts $V_{fA}$ and $V_{fB}$. In the case of a liquid-gas interface, this division is uniquely defined by the Gibbs dividing surface, where $V_{fA}$ and $V_{fB}$ are chosen such that the shaded areas are of the same size.}
\label{fig:AreasInterface}
\end{figure}

The liquid-gas interface can be interpreted as the dividing surface between the two volumes $V_A\cup V_{fA}$ and $V_B \cup V_{fB}$ (see also Fig. \ref{fig:AreasInterface}). Then, we define $\Omega_{A}$ as the grand potential of a system with volume $|V_A\cup V_{fA}|$ and with density $n_A$. $\Omega_B$ is defined analogously. Then, the bulk contribution $\Omega_{Bulk}$ to this system can be written as a sum of these two quantities:
\begin{align}
\Omega_{Bulk} = \Omega_A + \Omega_B
\qquad \text{where} \qquad
\begin{array}{l}
\Omega_A = -p\klamm{n_A} |V_A \cup V_{fA}| \\
\Omega_B = -p\klamm{n_B} |V_B \cup V_{fB}| 
\end{array}.
\end{align} 
Note that $\Omega_{Bulk}$ is not the grand potential of a sharp interface system, but the bulk contribution of two volumes of the sizes of $V_A\cup V_{fA}$ and $V_B\cup V_{fB}$, respectively. The difference to a sharp interface system is that long-range forces between the two volumes are not included in $\Omega_{Bulk}$.

However, the volumes $V_{fA}$ and $V_{fB}$ are not uniquely defined, as the density profile is smooth. The choice of the volumes can be restricted by imposing the condition that $\Omega_{ex}$ corresponds to the excess free energy per unit area. We get:
\begin{align}
\Omega_{ex} = \gamma A &= \Omega - \Omega_A - \Omega_B \label{eq:StatMech_SurfaceTensionDef} \\
&= \klamm{F - F_A - F_B} - \mu \klamm{N - N_A - N_B} \notag\\
 &= F_S - \mu N_S, \notag
\end{align}
where $F_S$ is the surface excess Helmholtz free energy and $N_S$ is the surface excess number of particles. Here, we have used that the grand potential can be written as $\Omega = F - \mu N$. In order to obtain $\Omega_{ex} = F_S$, the excess surface number of particles $N_S$ has to vanish. This yields:
\begin{align}
N &= N_A + N_B\\
\int n\ofR d{\bf r} &= n_A |V_A \cup V_{fA}| +  n_B |V_B \cup V_{fB}|\\
\Leftrightarrow \qquad \int_{V_{fA}} \klamm{n_A - n\ofR} d{\bf r} &=  \int_{V_{fB}}  \klamm{ n\ofR - n_B } d{\bf r},
\end{align}
which is the defining equation for the position of the surface. 

In order to link the grand potential of the system as defined in (\ref{eq:Uniform_Omega}) with the surface energy, we insert  (\ref{eq:Uniform_Omega_def}) in (\ref{eq:Uniform_Omega}) and get
\begin{align}
\Omega[n] = 
- \int p(n\ofR) d{\bf r}+ 
\frac{1}{2} \iint n\ofR \klamm{n\ofRD - n\ofR } \phi_{attr} \klamm{|{\bf r'}-{\bf r}|} d{\bf r}' d{\bf r}
 + \int n\ofR V \ofR  d{\bf r} .
\end{align}
The value of the grand potential as defined above depends on the size of the bulk volumes $V_A$ and $V_B$. In order to remove this inconsistency, we consider instead the excess grand potential defined by (\ref{eq:DefExcessOmega_Gamma})
\begin{align}
\Omega_{ex}[n] = \gamma A = 
&\Omega[n] - |V_{A} \cup V_{fA}| \klamm{ p\klamm{n_A}} - |V_{fB} \cup V_B| \klamm{ p\klamm{n_B}} \notag\\
= &- \int_{V_A \cup V_{fA}} \klamm{p(n\ofR)-p\klamm{n_A}} d{\bf r} 
- \int_{V_{fB} \cup V_B} \klamm{p(n\ofR)-p\klamm{n_B}} d{\bf r}+ \notag\\
&+\frac{1}{2} \iint n\ofR \klamm{n\ofRD - n\ofR } \phi_{attr} \klamm{|{\bf r'}-{\bf r}|} d{\bf r}' d{\bf r}
 + \int n\ofR V \ofR  d{\bf r}.
\label{eq:Def_ExcessGrandPotential}
\end{align}
Remark that the condition $N_S = 0$ only makes sense if a liquid-gas interface is considered. In the case of a solid substrate in contact with a bulk gas phase, it is more useful to use the surface of the substrate as a natural division between the two phases. The excess number of particles is then used as a measure for the amount of adsorbed liquid on the substrate. 

\subsection{Surface Tension of Droplets}
In this work, we study surface effects on planar and on spherical substrates. With respect to a spherical geometry, there have been recent studies which are of some interest for the interpretation of the results of this work and which we summarize briefly. 

In particular, the surface tension of a liquid drop in a bulk gas phase has been studied exceedingly during the past decades. The well-known Kelvin relation 
\begin{align}
\Delta p = \frac{2\gamma_{lg,R}}{R}
\end{align}
establishes a relation between the surface tension and the pressure difference of the liquid inside a droplet and its vapor outside the droplet. The radius $R$ at which this relation holds exactly defines the so-called surface of tension~\cite{Tolman}. However, the surface tension does also depend on the droplet size. Based on the thermodynamical Gibbs relation $d\gamma = -\Gamma d\mu$, where $\Gamma$ is the excess density related to the surface of tension, Tolman showed that the surface tension of a droplet with large radius $R$ depends in the first approximation on the distance $\bar\delta$ between the surface of tension $R_s$ and the Gibbs dividing surface $R_e$, i.e. $\bar \delta = R_e - R_s$. Particularly, Tolman showed that
\begin{align}
\frac{\gamma_{lg,R}}{\gamma_{lg,\infty}} = \frac{1}{1 + 2\bar\delta/R} = 1 - 2 \frac{\bar \delta}{R} + {O}\klamm{\klamm{\frac{\bar\delta}{R}}^2}, \label{eq:SurfaceTensionSpherical_Tolman}
\end{align} 
where $\gamma_{lg,\infty}$ is the surface tension of the plane interface. The parameter $\bar \delta$ is also known as the Tolman-length and was said to be of the range of $0.25$ to $0.6$ times the intermolecular distance of the liquid phase. However, Tolman restricted the validity of formula (\ref{eq:SurfaceTensionSpherical_Tolman}) to droplets of appropriately large sizes. For very small droplets, Tolman expected $\bar \delta$ to depend on the droplet radius $R$. He also questions the method of derivation based on thermodynamic methods~\cite{Tolman}. 

In recent molecular dynamical simulations, Sampayo \emph{et.
al.}~\cite{Sampayo} provide evidence that the macroscopic theory can only be used for droplets with a radius greater than ten times the molecular diameter. For smaller droplets, the effect of the second order energy fluctuation is not negligible in the expansion (\ref{eq:LambdaEpansion_FreeEnergy}) of the free energy. Indeed, for simulations with droplet sizes of five times the molecule diameter, the second-order fluctuation term is of the same order of magnitude as the first order term, with inverse sign \cite{Sampayo}. Consider that in the derivation of the minimal condition (\ref{eq:minCond_dimless}) used in this work, all second-order terms were neglected. Hence, in the sequel calculations for the spherical case are restricted to cases where the liquid-vapor interface appears at radius larger than ten molecule diameters.

\chapter{Density Profiles of Thin Films on a Solid Substrate \label{sec:DensityProfiles}}

We consider solid substrates which are in contact with a fluid. The interaction between a particle of the fluid and a particle of the substrate is described by a Lennard-Jones-Potential: 
\begin{align}
\phi_w\klamm{r} \defi 4 \varepsilon_w \klamm{\klamm{\frac{\sigma_w}{r}}^{12} - \klamm{\frac{\sigma_w}{r}}^6}
\label{eq:Wall_LennardJonesPotential}
\end{align}
with the two supplementary parameters $\sigma_w$ and $\varepsilon_w$, which describe the distance at which the potential vanishes and the depth of the potential, respectively. The external potential induced by the wall at a certain point in the fluid corresponds to the integrated LJ-potential over the wall $W$ times the density $n_w$ of the wall particles. For simplicity, $n_w$ is absorbed in the parameter $\varepsilon_w$ of the LJ-potential such that we obtain the general external potential
\begin{align}
V\klamm{\bf r} \defi \int_{W} \phi_w\klamm{|{\bf r} - {\bf r}'|} d{\bf r}'.
\label{eq:DenProf_GenDefWallPot}
\end{align}

\section{Numerical Methods for One-dimensional Geometries \label{sec:DenProf_Numerics}}

We consider one dimensional geometries which allow a reduction of the minimal condition (\ref{eq:minCond_dimless}) to the following expression 
\begin{align}
\mu_{HS}\klamm{n(z)} +
 \int  n\klamm{ z'} \Phi\klamm{z,z'} dz' 
+  V\klamm{z} -  \mu {\overset ! =} 0  \qquad \klamm{\forall z},
\label{eq:DenProf_1DMinCond}
\end{align}
where exact expressions for $\Phi\klamm{z,z'}$ and $V\klamm{z}$ will be defined with respect to the given geometries. We restrict the computations to an interval $[z_0,z_N]$. This domain is discretized into $N$ intervals of equal length $\Delta z \defi (z_N-z_0) / N$. The integral term is discretized with a trapezoidal rule. It is assumed that to the left of the domain, i.e. for $z < z_0$, the density of the fluid is $n_-$. For $z > z_N$, the density is assumed to be $n_+$. Hence, the above condition transforms to 
\begin{align}
g_i(n_0,\ldots,n_{N})
\defi& 
\mu_{HS}\klamm{n_i}+  n_- \Psi_-\klamm{z_i}+ n_+ \Psi_+\klamm{z_i} + V(z_i) - \mu +\notag\\
&+\frac{\Delta z}{2}\klamm{ n_0 \Phi\klamm{z_0,z_i} + 2\sum_{j=1}^{N-1} n_j \Phi\klamm{z_i,z_j} + n_N \Phi\klamm{z_N,z_i}} \istobe 0 \qquad
\klamm{\forall i = 0,\ldots,N},  \label{eq:g_Discretized}
\end{align}
where
\begin{align}
\Psi_-(z) \defi \int_{z<z_0}  \Phi\klamm{z,z'} dz' \qquad
\text{and} \qquad
\Psi_+(z) \defi \int_{z>z_N}  \Phi\klamm{z,z'} dz' 
\label{eq:DenProf_PsiDef}
\end{align}
are expressions for the influence of the boundary conditions on one particle in the fluid. (\ref{eq:g_Discretized}) gives $N+1$ nonlinear equations for the unknown densities ${\bf n}\defi\klamm{n_0,\ldots n_{N}}^T$. The Jacobian of ${\bf g}\defi\klamm{g_0,\ldots g_{N}}^T$ is 
\begin{align}
\klamm{\bf J}_{ij} = \diff{g_i}{n_j}
=& \delta_{ij} \mu_{HS}'(n_i) +  \frac{\Delta z}{2} \klamm{2 - \delta_{j0} - \delta_{jN}}  \Phi(z_i,z_j),
\label{eq:DenProf_Jacobi_Def}
\end{align}
where $\delta_{ij}$ is the Kronecker-Delta and the derivative of $\mu_{HS}$ is given by
\begin{align*}
\mu_{HS}'(n) = T\klamm{\frac{1}{y} + 2 \frac{4 - y}{(1-y)^4}}\frac{\pi}{6}
\end{align*}
with $y = \frac{\pi}{6} n$.

\paragraph{The Newton Method}
In order to solve (\ref{eq:g_Discretized}), a Newton scheme can be applied. In each iteration $k$, the linear system of equations 
\begin{align}
{\bf J} \cdot \Delta {\bf n}^{k} = - {\bf g}\klamm{\bf n^k}
\label{eq:DenProf_NewtonStep}
\end{align}
is solved using a LU decomposition method, where ${\bf J}$ is computed with respect to ${\bf n}^k$. When solving (\ref{eq:DenProf_NewtonStep}), one can make use of the structure of ${\bf J}$. The interaction potential $\Phi$ in (\ref{eq:DenProf_Jacobi_Def}) shows a fast decay with increasing distance of the diagonal (see also Sec.\ref{sec:Planar_AnalyticalExpr} and Sec.~\ref{sec:Sphere_AnalyticalExpr}). Hence, we set all off-diagonal elements to zero for which $|z_i-z_j| > 5$, where the cutoff of $5$ is an adjustable parameter. Doing this, we obtain a sparse system of linear equations, which is considerably easier to solve than the full system. 

Additionally, the singularities of ${\bf g}$ in zero and $6/\pi$ due to the hard sphere chemical potential $\mu_{HS}(n)$ (see Eq.(\ref{eq:Mu_HS_Def})) require some extra attention. They are accounted for by rescaling the vector $\Delta {\bf n}^{k}$ such that
\begin{align*}
{\bf n}^{k+1} = {\bf n}^k + \lambda \Delta {\bf n}^{k},
\end{align*}
is bounded to $[0,6/\pi]$, where $\lambda\in (0,1)$. The rescaling is done using the following short algorithm:
\begin{alltt}
\(\lambda\)=1.;
for(i=0;i<N;i++)\{
   if( \( \nk\)[i] + \( \Deltan\)[i] <= 0)
      \(\lambda\) = min( -\(\nk\)[i]/\(\Deltan\)[i], \(\lambda\));
   if(  \( \nk\)[i] + \( \Deltan\)[i] >= 6/\(\pi\) )
      \(\lambda\) = min( (6/\(\pi\) - \( \nk\)[i])/\(\Deltan\)[i] , \(\lambda\));
\}
if(\(\lambda\) < 1.)
   \(\lambda\) = \(\lambda\) * (1.-1.e-14);	

for(i=0;i<N;i++)
   \( \nkpe\)[i] = \( \nk\)[i] + \(\lambda\)*\( \Deltan\)[i];
\end{alltt}

\paragraph{Modified Newton Method}
The main difficulty when using the Newton method is that in each iteration the full linear system of equations (\ref{eq:DenProf_NewtonStep}) has to be solved for $N+1$ variables. One way to avoid this costly computation is to make use of the structure of the Jacobi matrix ${\bf J}$. In (\ref{eq:DenProf_Jacobi_Def}) it can be seen that the diagonal term of ${\bf J}$ is always the greatest term. This is because $\mu_{HS}'(n)$ can be assumed to be greater or equal one, whereas the grid size $\Delta z$ is usually smaller than $0.1$. In the modified Newton method, the Jacobian is approximated by its diagonal. In this case, the iteration simplifies to:
\begin{align*}
\Delta {n}^{k+1}_i &= -  \frac{{g_i}\klamm{\bf n^k}}{\tilde g_i '}\\
\text{ where }\qquad \tilde g_i ' &=\mu_{HS}'(n_i) + \Delta z \Phi_{\text{Pla}}(0)
\end{align*}
and $g_i$ as defined in (\ref{eq:g_Discretized}). Again, the singularities of $\mu_{HS}'(n_i)$ are considered using the algorithm described above.

\section{Density Profiles for a Thin Film on a Planar Wall \label{sec:DenProf_PlanarWall}}
A planar wall $W = \mathbb{R}^2\times \mathbb{R}^-$ suggests symmetry in the two directions parallel to the wall. We set the density $n(z)$ for negative $z$ to zero. Hence, the minimal condition (\ref{eq:DenProf_1DMinCond}) can be written as
\begin{align}
\mu_{HS}\klamm{n(z)} +
 \int_0^\infty  n\klamm{ z'} \Phi_{\text{Pla}}\klamm{|z - z'|} dz' 
+  V_{\text{Pla}}\klamm{z} -  \mu {\overset ! =} 0  \qquad \klamm{\forall z > 0}.
\label{eq:DenProf_MinCond_PlanarWall}
\end{align}
$\Phi_{\text{Pla}}(z)$ is the attractive interaction potential between a point in the fluid and a plane at distance $z$. For completeness, we also give an expression for the excess grand potential (\ref{eq:Def_ExcessGrandPotential}) per unit area. In the case of a liquid-gas interface, this yields
\begin{align}
\gamma_{lg,\infty}[n(\cdot)] \defi& - \int_{-\infty}^{z_G} \klamm{ p(n(z)) - p(n_l)} dz
- \int_{z_G}^\infty \klamm{ p(n(z)) - p(n_g)} dz+\notag\\
&+ \frac{1}{2} \int_0^\infty \int_{-\infty}^\infty n(z) \klamm{ n(z') - n(z) } \Phi_{\text{Pla}}(|z-z'|) dz' dz,
\label{eq:DenProf_LGSurfaceTension}
\end{align}
where $z_G$ is the position of the Gibbs dividing surface. $\gamma_{lg,\infty}$ is also known as the surface tension of the interface. In the case of a planar wall, the surface of the wall at $z=0$ is used as a natural division between the gas bulk phase and the substrate (see also Sec.~\ref{sec:SurfaceTensionExcessGP}). This yields:
\begin{align*}
\gamma_{wall,\infty}[n(\cdot)] \defi& - \int_0^\infty \klamm{ p(n(z)) - p(n_g)} dz
+ \frac{1}{2} \int_0^\infty \int_{-\infty}^\infty n(z) \klamm{ n(z') - n(z) } \Phi_{\text{Pla}}(|z-z'|) dz' dz +\ldots\\
&+ \int_0^\infty n(z) V_{\text{Pla}}(z) dz
\end{align*}
$\gamma_{wall,\infty}$ is also referred to as the surface energy of the wall.

\subsection{Analytical Expressions\label{sec:Planar_AnalyticalExpr}}
\paragraph{Interaction Potential}

Integrating the attractive interaction potential between one point and all points of a plane at distance $z$ leads to the exact expression for $\Phi_{\text{Pla}}$:
\begin{align}
\Phi_{\text{Pla}}\klamm{z} :&= \iint \phi_{attr}\klamm{\sqrt{x'^2+y'^2+z^2}} dx' dy'  \notag\\
&= 2\pi \int_0^{\infty} \phi_{attr}\klamm{\sqrt{r^2 + z^2}} rdr \notag\\
&{\overset {R = \sqrt{r^2+z^2}} =} 2\pi \int_z^\infty \phi_{attr} \klamm{R}  R dR\notag\\
&
{\overset{(\ref{eq:PerturbationPotential_dimless})}=}
 \left\{
\begin{array}{ll}
-\frac{6}{5}\pi & \text{ if } z < 1\\
4\pi\klamm{\frac{1}{5 z^{10}}-\frac{1}{2z^4}} & \text{ if } z > 1
\end{array}
\right. .\label{eq:Def_Phi}
\end{align}

\paragraph{Influence of Boundary Conditions}
The definition of the interaction potential (\ref{eq:Def_Phi}) in the planar case, together with (\ref{eq:DenProf_PsiDef}), leads to 
\begin{align*}
\Psi_-(z) = \Psi_{\text{Pla}}\klamm{|z-z_0|} 
\qquad \text{and}\qquad
\Psi_+(z) = \Psi_{\text{Pla}}\klamm{|z-z_N|},
\end{align*}
where 
\begin{align}
\Psi_{\text{Pla}}(z) \defi&  \int_z^\infty \Phi_{\text{Pla}}(z') dz' \notag\\
=& \left\{
\begin{array}{ll}
4\pi  \klamm{{\frac{1}{45z^9}}-{\frac{1}{6z^3}}} & \text{ for } z > 1\\
-\frac{16}{9} \pi + \frac{6}{5} \pi z & \text{ for } z \leq 1\\
\end{array} \right.
\label{eq:1D_Numerics_Vn}
\end{align}

\paragraph{The Wall Potential}
Following the general definition (\ref{eq:DenProf_GenDefWallPot}) for a wall potential, we get:
\begin{align}
V_{\text{Pla}}(z) :&= \int_{z' < 0} \phi_w\klamm{
\sqrt{ x'^2 + y'^2 + (z-z')^2 }
} dx' dy' dz' \notag\\
& =
4{\pi \varepsilon_w \sigma_w^3 }
\klamm{{ \frac{1}{45} \klamm{\frac{\sigma_w}{z}}^9}- 
\frac{1}{6}\klamm{\frac{\sigma_w}{z}}^3}.\label{eq:WallPotential}
\end{align}
The wall is assumed to be non-penetrable, i.e. $V(z) = \infty$ for negative $z$. 

\paragraph{The Adsorption}
The excess number of particles per unit area is also referred to as adsorption $\Gamma$. It is defined by:
\begin{align}
\Gamma[n(\cdot)] \defi \int_0^\infty \klamm{n(z) - n_g} dz.
\label{eq:PlanarAdsorption}
\end{align}

\subsection{Numerical Results \label{sec:DensityProfiles_PlanarWall}}

\paragraph{The Liquid-Gas Interface}
We compute the density profile of a liquid-gas interface at equilibrium chemical potential. For this, we set $n_- = n_l$ and $n_+ = n_g$ in Eq.(\ref{eq:g_Discretized}), whereas the external potential is set to zero. In Fig.\ref{fig:LiquidGasInterface_Profiles}, we show density profiles at different temperatures. At high temperatures close to the critical point, the profiles are very smooth. For low temperatures, the density profiles become steeper.

\begin{figure}[ht]
\centering
\includegraphics{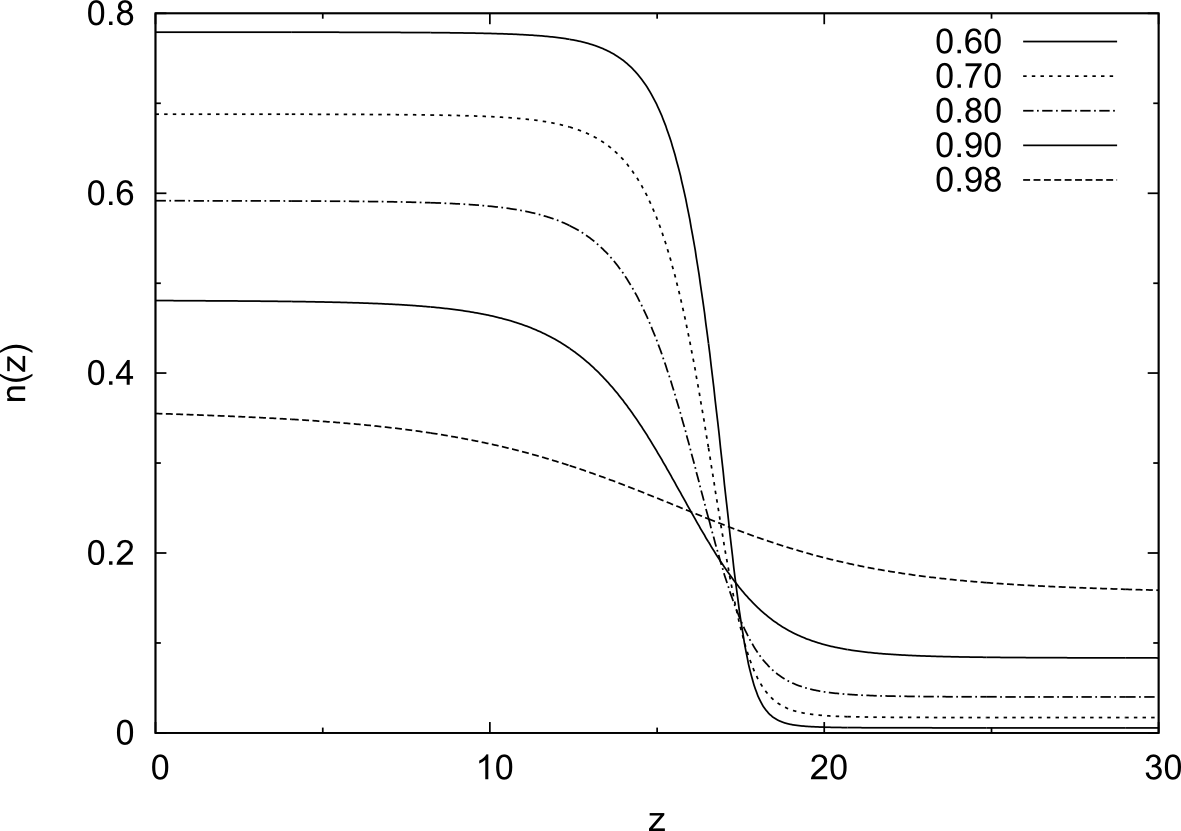}
\caption{Plots of density profiles of liquid-gas interfaces at different temperatures.}
\label{fig:LiquidGasInterface_Profiles}
\end{figure}

We compute the liquid-gas surface tension by means of (\ref{eq:DenProf_LGSurfaceTension}). In Fig. \ref{fig:SurfaceTensionData}, these results are compared to experiments, molecular dynamical simulations and
other DFT computations. Here, the surface tension is plotted against $T/T_c$. This leads to a good agreement, as errors in the critical region are avoided~\cite{Toxvaerd}.

\begin{figure}[pht]
\centering
\includegraphics[width=10cm]{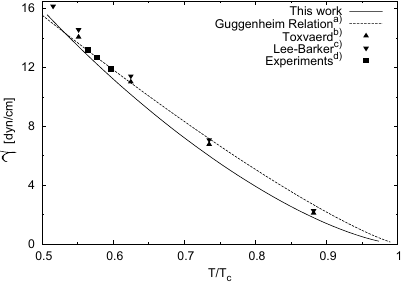}
\caption{Plots of surface tension as a function of dimensionless temperature
$8T/T_c$. Solid line: our model; dashed line: fit to
equation $\gamma(T) = \gamma_0 \klamm{1- {T}/{T_c}}^{1 + r}$ by
Guggenheim~\cite{Guggenheim}. The resulting coefficients are
$\gamma_0 = 36.31 \text{dyn}/\text{cm}$ and $r = \frac{2}{9}$;
triangles (up): computational results by Toxvaerd for a 12-6 LJ fluid
using the Barker-Henderson perturbation theory~\cite{BarkerHenderson2} with the Percus-Yevick solution~\cite{ThroopBearman} for the hard-sphere reference fluid and using the exact
hard sphere diameter~\cite{Toxvaerd}; triangles (down): Monte Carlo
simulations by Lee and Barker~\cite{LeeBarker}; squares:
experimental results by Guggenheim~\cite{Guggenheim}}
\label{fig:SurfaceTensionData}
\end{figure}

\begin{figure}[ht]
\centering
\includegraphics[width=10cm]{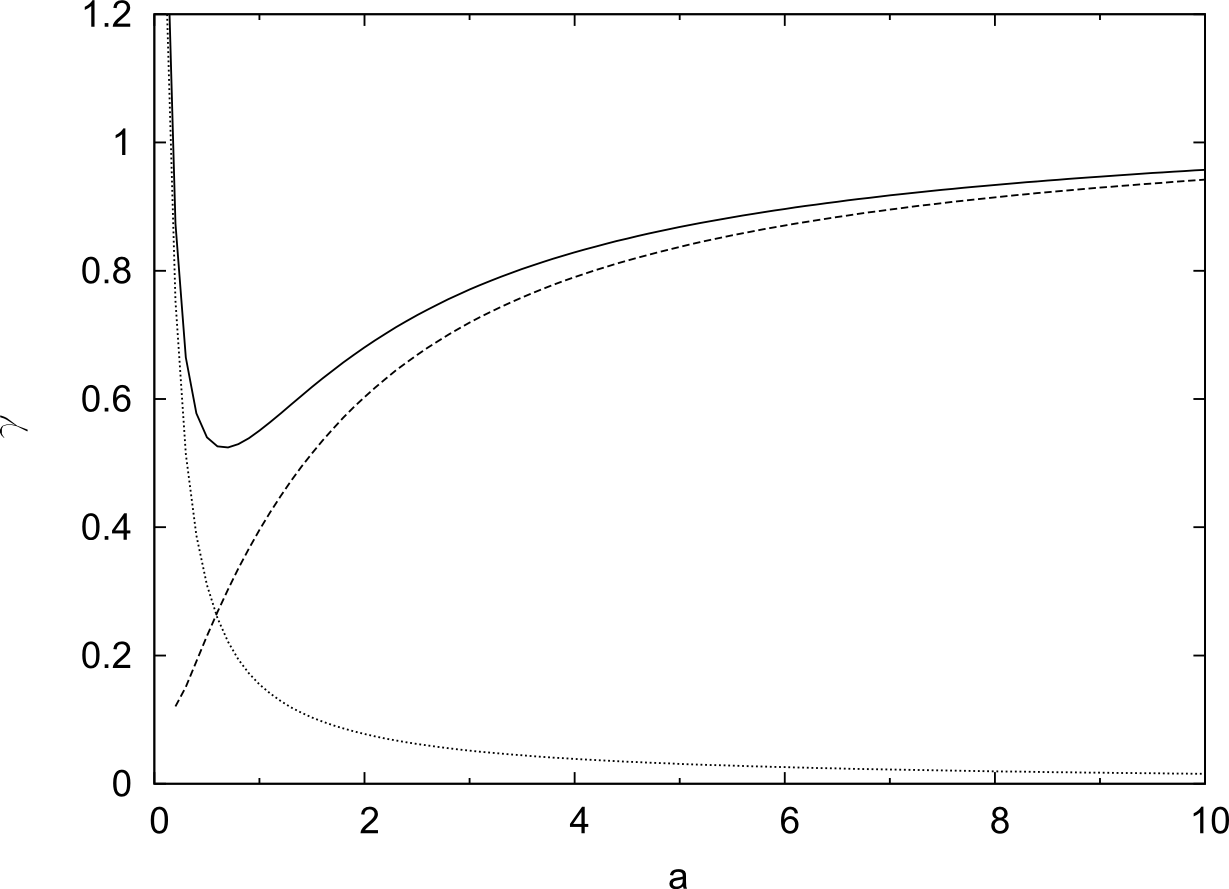}
\caption{Surface tension (\ref{eq:DenProf_LGSurfaceTension}) of the density profile (\ref{eq:Toxvaerd_DensityProfile}) for $T = 0.7$ as a function of the steepness-parameter $a$. The
dashed line represents the contribution from cohesion/ attractive forces to the surface tension, whereas the dotted line represents contributions from the ideal and repulsive terms of the free energy.}
\label{fig:1DInterface_SurfaceTensionVsSteepness}
\end{figure}
 
The density profiles at temperatures up to $T = 0.7$ in Fig.~\ref{fig:LiquidGasInterface_Profiles} suggest an approximation of the liquid-gas interface by a step function, often also called sharp-interface approximation. In order to check if this approximation is appropriate for the calculation of the surface tension, we impose the following analytic expression for the density profile:
\begin{align}
n_{a}(z) = \frac{n_{l}- n_{g}}{2}  \tanh \klamm{ a  \klamm{ \frac{L}{2} - z }} + \frac{n_{l} + n_{g}}{2},
\label{eq:Toxvaerd_DensityProfile}
\end{align}
where $L$ is the length of the domain and $a$ is a parameter for the steepness of the profile. This approximation was introduced by 
Toxvaerd~\cite{Toxvaerd} as the best analytical trial function for the density distribution of the interface. In Fig.~\ref{fig:1DInterface_SurfaceTensionVsSteepness}, the surface tension (\ref{eq:DenProf_LGSurfaceTension}) of $n_{\alpha}(z)$ is plotted versus the steepness parameter $a$ at $T = 0.7$. The minimum of this graph is at $a = 0.7$ and $\gamma_{lg,\infty}[n_{0.7}(\cdot)] = 0.524$. This value is slightly higher than the value obtained if the minimal condition is solved for the full density profile ($\min \gamma_{lg,\infty}[n(\cdot)] = 0.518$). 

For very steep profiles, we expect the surface tension to approach the surface tension of the sharp interface approximation, which is obtained by setting the volume $V_f$ of the interface in (\ref{eq:DFTApplication_DensityDistributionForm}) to zero. This yields for a planar interface
\begin{align*}
\gamma_{lg,\infty}^{SIA} &= 
- \frac{(n_l-n_g)^2}{2} \int_{-\infty}^0 \int_0^\infty \phi_{attr}(|z-z'|) dz' dz\\
&= \frac{3}{4}\pi (n_l-n_g)^2
\end{align*}
At $T= 0.7$, this is $\gamma_{lg,\infty}^{SIA}(T=0.7) = 1.061$, which is close to the surface tension $1.03$ of a very steep $\tanh$-profile (with a steepness-parameter $a=30$). Hence, the sharp-interface surface tension is almost twice the value of the exact surface tension. We conclude that the SIA is not an appropriate method for the calculation of surface tensions.

\paragraph{The Wall-Fluid Interface}

At small distance $z$ from the wall, the particle density $n(z)$ can be approximated by an analytical function. For this, we consider the minimal condition (\ref{eq:DenProf_MinCond_PlanarWall}) for a planar wall. The second term is bounded  by
\begin{align*}
\int_{0}^{\infty} n(z')\Phi_{\text{Pla}}\klamm{|z'-z|} dz' < 
\int_{-\infty}^{\infty} n(z')\Phi_{\text{Pla}}\klamm{|z'-z|} dz'
< 2 \alpha
\end{align*}
where $\alpha$ is defined in (\ref{eq:Uniform_def_alpha}). Remark that for small $z$, the wall potential (\ref{eq:WallPotential}) goes to infinity as $z^{-9}$. Hence, we can write
\begin{align*}
\mu_{HS}(n(z)) + V_{\text{Pla}}(z) + O\klamm{1} &= 0 \qquad \text{ for }\quad z \to 0
\end{align*}
We conclude that the hard-sphere chemical potential $\mu_{HS}$ has to equilibrate the external potential for small $z$. Furthermore, the external potential is positive for $z \to 0$. We now have a closer look at the hard-sphere chemical potential (\ref{eq:Mu_HS_Def}). The algebraic term is positive for $z > 0$. Hence, the external potential is equilibrated by the logarithm in $\mu_{HS}$, which means that the density $n$ has to go to zero as we approach the wall. 

It is easy to show that the algebraic term of $\mu_{HS}$ vanishes with $O(n)$ as $n \to 0$. We conclude that
\begin{align*}
\ln \klamm{\frac{\pi}{6}n(z)} = -\frac{1}{T} V_{\text{Pla}}(z) + O\klamm{1} \qquad \text{ for }\quad z \to 0,
\end{align*}
which finally leads to the prediction 
\begin{align}
n(z) \approx \exp \klamm{ -4\frac{\pi \varepsilon_w \sigma_w^3 }{T}
\klamm{{ \frac{1}{45} \klamm{\frac{\sigma_w}{z}}^9}- 
\frac{1}{6}\klamm{\frac{\sigma_w}{z}}^3}} \qquad \text{ for } \quad z \to 0.
\label{eq:DenProf_ClosetoWallPrediction}
\end{align}
In Fig. \ref{fig:CloseToWall}, this prediction is compared with numerical computations. It shows very good agreement in the relevant range. 

\begin{figure}[ht]
\centering
\includegraphics{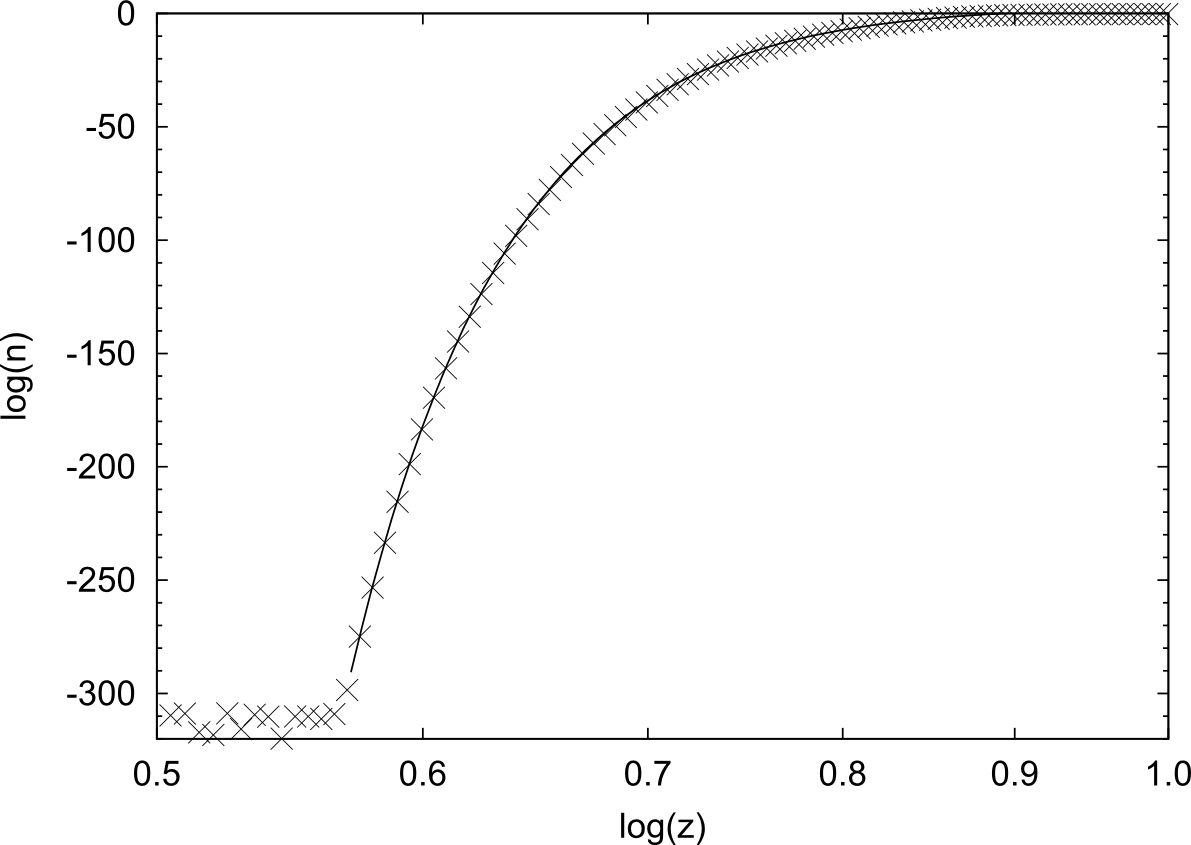}
\caption{Log-Log plot of the particle density close to a wall over the distance $z$ from the wall at $T = 0.7$ and $\Delta \mu = -0.01$. The wall parameters are $\varepsilon_w = 0.8$ and $\sigma_w = 1.25$. The crosses are results of a numerical computation. The solid line is the prediction (\ref{eq:DenProf_ClosetoWallPrediction}). It shows a very good agreement for distances greater than $0.57$ from the wall. For values closer to the wall, the numerical accuracy breaks down.}
\label{fig:CloseToWall}
\end{figure}
\paragraph{A Liquid Film on a Planar Wall}
Fig.~\ref{fig:1DProfile} depicts the density profile of a thin film on a planar wall. It is compared with the density profile of a liquid-gas interface at saturation without external potential. Both profiles show a very good agreement. However, the density of the liquid film is slightly larger than the liquid bulk density close to the wall. It decays with increasing distance from the wall. Close to the wall, the density goes to zero. A more detailed analysis of the behavior close to the wall is given in Eq.~\ref{eq:DenProf_ClosetoWallPrediction}. Furthermore, the density profile shows one oscillation at a distance of some molecule diameters of the wall. This oscillation is due to the hard-sphere reference fluid. Further oscillations are suppressed by the local Carnahan-Starling model which is used in the present work. If a non-local model for the hard-sphere fluid is used, more detailed oscillations close to the wall could be observed (see also Sec.\ref{sec:ModelsHardSphere}). 
\begin{figure}[ht]
\centering
\includegraphics{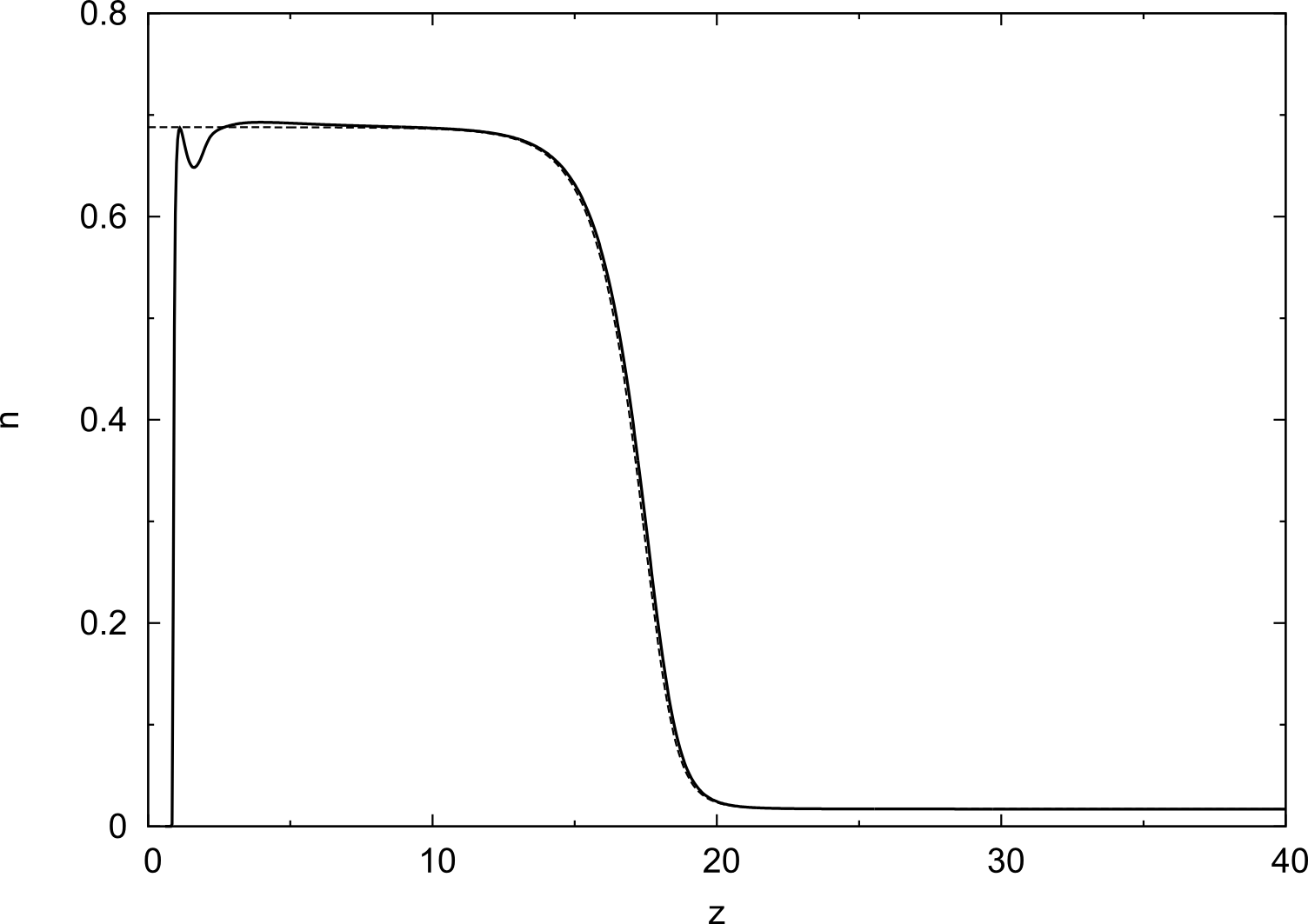}
\caption{Plot of density $n$ over the distance of the wall $z$ for a planar wall with parameters $\varepsilon_w = 0.8$ and $\sigma_w = 1.25$, at temperature $T= 0.7$ (solid line). The deviation of the chemical potential from saturation is $\Delta \mu = -0.001$. The dashed line is the density profile of a liquid-gas interface at saturation without external potential.}
\label{fig:1DProfile}
\end{figure}

\section{Density Profiles for a Thin Film on a Sphere \label{sec:DenProf_Sphere}}
We consider the case of a spherical wall $W = \{ {\bf r} \in \mathbb{R}^3:  |{\bf r}| < R \}$ with radius $R$. The symmetry of this system suggests rotational invariance in the two angular variables in spherical coordinates. The extremal condition (\ref{eq:minCond_dimless}) then reduces to 
\begin{align}
\mu_{HS}\klamm{n(r)} +
 \int_R^\infty  n\klamm{ r'} \Phi_{sph}\klamm{r,r'} dr' 
+  V_{sph,R}\klamm{r} -  \mu {\overset ! =} 0  \qquad \klamm{\forall r > R}.
\label{eq:DenProf_MinCond_SphericalWall}
\end{align}
$\Phi_{Sph}(r,r')$ is the attractive interaction potential between a point at distance $r$ from the origin and the surface of a sphere with radius $r'$. $V_{sph,R}(r)$ is the external potential induced by a sphere of radius $R$. 

\subsection{Analytical Expressions \label{sec:Sphere_AnalyticalExpr}}
\paragraph{Interaction Potential}
The interaction potential between a point at distance $r$ from the origin and the surface of a sphere with radius $r'$ can be written in spherical coordinates as
\begin{align*}
\Phi_{sph}\klamm{r,r'} :&= r'^2  \int_0^{2\pi}  \int_0^\pi \phi_{attr}\klamm{|{\bf r'}- {\bf r}|} \sin \vartheta ' d\vartheta ' d\varphi ',
\end{align*}
where ${\bf r}$ is any point at distance $r$ from the origin. This definition is consistent due to the rotational invariance of the expression. Hence, we set for simplicity ${\bf r} = \klamm{0,0,r}^T$ in cartesian coordinates. Consequently, one gets
\begin{align}
\Phi_{sph}\klamm{r,r'} 
&= 2 \pi r'^2 \int_0^\pi \phi_{attr}\klamm{ \sqrt{r^2 - 2rr'\cos \vartheta ' + r'^2}} \sin \vartheta ' d\vartheta '\notag\\
&= \pi \frac{r'}{r} \int_{(r-r')^2}^{(r+r')^2} \phi_{attr}\klamm{ \sqrt{t}} dt \qquad \text{ where } t = r^2 - 2rr'\cos \vartheta ' + r'^2\notag\\
&=  4 \pi\frac{r'}{r}
\left\{
\begin{array}{ll}
 \int_{(r-r')^2}^{(r+r')^2} dt \klamm{
\frac{1}{t^6}
-\frac{1}{t^3}
} & \text{ if } |r-r'|>1\\
\int_{1}^{(r+r')^2} dt \klamm{
\frac{1}{t^6}
-\frac{1}{t^3}
} & \text{ if } |r-r'|<1, |r+r'|>1\\
0 & \text{ else }\\
\end{array}
\right.
\notag\\
&= 
4 \pi\frac{r'}{r}
\left\{
\begin{array}{ll}
\klamm{
\frac{1}{5}\klamm{
\frac{1}{(r-r')^{10}}
-\frac{1}{(r+r')^{10}}
}
+
\frac{1}{2}\klamm{
\frac{1}{(r+r')^4}
-\frac{1}{(r-r')^4}
}
} & \text{ if } |r-r'|>1 \\
\klamm{
-\frac{1}{5}{
\frac{1}{(r+r')^{10}}
} + \frac{1}{2}{
\frac{1}{(r+r')^4}
} - \frac{3}{10}
}& \text{ if } |r-r'|<1 \text{ and } |r+r'|>1\\
0 & \text{ else }
\end{array}
\right.
\label{eq:DenProf_Phi_Spherical}
\end{align}
This can be written in terms of the planar interaction potential (\ref{eq:Def_Phi}) as
\begin{align*}
\Phi_{sph}\klamm{r,r'} = \frac{r'}{r} \klamm{ \Phi_{\text{Pla}}\klamm{r-r'} - \Phi_{\text{Pla}}\klamm{r + r'}}.
\end{align*}

\paragraph{Influence of Boundary Conditions}
The definition of the interaction potential (\ref{eq:DenProf_Phi_Spherical}) in the spherical case, together with (\ref{eq:DenProf_PsiDef}), leads to 
\begin{align}
\Psi_{in,R}(r) :&= \int_0^R \Phi_{sph}(r,r') dr' \notag\\
&= 
\frac{\pi}{3r}
\left\{
\begin{array}{ll}
\frac{1}{30}\klamm{
\frac{r+9R}{(r+R)^9} - \frac{r-9R}{(r-R)^9}
}+ \frac{r-3R}{(r-R)^3} - \frac{r+3R}{(r+R)^3}
& \text{ if }  R + 1 < r\\
 {\frac {27}{10}}\,-{\frac {26}{15}r}-\frac{9}{5}\,\klamm{{
{R}^{2}- \left( r-1 \right) ^{2}}}+{\frac {r+9\,R}
{30 \left( r+R \right) ^{9}}}-{\frac {r+3\,R}{\left( r+R
 \right) ^{3}}} 
& \text{ if } 1 < R < r < R + 1 \\
 -\frac{9}{5} R^2+{\frac {1}{30}}\,{\frac {r+9\,R}{ \left( r
+R \right) ^{9}}}-{\frac {1}{30 r ^ 8}}-{\frac {r+3\,R}{
 \left( r+R \right) ^{3}}}+\frac{1}{r^2}  
& \text{ if } 0< R < r < 1\\
0 &  \text{ if } R = 0\\
\end{array}\right. .
\label{eq:Sph_PsiIn}
\end{align}

The attractive interaction potential induced by the volume outside a sphere with radius $R>1$ is given by:
\begin{align}
\Psi_{out,R}(r) :&= \int_R^\infty \Phi_{sph}(r,r') dr' \notag\\
&= 
-\frac{\pi}{3r}
\left\{
\begin{array}{ll}
\frac{1}{30}\klamm{
\frac{r+9R}{(r+R)^9} - \frac{r-9R}{(r-R)^9}
}+ \frac{r-3R}{(r-R)^3} - \frac{r+3R}{(r+R)^3}
& \text{ if } |r - R| > 1\\
{\frac {27}{10}}+{\frac {26}{15}}r-\frac{9}{5}\klamm{{
{R}^{2}- \left( r+1 \right) ^{2}}} + {\frac {r+9\,R}
{30 \left( r+R \right)^{9}}} - {\frac {r+3R}{\left( r+R
 \right) ^{3}}}
& \text{ else }\\
\end{array}\right. .
\notag
\end{align}

\paragraph{Wall Potential}
Analogously to the computations above, one easily checks that the external potential (\ref{eq:DenProf_GenDefWallPot}) of a hard sphere $W = \{ {\bf r} \in \mathbb{R}^3:  |{\bf r}| < R \}$ with the wall-fluid interaction potential (\ref{eq:Wall_LennardJonesPotential}) is given by
\begin{align}
V_{sph,R}(r) = 
\varepsilon_w \sigma_w^3 \pi 
\frac{\sigma_w}{3r}\klamm{
\frac{\sigma_w^8}{30}\klamm{
\frac{r+9R}{(r+R)^9} - \frac{r-9R}{(r-R)^9}
}+ 
\sigma_w^2\klamm{
\frac{r-3R}{(r-R)^3} - \frac{r+3R}{(r+R)^3}
}} .
\label{eq:Sph_WallPot}
\end{align}
The external potential induced by a cavity $W = \{ {\bf r} \in \mathbb{R}^3:  |{\bf r}| > R \}$ is given by:
\begin{align}
V_{cav,R}(r) =
- \varepsilon_w \sigma_w^3 \pi \frac{\sigma_w}{3r}
\klamm{
\frac{\sigma_w^8}{30}\klamm{
\frac{r+9R}{(r+R)^9} - \frac{r-9R}{(r-R)^9}
}+ 
\sigma_w^2 \klamm{
\frac{r-3R}{(r-R)^3} - \frac{r+3R}{(r+R)^3}
}}.
\notag
\end{align}

\paragraph{The Adsorption}
In the spherical case, the adsorption per unit area of the substrate is defined by:
\begin{align}
\Gamma[n(\cdot)] \defi \int_{R}^{\infty} \klamm{\frac{r}{R}}^2 \klamm{ n(r) - n_g} dr.
\end{align}

\subsection{Numerical Results}

The density profile of a thin film on a spherical wall is depicted in Fig.~\ref{fig:SphProfile}, where it is compared with the density profile of a planar liquid-gas interface at saturation. Similar to the planar case, both profiles are practically indistinguishable far away from the wall. In Fig.~\ref{fig:RToInf} , density profiles for different radius of the substrate are compared with a density profile of a thin film on a planar wall. The film thickness of the thin film on a spherical substrate is less than the film thickness on a planar substrate. It approaches slowly the planar value with increasing radius of the wall.   

\begin{figure}[ht]
\centering
\includegraphics{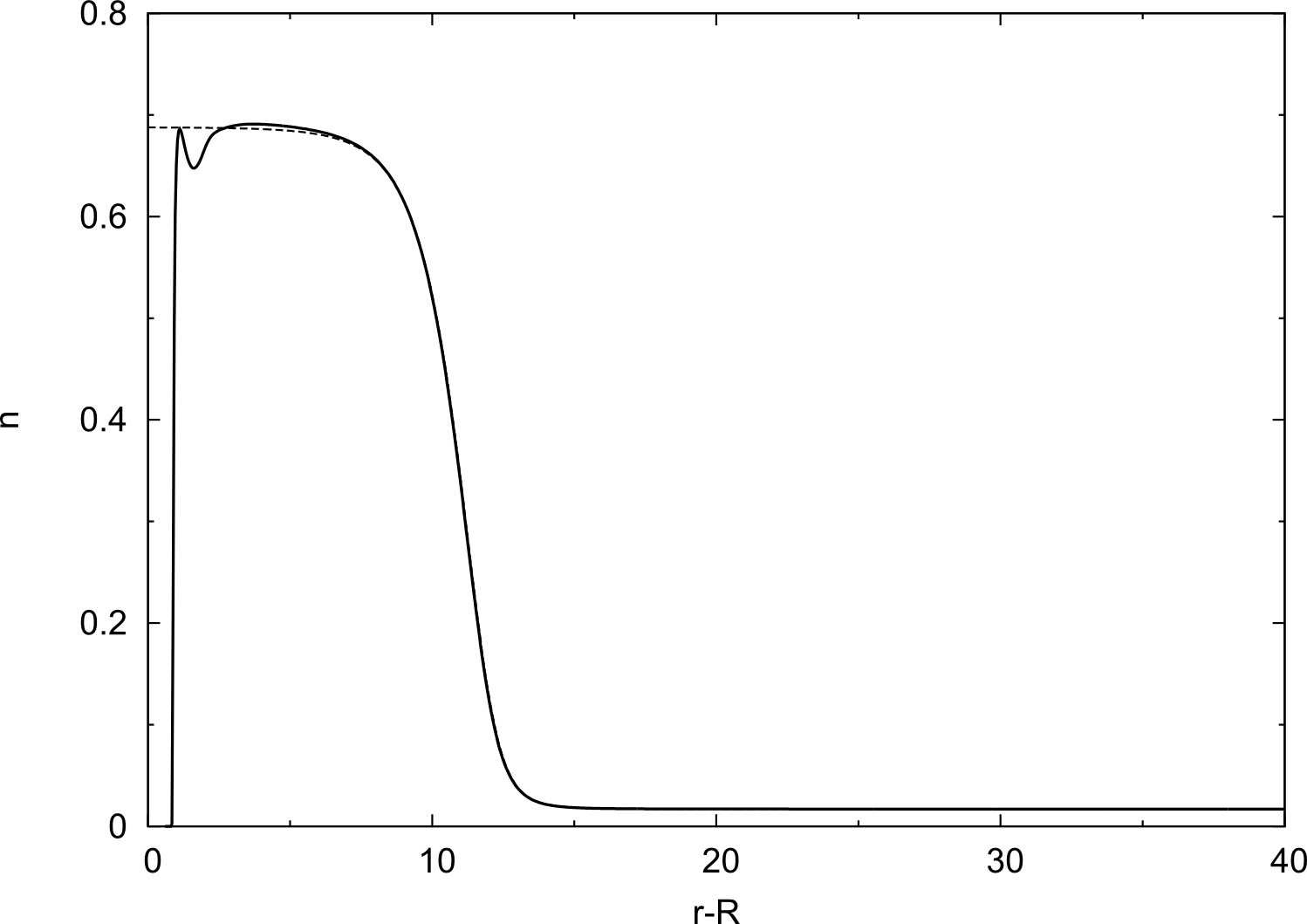}
\caption{Plot of density $n$ over the distance of the wall $r-R$ for a spherical wall with radius $R = 100$ and wall parameters $\varepsilon_w = 0.8$ and $\sigma_w = 1.25$, at temperature $T= 0.7$ (solid line). The deviation of the chemical potential from saturation is $\Delta \mu = -0.001$. The dashed line is the density profile of a liquid-gas interface at saturation without external potential. Far away from the wall, it is practically indistinguishable from the density profile of the thin film.}
\label{fig:SphProfile}
\end{figure}

\begin{figure}[ht]
\centering
\includegraphics{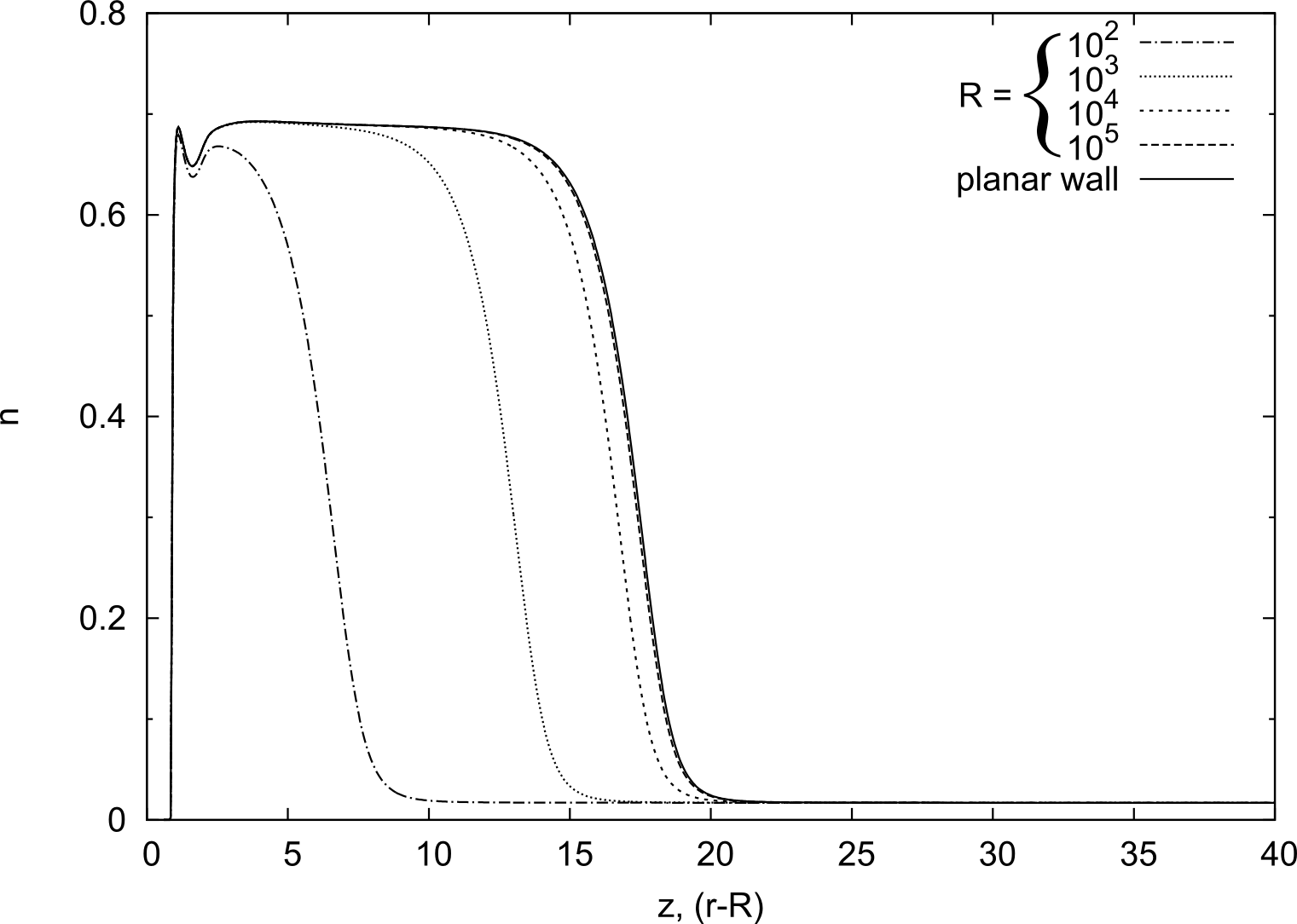}
\caption{Plots of density $n$ over the distance of the wall for spherical walls with varying radius $R$ and for a planar wall at temperature $T= 0.7$. The wall parameters are $\varepsilon_w = 0.8$ and $\sigma_w = 1.25$. The deviation of the chemical potential from saturation is $\Delta \mu = -0.001$.}
\label{fig:RToInf}
\end{figure}

\chapter{Wetting Behavior on a Solid Substrate \label{sec:WettingBehavior}}

In the study of wetting, we are interested in describing the amount of fluid adsorbed on a substrate as a function of the temperature, the chemical potential and the attractiveness of the wall. Analytical methods approximate the density profile $n\ofR$ of the liquid film by a test function $n_\ell\ofR$, where $\ell$ is the thickness of the liquid film. As a result, the grand potential can be written as a function of the film thickness $\ell$: $\Omega(\ell) \defi \Omega[n_\ell\klamm{\cdot}]$. Now, instead of minimizing the functional $\Omega[n\klamm{\cdot}]$ with respect to the density profile $n(\cdot)$, one minimizes $\Omega(\ell)$ with respect to the film thickness $\ell$, thus reducing the complexity of the problem substantially. The main drawback of the analytical methods is that they mainly depend on the quality of the test function $n_\ell\klamm{\cdot}$. In particular, most approximations are not suitable for small film thicknesses. 

This can be avoided, if the variational principle (\ref{eq:minCond_dimless}) is solved for the full density profile $n\ofR$, which can only be done numerically. We present a continuation method which allows to compute the full set of density profiles for a fixed temperature and for a varying chemical potential.
\section{Analytical Methods for the Prediction of Wetting Behavior for One-Dimensional Geometries \label{sec:Wetting_Analytical}}

We present two methods to approximate the grand potential function as a function of the film thickness, namely the sharp-interface approximation (SIA) and the piecewise function approximation (PFA). In particular, we introduce the analytic expressions to calculate these approximations, without specifying the geometry. This allows to apply the given methods on more complex structures beyond the planar or the spherical substrate at a later stage of research. In order to do this, the grand potential will be written in terms of volumes, basically the wall volume $W$, a film volume $V_f$ and a bulk volume $V_B$. At a later stage, the volumes can be parameterized by the film thickness $\ell$, which leads to the grand potential as a function of $\ell$.

We assume a density distribution as in (\ref{eq:DFTApplication_DensityDistributionForm}) where $V_A$ is the volume occupied by the wall.  The fluid density in the wall is zero such that we get
\begin{align*}
n\ofR =
\left\{
\begin{array}{lll}
0 & \text{ if } & {\bf r} \in W\\
n_f\ofR & \text{ if } & {\bf r} \in V_f\\
n_B & \text{ if } & {\bf r} \in V_B
\end{array}
\right. . 
\end{align*}
In Fig. \ref{fig:PartitionSpace}, a typical configuration is shown together with the respective density profile close to a solid substrate.

\begin{figure}[ht]
\centering
\includegraphics[width=10cm]{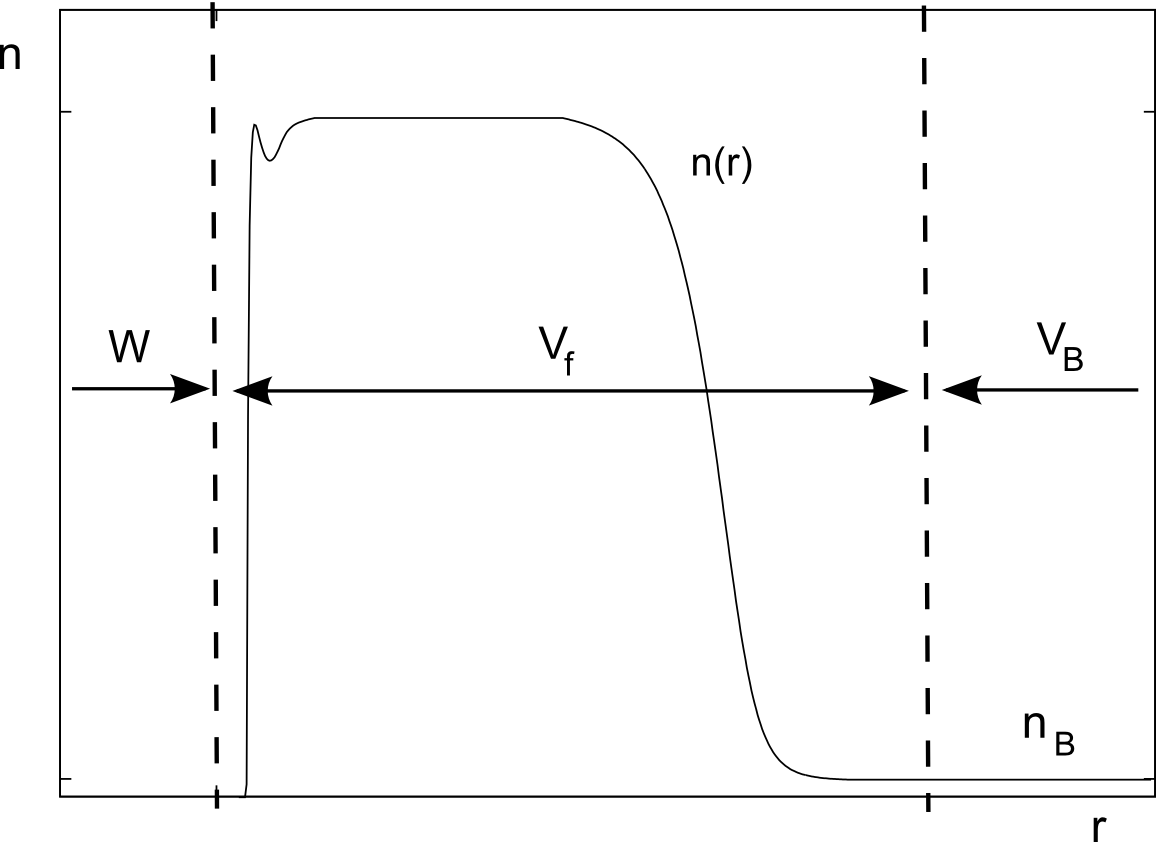}
\caption{Sketch of a partition of the space for a liquid film on a solid substrate. $V_A$ and $V_B$ are the volumes in which uniform density is assumed. $V_f$ is the volume in which computations are executed.}
\label{fig:PartitionSpace}
\end{figure}

\subsection{The Sharp Interface Approximation}
\begin{figure}[ht]
\centering
\includegraphics{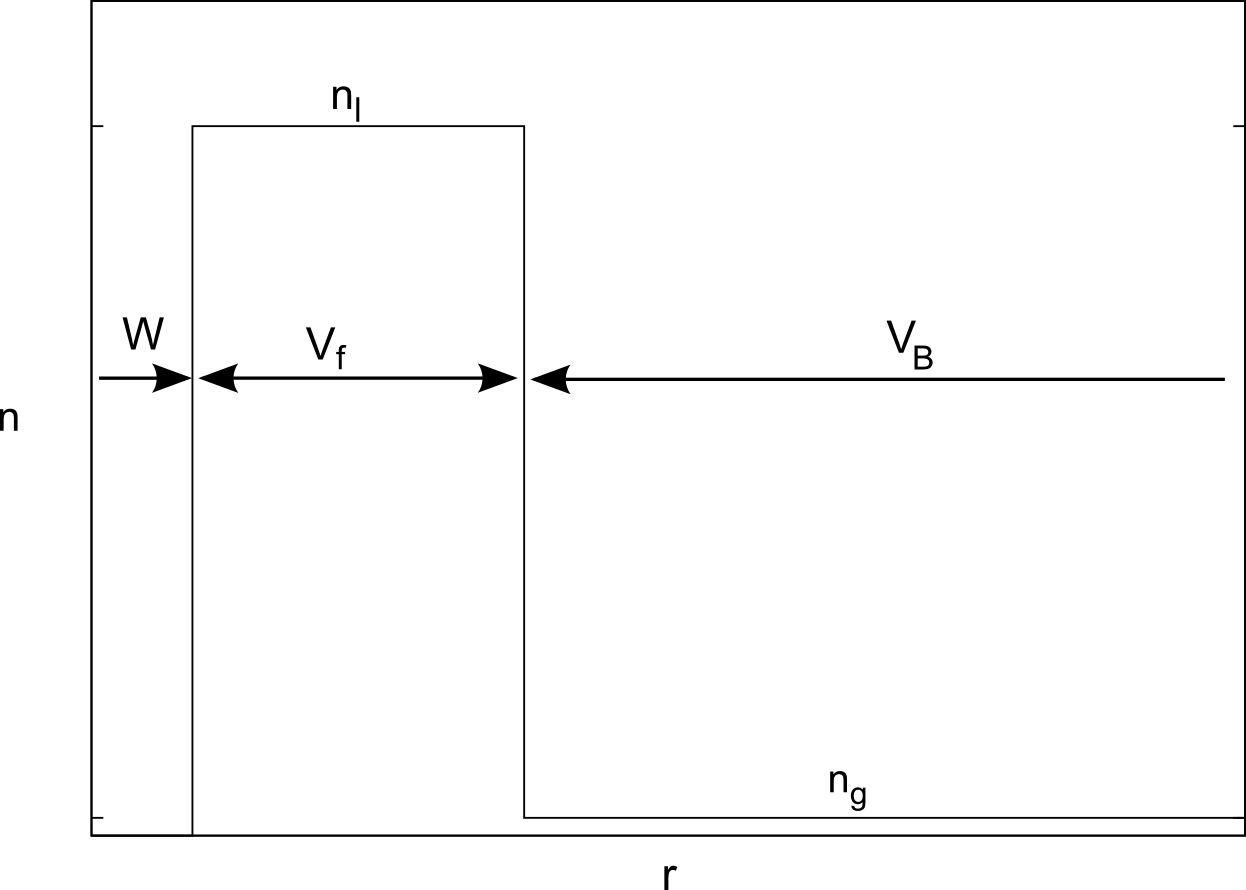}
\caption{Sketch of a density profile of a sharp interface approximation. $W$ is the volume in which the density is assumed to be zero. This includes the wall plus a thin layer around the wall, in which the repulsive forces of the wall dominate.}
\label{fig:SIA_Density}
\end{figure}

In the SIA for a liquid film on a solid substrate, the density in the film volume $V_f$ is assumed to be uniform equal to the bulk liquid density $n_l$. Remark that close to the wall, the film density goes to zero due to the repulsive character of the wall (see also Fig.~\ref{fig:CloseToWall}). This can be taken into account by extending the volume $W$ at which the density vanishes up to a certain distance from the wall. Consider that in this case, $W$ is the volume of the wall plus a thin layer close to the wall (see Fig.~\ref{fig:SIA_Density}). 
\begin{align*}
n^{SIA}\ofR \defi
\left\{
\begin{array}{lll}
0 & \text{ if } & {\bf r} \in W\\
n_l & \text{ if } & {\bf r} \in V_f\\
n_g & \text{ if } & {\bf r} \in V_B
\end{array}
\right. . 
\end{align*}
Making the above assumptions, the excess grand potential (\ref{eq:Def_ExcessGrandPotential}) can be written as follows:
\begin{align}
\Omega_{ex}^{SIA} \defi \Omega_{ex}[n^{SIA}\klamm{\cdot}] 
= &- \int_{V_A} \klamm{p(n\ofR)-p\klamm{n_A}} d{\bf r} 
- \int_{V_B} \klamm{p(n\ofR)-p\klamm{n_B}} d{\bf r}+   \int n\ofR V \ofR  d{\bf r} +\notag\\
&+\frac{1}{2} \iint n\ofR \klamm{n\ofRD - n\ofR } \phi_{attr} \klamm{|{\bf r'}-{\bf r}|} d{\bf r}' d{\bf r}\\
=&   - \Delta p |V_f|
+  n_l \int_{V_f} V\ofR d{\bf r} +  n_g \int_{V_B} V\ofR d{\bf r} + \ldots\notag  \\
& 
- n_l^2 I \klamm{V_f,W}
- \klamm{n_l - n_g}^2 I\klamm{V_f,V_B}
- n_g^2 I \klamm{W,V_B}
 \label{eq:SharpKink_OmegaEx_2}
\end{align}
where $\Delta p = p\klamm{n_l} - p\klamm{n_g}$, $|\cdot|$ is a measure for the volume and $I\klamm{\cdot,\cdot}$ is an operator defined by:
\begin{align*}
I: \mathcal{P}\klamm{\mathbb{R}^3} \times \mathcal{P}\klamm{\mathbb{R}^3} \to \mathbb{R} \quad,\quad
I(A,B) \defi \int_A  \int_B \phi_{attr} \klamm{|{\bf r} - {\bf r}'|} d{\bf r}' d{\bf r}.
\end{align*}
Here, $\mathcal{P}\klamm{\mathbb{R}^3}$ is the powerset of $\mathbb{R}^3$, i.e. the set of all subsets of $\mathbb{R}^3$. The first term in (\ref{eq:SharpKink_OmegaEx_2}) can be directly linked with the deviation of the chemical potential from its saturation value $\mu_{sat}$ at a given temperature. For this, assume that the density and the chemical potential are close to their values at saturation. Then, expanding the pressure (\ref{eq:Uniform_Omega_def}) as a function of the density and the chemical potential around saturation up to terms of first order, gives:
\begin{align*}
- p(n,\mu) =& f_{HS}(n) n + \alpha n^2 - \mu n \\
=& - p\klamm{n_{sat},\mu_{sat}} + \klamm{ \mu_{HS}(n_{sat}) + 2\alpha n_{sat} - \mu_{sat} }\klamm{n-n_{sat}} - n_{sat} \klamm{\mu-\mu_{sat}} + HOT.
\end{align*} 
The term $\mu_{HS}(n_{sat}) + 2\alpha n_{sat} - \mu_{sat}$ vanishes, as $n_{sat}$ is the equilibrium density at $\mu_{sat}$ (see also Eq. (\ref{eq:Uniform_Omega_prime})). $n_l$ and $n_g$ are the liquid and gas bulk densities at $\mu$, whereas $n_{l,sat}$ and $n_{g,sat}$ are the liquid and gas bulk densities at $\mu_{sat}$, respectively. Hence, we obtain
\begin{align}
- p\klamm{n_l,\mu}
 \approx& -p\klamm{n_{l,sat},\mu_{sat}}  - n_{l,sat} \klamm{\mu - \mu_{sat}}\notag\\
\text{and} \qquad - p\klamm{n_g,\mu} \approx& -p\klamm{n_{g,sat},\mu_{sat}} - n_{g,sat} \klamm{ \mu - \mu_{sat}}.\notag
\end{align}
At saturation, the bulk gas pressure equals the bulk liquid pressure (see Eq.(\ref{eq:Uniform_eq})). Hence, the pressure difference can be written as 
\begin{align}
\Delta p = p\klamm{n_l} - p\klamm{n_g} \approx& \klamm{n_{l,sat} - n_{g,sat}}\Delta \mu,
\label{eq:Wetting_ApproxPressureChemPot}
\end{align}
where $\Delta \mu = \mu - \mu_{sat}$. Note that now assumptions on the form of the pressure have been made in order to obtain Eq. (\ref{eq:Wetting_ApproxPressureChemPot}). We now turn our attention to the terms in the last line of Eq. (\ref{eq:SharpKink_OmegaEx_2}), which can be rearranged as follows:
\begin{align*}
- n_l^2 I\klamm{V_f\cup V_B,W}
-\klamm{n_l - n_g}^2 I \klamm{ V_f\cup W,V_B}
+ 2n_l\klamm{n_l - n_g} I \klamm{W,V_B}.
\end{align*}
This leads to the following
\begin{align}
\Omega_{ex}^{SIA} = &  - \Delta \mu \klamm{ n_{l,sat} - n_{g,sat} } |V_f|
+ \Omega_{wl}^{SIA} + \Omega_{lg}^{SIA} + \Omega_B^{SIA}, 
 \label{eq:SharpKink_OmegaEx_4}
\end{align}
where $\Omega_{wl}^{SIA}$ and $\Omega_{lg}^{SIA}$ are the sharp-interface wall-liquid and the liquid-gas excess grand potentials defined by
\begin{align}
\Omega_{wl}^{SIA} &\defi - \frac{n_l^2}{2} I\klamm{W,V_f \cup V_B} +  n_l \int_{V_f \cup V_B} V\ofR d{\bf r}
\label{eq:SIA_OmegaWL}
\\
\text{and}\qquad 
\Omega_{lg}^{SIA} &\defi - \frac{(n_l-n_g)^2}{2} I\klamm{W\cup V_f , V_B}.
\label{eq:SIA_OmegaLG}
\end{align}
$\Omega_B^{SIA}$ is the binding potential defined by
\begin{align}
\Omega_B^{SIA} \defi n_l \klamm{n_l - n_g} I\klamm{W,V_B} - (n_l-n_g)\int_{V_B} V\ofR d{\bf r}. \label{eq:SharpKing_defBinding}
\end{align}
Physically, (\ref{eq:SharpKink_OmegaEx_4}) can be explained as follows: $\Omega_{wl}^{SIA}$ and $\Omega_{lg}^{SIA}$ are needed to create the wall-liquid and the liquid-gas interface. However, it has to be taken into account that the surfaces interact with each other, as their distance is finite. This is done by introducing the binding energy $\Omega_B$. For further details, see also Israelachvili~\cite{Israelachvili}.

\subsection{The Piecewise Function Approximation \label{sec:PFA}}

The main drawback of the SIA is that it overestimates the liquid-gas surface tension by up to one hundred percent (see also Fig.~\ref{fig:1DInterface_SurfaceTensionVsSteepness}). In order to avoid effects due to this error, we introduce the PFA, where it is assumed that the wall-liquid and the liquid-gas-interface are smooth and have a finite width. The test function can then be written as:
\begin{align}
n^{PFA}\klamm{\bf r} \defi
\left\{
\begin{array}{lll}
0 & \text{ if } & {\bf r} \in W \\
n_{wl}\ofR & \text{ if } & {\bf r} \in V_{wl} \\
n_{l} & \text{ if } & {\bf r} \in V_{f}\\
n_{lg}\ofR & \text{ if } & {\bf r} \in V_{lg}\\
n_{g} & \text{ if } & {\bf r} \in V_{B} 
\end{array}
\right.,
\label{eq:FiniteDifference_Splitting}
\end{align}
where $V_{wl}$ and $V_{lg}$ are the volumes of the wall-liquid and the liquid-gas interface (see also Fig.~\ref{fig:PFA_Volumes}).
\begin{figure}[ht]
\centering
\includegraphics{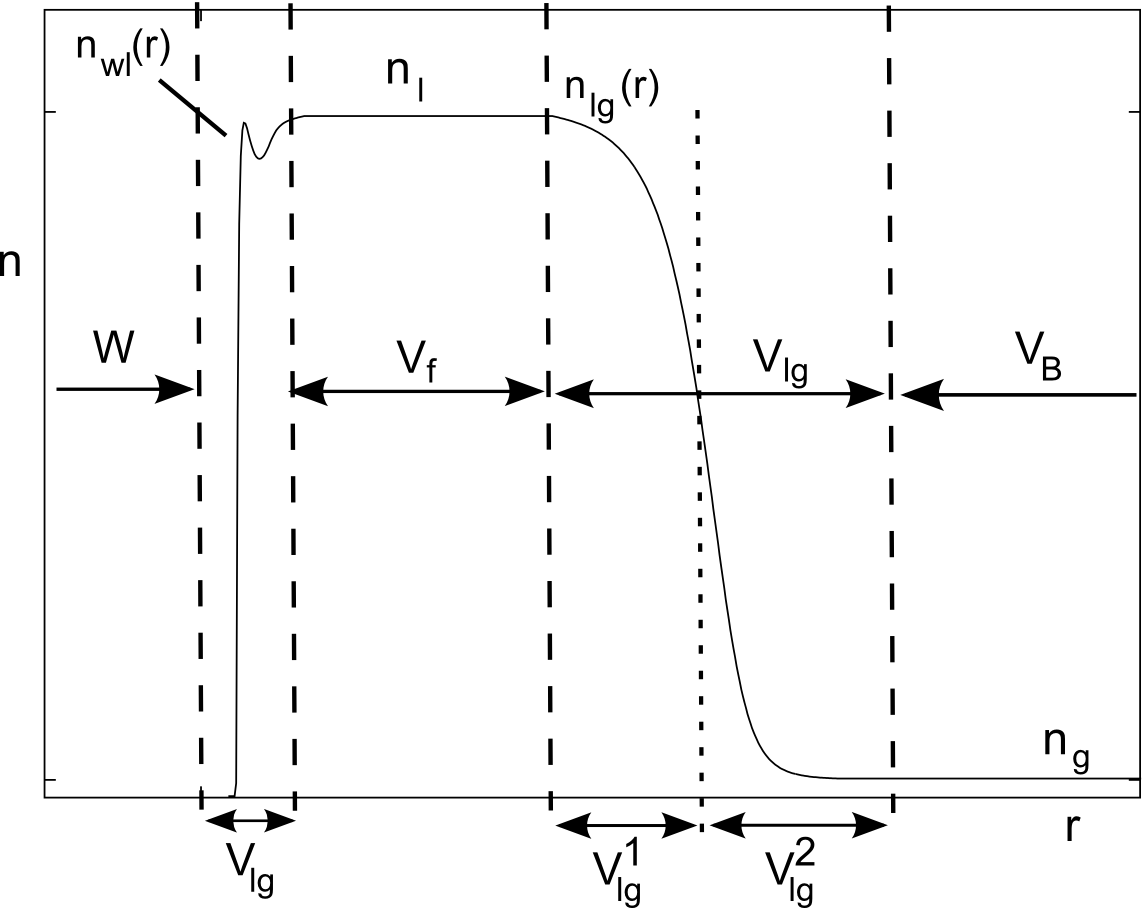}
\caption{Sketch of a density profile of a piecewise function approximation. $n_{wl}\klamm{\bf r}$ and $n_{lg}\klamm{\bf r}$ are the density profiles of the wall-liquid and the liquid-gas interface, respectively. The width of both interfaces is supposed to be independent on the amount of adsorbed fluid. The liquid-gas interface $V_{lg}$ is separated in the two parts $V_{lg}^{1}$ and $V_{lg}^2$ by the Gibbs dividing surface.}
\label{fig:PFA_Volumes}
\end{figure}

It is our aim to write the excess grand potential (\ref{eq:Def_ExcessGrandPotential}) for the test function (\ref{eq:FiniteDifference_Splitting}) in a similar way to the formulation (\ref{eq:SharpKink_OmegaEx_4}) in the SIA. For this, first the attractive contribution to the excess grand potential (\ref{eq:Def_ExcessGrandPotential}) is analyzed:
\begin{align*}
\text{A} \defi
\frac{1}{2}\iint n\ofRD \klamm{n\ofRD - n\ofR} \phi_{attr}\klamm{|{\bf r} - {\bf r}'|} d{\bf r}' d{\bf r} = 
\quad
&[WV_{wl}] + \{V_{wl}V_{wl}\} 
+ [\klamm{W \cup V_{wl}} V_f ]+\\
+ &[V_{lg}V_B] +  \{V_{lg}V_{lg}\}
+ [V_f \klamm{V_{lg} \cup V_{B}}]+\\
+ &[\klamm{W \cup V_{wl}}\klamm{V_{lg} \cup V_B}]
\end{align*}
where the operators $[\cdot\cdot]$ and $\{\cdot\cdot\}$ are defined as
\begin{align*}
[XY] &\defi -\frac{1}{2}\int_X \int_Y \klamm{n^{PFA}\ofR - n^{PFA}\ofRD}^2 \phi_{attr}\klamm{|{\bf r}-{\bf r}'|}d{\bf r}' d{\bf r}\\
\text{and}\qquad 
\{XX\} &\defi \frac{1}{2}\int_X \int_X n_X\ofR \klamm{n^{PFA}\ofRD - n^{PFA}\ofR} \phi_{attr}\klamm{|{\bf r}-{\bf r}'|} d{\bf r}' d{\bf r}.
\end{align*}

We define the density profiles of a wall-liquid and a liquid-gas interface as
\begin{align}
\bar n_{wl}\klamm{\bf r} \defi
\left\{
\begin{array}{lll}
0 & \text{ if } & {\bf r} \in W \\
n_{wl}\ofR & \text{ if } & {\bf r} \in V_{wl} \\
n_{l} & \text{ if } & {\bf r} \in  V_{f} \cup V_{lg} \cup V_{B} 
\end{array}
\right.
\qquad \text{and}\qquad
\bar n_{lg}\klamm{\bf r} \defi
\left\{
\begin{array}{lll}
n_l & \text{ if } & {\bf r} \in W \cup V_{wl} \cup V_{f}\\
n_{lg}\ofR & \text{ if } & {\bf r} \in V_{lg} \\
n_{g} & \text{ if } & {\bf r} \in  V_{B} 
\end{array}
\right. .
\end{align}
The excess grand potential (\ref{eq:Def_ExcessGrandPotential}) with respect to these density distributions will be denoted by $\Omega_{wl}^{PFA}$ and $\Omega_{lg}^{PFA}$, respectively. Their attractive contributions are
\begin{align*}
\text{A}^{wl} = & [WV_{wl}] + \{V_{wl}V_{wl}\} 
+ [\klamm{W \cup V_{wl}} V_f ]
+ [\klamm{W \cup V_{wl}}\klamm{V_{lg} \cup V_B}]_{wl}\\
\text{and} \qquad
\text{A}^{lg} = & [V_{lg}V_B] + \{V_{lg}V_{lg}\} 
+ [V_f \klamm{V_{lg} \cup V_{B}}] 
+[\klamm{W \cup V_{wl}}\klamm{V_{lg} \cup V_{B}}]_{lg},
\end{align*}
where $[\cdot \cdot]_{wl}$ or $[\cdot \cdot]_{lg}$ means that the operator $[\cdot \cdot]$ is evaluated with respect to $\bar n_{wl}(\cdot)$ and $\bar n_{lg}(\cdot)$ instead of $n^{PFA}(\cdot)$, respectively. The contributions of the pressure term to $\Omega_{wl}^{PFA}$ and $\Omega_{lg}^{PFA}$ can be written as
\begin{align*}
\text{P}^{wl} &\defi  - \int_{V_{wl}} \klamm{p\klamm{n_{wl}\ofR} - p\klamm{n_{l}}} d{\bf r}\\
\text{and}\qquad
\text{P}^{lg} &\defi  - \int_{V_{lg}^1} \klamm{p\klamm{n_{lg}\ofR} - p\klamm{n_{l}}} d{\bf r}
- \int_{V_{lg}^2} \klamm{p\klamm{n_{lg}\ofR} - p\klamm{n_{g}}} d{\bf r},
\end{align*}
where $V_{lg}^1$ and $V_{lg}^2$ correspond to the division of $V_f$ by the Gibbs dividing surface for the liquid-gas interface such that
\begin{align*}
\int_{V_{lg}^1} \klamm{ n_{lg}\ofR - n_l } d{\bf r}+
\int_{V_{lg}^2} \klamm{ n_{lg}\ofR - n_g } d{\bf r}
=
0.
\end{align*}
The pressure-term contributions and the contributions of the attractive terms sum up with contributions from the external potential to
\begin{align}
\Omega_{wl}^{PFA} &\defi \Omega_{ex}[\bar n_{wl}(\cdot)] = \text{P}^{wl}+ \text{A}^{wl} +  \int_{V_{wl}} V\ofR n_{wl}\ofR d{\bf r} +  n_{l} \int_{V_f \cup V_{lg}^l \cup V_B^l} V\ofR d{\bf r}
\label{eq:PFA_OmegaWL}
\\
\text{and} \qquad
\Omega_{lg}^{PFA} &\defi \Omega_{ex}[\bar n_{lg}(\cdot)] = \text{P}^{lg}+ \text{A}^{lg}.
\label{eq:PFA_OmegaLG}
\end{align}
After having calculated the excess grand potentials for the wall-liquid and liquid-gas interface, we analyze the remaining terms of the excess grand potential of $n^{PFA}\ofR$. For this, we subtract the attractive terms $\text{A}^{wl}$ and $\text{A}^{lg}$ from the attractive term $\text{A}$, which yields:
\begin{align}
\text{A} - \text{A}^{wl} - \text{A}^{lg} &=[\klamm{W \cup V_{wl}}\klamm{V_{lg}\cup V_B}]
-[\klamm{W \cup V_{wl}}\klamm{V_{lg}\cup V_B}]_{wl}
-[\klamm{W \cup V_{wl}}\klamm{V_{lg}\cup V_B}]_{lg}\notag\\
&=- \frac{1}{2}\int_{W \cup V_{wl}} \int_{V_{lg}\cup V_B}
\left( \klamm{n^{PFA}\ofR - n^{PFA}\ofRD}^2 - \klamm{n^{PFA}\ofR - n_l}^2 - \right.\notag\\
& \qquad \qquad \qquad \qquad -\left. \klamm{n_l - n^{PFA}\ofRD}^2\right) \phi_{attr}\klamm{|{\bf r} - {\bf r}'|} d{\bf r}' d{\bf r}\notag\\
&=
\int_{W\cup V_{wl}} \int_{V_{lg}\cup V_B} \klamm{n_{l} - n^{PFA}\ofRD }\klamm{n_{l} - n^{PFA}\ofR} \phi_{attr}\klamm{|{\bf r} - {\bf r}'|} d{\bf r}'d{\bf r},
\label{eq:PFA_RemainNL}
\end{align}
where we have used that in $W \cup V_{wl}$, $n^{PFA}\ofR$ equals $\bar n_{wl}\ofR$ and in $V_{lg}\cup V_B$, $n^{PFA}\ofR$ equals $\bar n_{lg}\ofR$. The pressure terms of the excess grand potential can be written as follows:
\begin{align}
P &= 
-\int_{V_{wl}} \klamm{p\klamm{n_{wl}\ofR} - p\klamm{n_g\ofR}} d{\bf r}
- \int_{V_f} \klamm{p\klamm{n_l} - p\klamm{n_g}} d{\bf r}
- \int_{V_{lg}} \klamm{p\klamm{n_{lg}\ofR} - p\klamm{n_g}} d{\bf r}\notag\\
&= 
\text{P}^{wl} + \text{P}^{lg} - 
 \int_{V_{wl}\cup V_f \cup V_{lg}^1} \klamm{p\klamm{n_{l}} - p\klamm{n_{g}}} d{\bf r}\notag\\
&\approx \text{P}^{wl} + \text{P}^{lg} - \Delta \mu \klamm{n_l - n_g} |V_{wl}\cup V_f \cup V_{lg}^1|, \label{eq:PFA_LocalTerm}
\end{align}
where we have made use of approximation (\ref{eq:Wetting_ApproxPressureChemPot}) for the pressure-term. Putting together (\ref{eq:PFA_LocalTerm}) and (\ref{eq:PFA_RemainNL}) leads to the following approximate expression for the excess grand potential of the configuration $n^{PFA}\ofR$:
\begin{align}
\Omega_{ex}^{PFA} \defi& 
- \Delta \mu \klamm{n_l - n_g} |V_{wl}\cup V_f \cup V_{lg}^1|
+ \Omega_{lg}^{PFA} + \Omega_{wl}^{PFA} + \Omega_B^{PFA} \approx 
\Omega_{ex}[n^{PFA}(\cdot)].
\label{eq:OmegaEx_PFA}
\end{align}
The binding potential in the spherical case is the sum of the remaining attractive contribution (\ref{eq:PFA_RemainNL}) and the contribution from the external potential:
\begin{align}
\Omega_B^{PFA} \defi&
\int_{W\cup V_{wl}} \int_{V_{lg}\cup V_B} \klamm{n_{l} - n^{PFA}\ofRD }\klamm{n_{l} - n^{PFA}\ofR} \phi_{attr}\klamm{|{\bf r} - {\bf r}'|} d{\bf r}'d{\bf r}  \ldots \notag\\
&- \int_{V_{lg}\cup V_B} \klamm{ n_{l} - n^{PFA}\ofR}V\ofR d{\bf r}.
\label{eq:OmegaB_PFA}
\end{align}

\section{Numerical Method: The Pseudo Arc Length Continuation \label{sec:ContinuationMethod}}

Additionally to the analytical approaches, the minimization problem is also solved numerically for the full density profile. In Sec.~\ref{sec:DenProf_Numerics}, we introduced numerical methods in order to get one density profile ${\bf n}$ for each chemical potential $\mu$. However, in the case of a prewetting transition, there can be multiple solutions for one chemical potential. Out of these solutions, only one is stable, whereas the other solutions are meta- or unstable. In order to compute the full bifurcation diagram of the set of density profiles over the chemical potential, a pseudo arc length continuation method is employed. 

We introduce an arc-length parametrization such that $\klamm{\mu(s),{\bf n}(s)}$ with $s \in \mathbb{R}$ is a connected set of solutions of condition (\ref{eq:g_Discretized}), where we include the chemical potential $\mu$ as an additional variable:
\begin{align}
{\bf g}\klamm{\mu,{\bf n}} \istobe 0. \label{eq:Annex_Conti_1DDiscCondition}
\end{align}
The density inside the wall is zero whereas far away from the substrate the density is assumed to equal the bulk gas density. Hence we set in (\ref{eq:g_Discretized}) $n_- = 0$ and $n_+ = n_g$. The main idea of the continuation scheme is to trace the set of solutions along the curve parametrized by $s$. In order to do so, it is assumed that a point $\klamm{\mu^n,{\bf n}^n}$ at position $s^n$ on the curve of solutions is given, where $n$ is the step of the continuation scheme being solved for.
\begin{figure}[htb]
\begin{center}
\includegraphics[width=7cm]{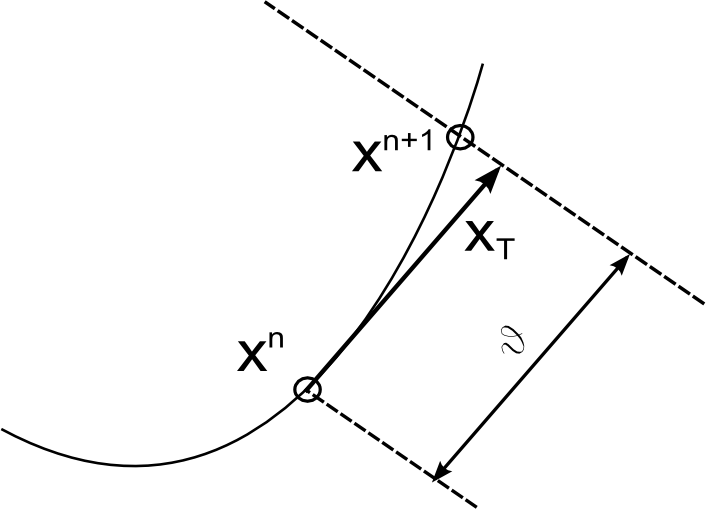}
\caption{Sketch of one iteration step of the continuation scheme. ${\bf x}^n$ and ${\bf x}^{n+1}$ are consecutive points of the iteration, where, ${\bf x} = \klamm{\mu, {\bf n}}$. ${\bf x}_T$ is the tangent vector in ${\bf x}^n$. By following the curve of solutions in direction of the tangent vector, the pseudo arc length continuation scheme is able to trace the curve of solutions through turning points. 
\label{fig:SktechContinuationMethod}}
\end{center}
\end{figure}

First, the tangent vector $\klamm{\dif{\mu}{s},\dif{{\bf n}}{s}}$ at position $s^n$ is computed. This is done by differentiating the function ${\bf g}(s) \defi {\bf g}\klamm{\mu(s),{\bf n}(s)}$ with respect to $s$. From (\ref{eq:Annex_Conti_1DDiscCondition}), it is known that ${\bf g}$ is zero on the curve of solutions $\klamm{\mu(s),{\bf n}(s)}$. Hence, the differential $\dif{{\bf g}}{s}$ vanishes:
\begin{align*}
\dif{{\bf g}}{s} = 
\left(
\begin{array}{cc}
\diff{{\bf g}}{ \mu} & 
{\bf J}
\end{array}
\right)
\cdot 
\left(
\begin{array}{c}
\dif{\mu}{s}\rule[-1.7ex]{0pt}{0pt} \\
\dif{{\bf n}}{s}
\end{array}
\right)
= 0,
\end{align*}
where ${\bf J}$ is the Jacobian $\diff{{g_i}}{ n_j}$ defined in (\ref{eq:DenProf_Jacobi_Def}) and
\begin{align}
\diff{{g_i}}{ \mu} = -1 + \dif{n_{g}}{\mu} \Psi_{+}\klamm{ z_N -  z_i}. 
\label{eq:Continuation_GiMu}
\end{align}
The second term takes into account that the boundary condition $ n_+ = n_{g}$ is a function of the chemical potential. From the equilibrium condition (\ref{eq:Uniform_Omega_prime}), it can be easily seen that 
\begin{align}
\dif{n_g}{\mu} = \frac{1}{2\alpha + \mu_{HS}'(n_g)}.
\label{eq:DngDmu}
\end{align}
In Tab.~\ref{tab:DngDmu}, some values for $\dif{n_g}{\mu}$ at saturation are given. The last term in (\ref{eq:Continuation_GiMu}), $\Psi_+$, is given in the planar case by (\ref{eq:1D_Numerics_Vn}). The absolute value of this expression is bounded by $1.81$ for $|z_n-z_i| > 1$. At distances less than $1$ from the wall, the $|\Psi_{+}|$ is bounded by $5.59$.
\begin{table}
\begin{center}
\begin{tabular}[ht]{l|llllllll}
\toprule
$T$ &
0.5& 0.55	&0.6	&0.65&	0.7&	0.75&	0.8&	0.85\\
$\left. dn_g/d\mu\right|_{sat}$ &
$0.00183$	&$0.00402$	&$0.00856$	&$0.0165$&	$0.0303$	&$0.0551$&	$0.1049$	&$0.2393$\\
\bottomrule
\end{tabular}
\caption{Values of the derivative of the gas bulk density with respect to the chemical potential at different temperatures (see Eq.~(\ref{eq:DngDmu})).}
\label{tab:DngDmu}
\end{center}
\end{table}
We conclude that for temperatures less than $0.8$, $\diff{g_i}{\mu} \approx -1$ is a reasonable approximation which only leads to a slight torsion of the tangent vector. Hence, the defining equation for the approximate tangent vector  $\klamm{\mu_T, {\bf n}_T} \approx \klamm{\dif{\mu}{s},\dif{{\bf n}}{s}}$ is written as:
\begin{align}
\left(
\begin{array}{cc}
- {\bf 1} & 
{\bf J}
\end{array}
\right)
\cdot 
\left(
\begin{array}{c}
\mu_T \\
{\bf n}_T
\end{array}
\right)
= 0. \label{eq:Annex_DefiningEquationTangent}
\end{align}
Remark that this homogeneous system of linear equations leaves one degree of freedom, as we only have $N+1$ equations, but $N+2$ variables $(\mu_T,{\bf n}_T)$. The additional equation can be used to determine whether the tangent vector points in positive or negative direction in terms of the arc length $s$.

In order to find the next point $\klamm{\mu^{n+1},{\bf n}^{n+1}}$ on the curve, an additional equation has to be set up. For this purpose we introduce a scalar product for the $\klamm{\mu,{\bf n}}$-space, which takes into account the discretization of the density profile into $N$ intervals of length $\Delta  z$:
\begin{align}
\langle \klamm{\mu_1 , {\boldsymbol n}_1} |  \klamm{\mu_2 , {\boldsymbol n}_2}\rangle 
\defi& 
\mu_1 \mu_2 +
\frac{\Delta z}{2} n_{10}n_{20} +
\Delta z \sum_{j=1}^{N-1} n_{1j}n_{2j}+
\frac{\Delta z}{2} n_{1N}n_{2N}  \notag\\
=&
\mu_1 \mu_2 +\frac{\Delta z}{2} \sum_{j=0}^N \klamm{2 - \delta_{j0} - \delta_{jN}} n_{1j}n_{2j}, \notag
\end{align}
where $\delta_{ij}$ is the Kroenecker delta.
The norm with respect to the scalar product is defined as 
\begin{align*}
\norm{\klamm{ \mu,{\boldsymbol n}}} \defi
\langle
\klamm{ \mu,{\boldsymbol n}}|
\klamm{ \mu,{\boldsymbol n}}
\rangle^{1/2}.
\end{align*}
We say that the curve of solutions $\klamm{ \mu(s),{\boldsymbol n}(s)}$ is parametrized by arc length with respect to the norm given above. Hence, the length between two points on the curve can be written as
\begin{align*}
\int_{s^n}^{s^n+\theta} \norm{\klamm{\dif{\mu}{s},\dif{{\boldsymbol n}}{s}}} ds = \theta.
\end{align*}
Linearizing the norm around $s^n$ and making use of the approximate tangent vector $(\mu_T,{\boldsymbol n}_T)$ at $s^n$, one obtains
\begin{align*}
\langle
\klamm{ \mu_T,{\boldsymbol n}_T}|
\klamm{ \mu(s^n + \theta) -  \mu(s^n) , {\boldsymbol n}(s^n + \theta) -  {\boldsymbol n}(s^n)}
\rangle \approx \theta,
\end{align*}
where we used that the tangent vector is normalized such that
\begin{align*}
\norm{\klamm{ \mu_T,{\boldsymbol n}_T}} = 1.
\end{align*}
Using $\klamm{\mu^{n+1},{\boldsymbol n}^{n+1}}$ instead of $\klamm{\mu(s^n + \theta), {\boldsymbol n}(s^n + \theta)}$ leads to the additional equation for the next point on the curve of solutions:
\begin{align}
K_n\klamm{\mu^{n+1},{\boldsymbol n}^{n+1}} \defi \langle
\klamm{ \mu_T,{\boldsymbol n}_T}|
\klamm{ \mu^{n+1} -  \mu^n, {\boldsymbol n}^{n+1} -  {\boldsymbol n}^{n}}
\rangle - \theta \quad \istobe \quad 0. \label{eq:Annex_Conti_AdditionalEq}
\end{align}
For a geometric interpretation of equation (\ref{eq:Annex_Conti_AdditionalEq}), see also Fig. \ref{fig:SktechContinuationMethod}.

In order to obtain the next point $\klamm{\mu^{n+1},{\boldsymbol n}^{n+1}}$ on the curve, (\ref{eq:Annex_Conti_AdditionalEq}) is solved together with (\ref{eq:Annex_Conti_1DDiscCondition}). This is done using a modified Newton-Scheme. In each Newton-step, the following system of linear equations is solved:
\begin{align}
\left(
\begin{array}{cc}
\mu_T & ({\bar{\boldsymbol n}}_T)^T\\
-{\bf 1} & {\bf J}
\end{array}
\right)
\cdot 
\left(
\begin{array}{c}
\Delta \mu\\
\Delta {\boldsymbol n}
\end{array}
\right)
= 
\left(
\begin{array}{c}
K_n\klamm{\mu^{n,m},{\boldsymbol n}^{n,m}} \\
 {\bf g}(\mu^{n,m},{\boldsymbol n}^{n,m})
\end{array}
\right),
\label{eq:LinearEqnIsotherm}
\end{align}
where we are considering the $n$-th step of the continuation scheme and the $m$-th step of the Newton method, such that $\Delta  \mu\defi \mu^{n,m+1}-\mu^{n,m}$ and $\Delta  {\boldsymbol n} \defi {\boldsymbol n}^{n,m+1}-{\boldsymbol n}^{n,m}$. Again, we have 
approximated $\diff{g_i}{ \mu}$ by $-1$. In (\ref{eq:LinearEqnIsotherm}), $\bar{\bf n}_T$ is defined by
\begin{align*}
{\bar{\boldsymbol n}}_{T,j} \defi \frac{\Delta z}{2} \klamm{2 - \delta_{j0} - \delta_{jN}}{\bf n}_{T,j}.
\end{align*}
Finally, (\ref{eq:LinearEqnIsotherm}) is solved using a conjugate gradient method, where the Jacobian of the system is approximated by introducing a cutoff of $5$  for the intermolecular potential $\Phi$ (see also Sec.~\ref{sec:DenProf_Numerics}).

Remark that the approximation made for $\diff{g_i}{\mu}$ does not affect the accuracy of the result. This is because the defining equations $ {\bf g}(\mu^n,{\bf n}^n)$ of the isotherm are not affected. Instead, the approximation leads to negligible deviations of the step size between two points of the iteration process $\klamm{\mu^{n+1},{\bf n}^{n+1}}$ and $\klamm{\mu^{n},{\bf n}^{n}}$ on the curve of solutions.

\subsection{The Maxwell Construction \label{sec:MaxwellConstruction}}

Once a set of solutions $\klamm{\mu(s),{\bf n}(s)}$ is computed, it is of particular interest to find solutions for which the excess grand potential is equally large and which are at the same chemical potential. These solutions denote first order wetting transitions, as we shall demonstrate in Sec.~\ref{sec:PlanarIsotherm}. The Maxwell construction offers an easy way to compute such points. In order to introduce this method, we assume that the set of solutions is given in its continuous form as $\klamm{\mu(s),n(s)(\cdot)}$, where $n(s)(\cdot)$ is the continuous density profile at position $s$ on the curve of solutions.

The excess grand potential (\ref{eq:Def_ExcessGrandPotential}) is a function of the chemical potential and a functional of the density profile $n(\cdot)$. Hence, the difference of the excess grand potential between two points on the curve can be written as follows:
\begin{align*}
\Omega_{ex}(s + \theta) - \Omega_{ex}(s)
&=
\int_s^{s+\theta}
\klamm{
 \left.
\diff{\Omega_{ex}}{\mu} \dif{\mu}{s}
\right|_{\mu(s),n(s)(\cdot)}
+
\int \klamm{ \left.\frac{\delta \Omega_{ex}}{\delta n\ofR}\right|_{\mu(s),n(s)(\cdot)}  \dif{n}{s}\ofR} d{\bf r} 
}
ds,
\end{align*}
where $\frac{\delta \Omega_{ex}}{\delta n\ofR}$ is the functional derivative of $\Omega_{ex}$ in ${\bf r}$ and $ \Omega_{ex}$ as a function of the parameter $s$ is the external potential evaluated at $\klamm{\mu(s),n(s)(\cdot)}$. $n(s)(\cdot)$ is a solution of the variational principle (\ref{eq:minCond_dimless}). Hence, the second term in the integral vanishes such that we obtain
\begin{align*}
\Omega_{ex}(s + \theta) - \Omega_{ex}(s)
&=
\int_s^{s+\theta}
\klamm{
\diff{\Omega_{ex}}{\mu} \dif{\mu}{s}}
ds
=
\int_{\mu(s)}^{\mu(s+\theta)}
\diff{\Omega_{ex}}{\mu} d\mu
\end{align*}
given that the mapping $\mu(\cdot): [s,s+\theta]\to\mathbb{R}$ is injective, i.e. given that there are no turning points with respect to $\mu$ between $s$ and $s+\theta$. The derivative of $\Omega_{ex}$ (see Eqs.(\ref{eq:Def_ExcessGrandPotential}) and (\ref{eq:Uniform_Omega_def})) with respect to $\mu$ yields
\begin{align*}
\diff{\Omega_{ex}}{\mu} = - \int_{V_f} \klamm{ n\ofR - n_g} d{\bf r} = - \Gamma,
\end{align*}
where $\Gamma$ is the excess number of particles of the system, also denoted as adsorption. This gives
\begin{align*}
\Omega_{ex}(s + \theta) - \Omega_{ex}(s)
&=
- \int_{\mu(s)}^{\mu(s+\theta)}
\Gamma(\mu) d\mu.
\end{align*}

\begin{figure}[ht]
\centering
\includegraphics{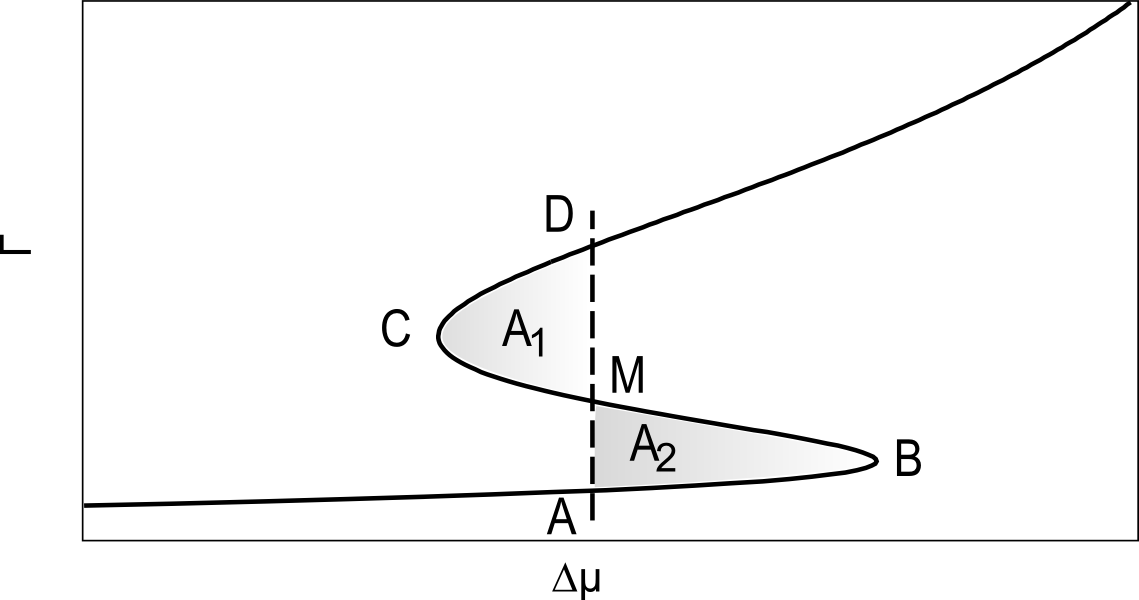}
\caption{Projection of a curve of solutions $\klamm{\mu(s),n(s)(\cdot)}$ on a $\Gamma-\Delta \mu$ - diagram. The multi-valued S-type curve of solution is the signature of a first-order wetting transition.}
\label{fig:MaxwellConstruction}
\end{figure}
We apply this equation on a scenario where the curve of solutions is multi-valued in $\mu$ as shown in Fig.\ref{fig:MaxwellConstruction}. Then, the difference of the excess grand potential between points M and A corresponds with the area $A_2$ enclosed by the curve of solutions A-B-M and the dashed line in Fig.~\ref{fig:MaxwellConstruction}:
\begin{align*}
\Omega_M - \Omega_A &= \klamm{\Omega_M - \Omega_B} + \klamm{\Omega_B - \Omega_A}
= 
\int_M^B\Gamma(\mu) d\mu - \int_A^B \Gamma(\mu) d\mu 
= A_2.
\end{align*}
Analogously, $\Omega_D - \Omega_M$ corresponds with the negative area $A_1$ enclosed by M-C-D and the dashed line in Fig.~\ref{fig:MaxwellConstruction}.
\begin{align*}
\Omega_D - \Omega_M &= \klamm{\Omega_D - \Omega_C} + \klamm{\Omega_C - \Omega_M}
= 
- \klamm{ \int_C^D\Gamma(\mu) d\mu - \int_C^M \Gamma(\mu) d\mu} 
= - A_1.
\end{align*}
Hence, the solution in D has the same excess grand potential as the solution in A, if the areas $A_1$ and $A_2$ in the $\Gamma-\mu$-diagram are equally large:
\begin{align*}
\Omega_D - \Omega_A = A_2 - A_1
\end{align*}

\section{Wetting on a Planar Wall \label{sec:Wetting_PlanarWall}}

\subsection{The Isotherm}

We now present connected sets of solutions of the extremal condition
(\ref{eq:g_Discretized}) that include metastable and unstable
branches of the isotherms. These results were obtained using the pseudo arc length continuation method introduced in Sec.~\ref{sec:ContinuationMethod}. We will also present
phase diagrams for the prewetting line and compare our results with
analytical predictions obtained from a SIA. We note that the majority of previous DFT
computations trace stable or metastable equilibrium density
profiles, but not unstable branches, with a few notable exceptions
which use continuation schemes, e.g. the DFT study of polymer
systems by Frischknecht \emph{et. al.}~\cite{Frischknecht}.

\subsection{Isotherms for a Planar Wall \label{sec:PlanarIsotherm}}
\begin{figure}[hbt]
\begin{center}
\includegraphics[width=8.5cm]{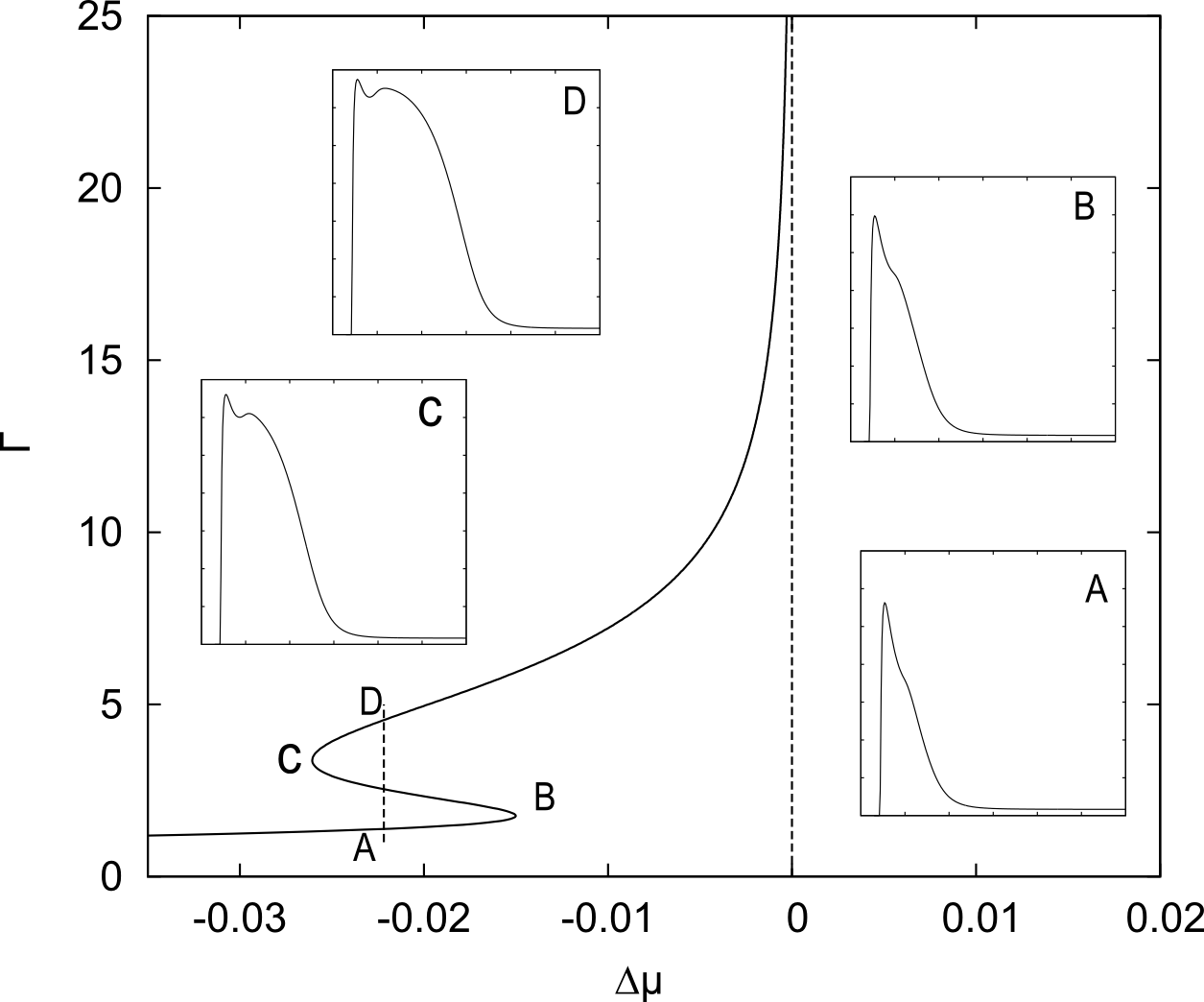}
\caption{ $\Gamma$--$\Delta \mu$
bifurcation diagram at $T= 0.7$ for a wall with
$\varepsilon_w = 0.8$ and $\sigma_w = 1.25$.
$\Delta \mu$ is the deviation of the chemical potential from its
saturation value, $\mu_{sat}$. The prewetting transition, marked by
the dashed line, occurs at chemical potential $\Delta \mu_{pw} = -0.022$.
The inset subplots show the density $n$ as a function of
the distance $z$ from the wall in the ranges $[0,0.7]$ over
$[0,12]$}
\label{fig:IsothermProfiles}
\end{center}
\end{figure}
A typical  bifurcation diagram of the adsorption $\Gamma$ as
a function of the deviation of the
chemical potential from its saturation value $\Delta \mu$, is shown in Fig.~\ref{fig:IsothermProfiles}. At $\Delta \mu = 0$ the bulk gas and the
bulk liquid phases are equally stable, whereas for $\Delta \mu < 0$,
the bulk gas phase is more stable.

The isotherm depicted in Fig.~\ref{fig:IsothermProfiles} is a
multi-valued S-type curve with two turning points (saddle nodes). A
thin liquid film of at most a few molecular diameters is effectively
formed between the wall and the gas and with a few small
oscillations in the density profile near the wall, corresponding to
adsorption of the liquid particles there. This film can only exist
due to the attraction to the wall. For large negative $\Delta
\mu$, the film is about one molecular layer on the wall
($\Gamma \approx 1$). In this case, the shape of the density
distribution is similar to the profile shown in subplot A of
Fig.~\ref{fig:IsothermProfiles}. With increasing $\Delta \mu$, point
A is reached, where two equally stable states A and D coexist,
corresponding to a thin and a thick film, respectively. At A a
first-order phase transition takes place, also known as prewetting
transition, and the corresponding value of the chemical potential
will be denoted as $\mu_{pw}$. The line $\mu = \mu_{pw}$ will be
referred to as the prewetting line.

At $\mu_{pw}^+$ the lower branch stops representing equilibrium
density profiles. The corresponding states are no longer global
minima of the grand potential, as can be inferred from Fig.
\ref{fig:OmegaMuIsotherm} but local ones and the branch from A to B
is a metastable one. At point B a saddle-node bifurcation occurs
connecting the metastable branch A-B from the unstable branch B-C
which is connected with the metastable branch C-D by a second
saddle-node bifurcation at C. We note that the location of the
prewetting line can be obtained from a Maxwell construction in which
the area between the isotherm to the left of the line and the line
equals to the area between the isotherm to the right and the line (see also Sec.~\ref{sec:MaxwellConstruction}).

After crossing the prewetting line and as saturation line is approached,
the thickness of the (single stable) film tends to infinity and we
approach the case of a liquid-gas interface in the absence of the
wall, i.e. it is like the wall is not even present -- with the
exception of course of the area close to it where the density
oscillations occur.

\begin{figure}
\begin{center}
\includegraphics[width=8.5cm]{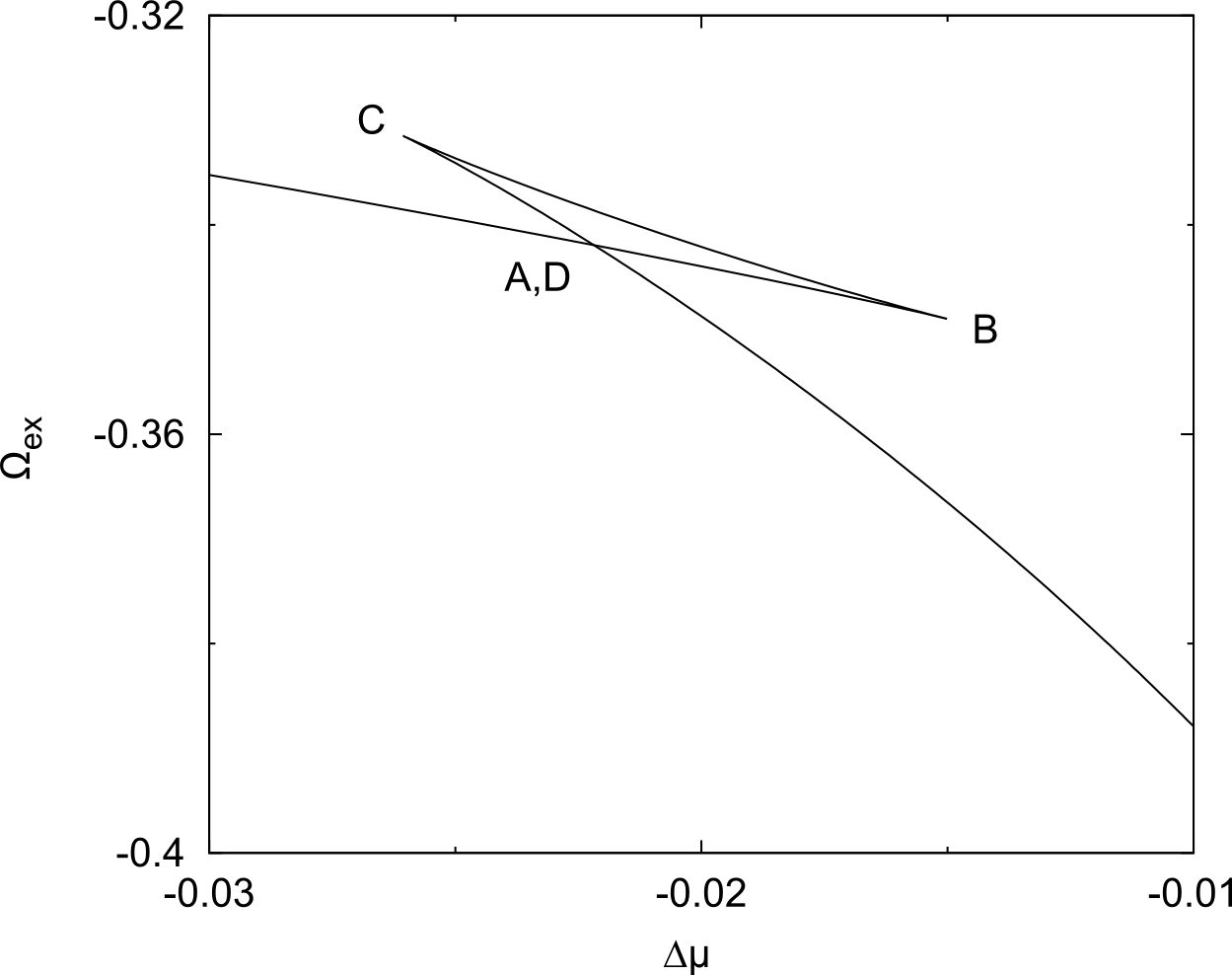}
\caption{\label{fig:OmegaMuIsotherm}The excess grand potential
$\Omega_{ex}$ as a function of $\Delta \mu$ in the vicinity of the prewetting transition
in~\ref{fig:IsothermProfiles}.}
\end{center}
\end{figure}

\subsection{The Prewetting Line \label{sec:PlanarPrewetting}}

The dependence of the prewetting chemical potential, $\mu_{pw}$, as
a function of temperature and the attractiveness of the wall is of particular interest. Increasing the
temperature results in the prewetting line shifted to the left. At
the same time, the jump of the film thickness between point A and
point D decreases. Above a certain temperature, the jump vanishes
and we have a complete wetting scenario, for which the film
thickness grows continuously to infinity (see Fig.~\ref{plot:CompleteWetting}) . On the other hand,
decreasing the temperature will lead to a shift of the prewetting
transition chemical potential, $\mu_{pw}$, towards the saturation
value $\mu_{sat}$. Let us denote with $T_w$ the temperature at which
the prewetting line coincides with the saturation line. For
temperatures below $T_w$, we obtain a partial wetting scenario,
characterized by a stable thin film at saturation (see Fig.~\ref{plot:PartialWetting}). We note that
$\mu_{sat}$ imposes an upper bound on $\mu_{pw}$ such that we cannot
have $\mu_{pw} \geq \mu_{sat}$ (equivalently a Maxwell construction
in this region is not possible).

\begin{figure}[p]
\centering
\subfigure[Partial Wetting]{
\includegraphics[width=8cm]{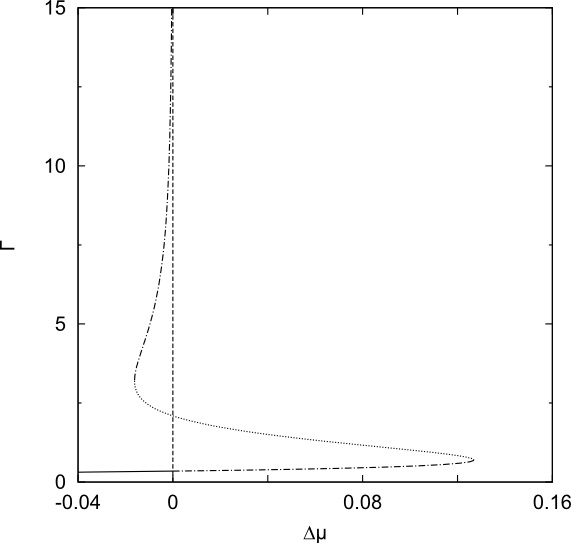}
\label{plot:PartialWetting}
}
\subfigure[Complete Wetting, preceded by a prewetting transition]{
\includegraphics[width=8cm]{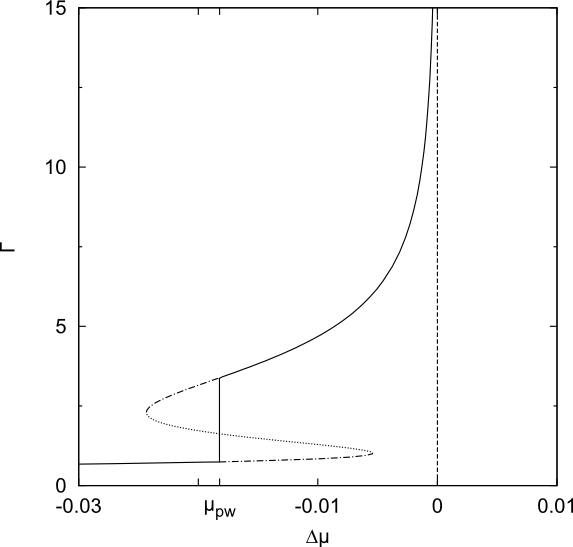}
\label{plot:PreWetting}
}
\subfigure[Complete Wetting without prewetting transition]{
\includegraphics[width=8cm]{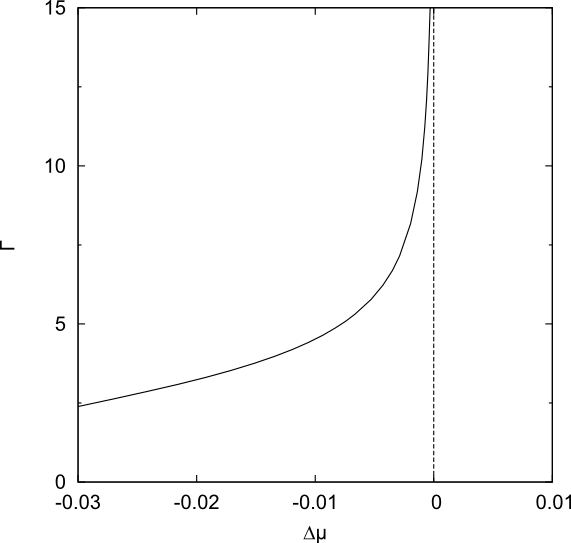}
\label{plot:CompleteWetting}
}
\caption{Plots of the isotherms of adsorption $\Gamma$ over the deviation of the chemical potential from saturation $\Delta \mu$ for (a) a partial wetting scenario at $T = 0.6$, (b) a complete wetting scenario, preceded by a prewetting transition at $T = 0.69$ and (c) a complete wetting scenario without prewetting transition at $T = 0.74$. The wall parameters are $\varepsilon_w = 0.8$ and $\sigma_w = 1.25$. The solid lines are the stable branches of the isotherms. The metastable branches are the dot-dash lines whereas unstable branches are drawn as dotted lines.}
\label{fig:CompletePartialPre}
\end{figure}

Figure~\ref{fig:MuNodd_Temp} depicts the deviations of the chemical
potential from the saturation one at prewetting, $\Delta\mu_{pw}$,
and at the left and right saddle nodes as a function of temperature.
$\Delta\mu_{pw}$ approaches saturation as $(T-T_w)^{3/2}$. Consequently, the slope of the prewetting line at $T_w$ is zero:
\begin{align*}
\left.\dif{\klamm{\Delta \mu_{pw}}}{T}\right|_{T_w} = 0
\end{align*}
 This
appears to be in agreement with the analytical thermodynamical prediction based on Clapeyron-type equations by Hauge and Schick~\cite{HaugeSchick}, who stated that the prewetting line approaches the saturation line tangentially.

The influence of the attractiveness of the wall on the chemical potential at the prewetting line is shown in Fig.~\ref{fig:MuNodd_Epsilonw}. Similar to the dependence of the temperature, there is a $\varepsilon_{w,w}$ which separates a partial wetting scenario ($\varepsilon_w < \varepsilon_{w,w}$) from a complete wetting scenario, which is preceded by a prewetting transition. Increasing the attractiveness of the wall above a certain value $\varepsilon_{w,cw}$ leads to a complete wetting scenario. In the vicinity of the transition from a prewetting scenario to a complete wetting scenario, the jump of the film thickness at the prewetting transition decays as $\klamm{\varepsilon_{w,cw}-\varepsilon_w}^{1/2}$ (see also Fig.~\ref{fig:TransitionPrewettingCompleteWetting} and Fig.~\ref{fig:JumpPrewettingEpsilonW}). 

\begin{figure}[phtb]
\begin{center}
\includegraphics[width=12cm]{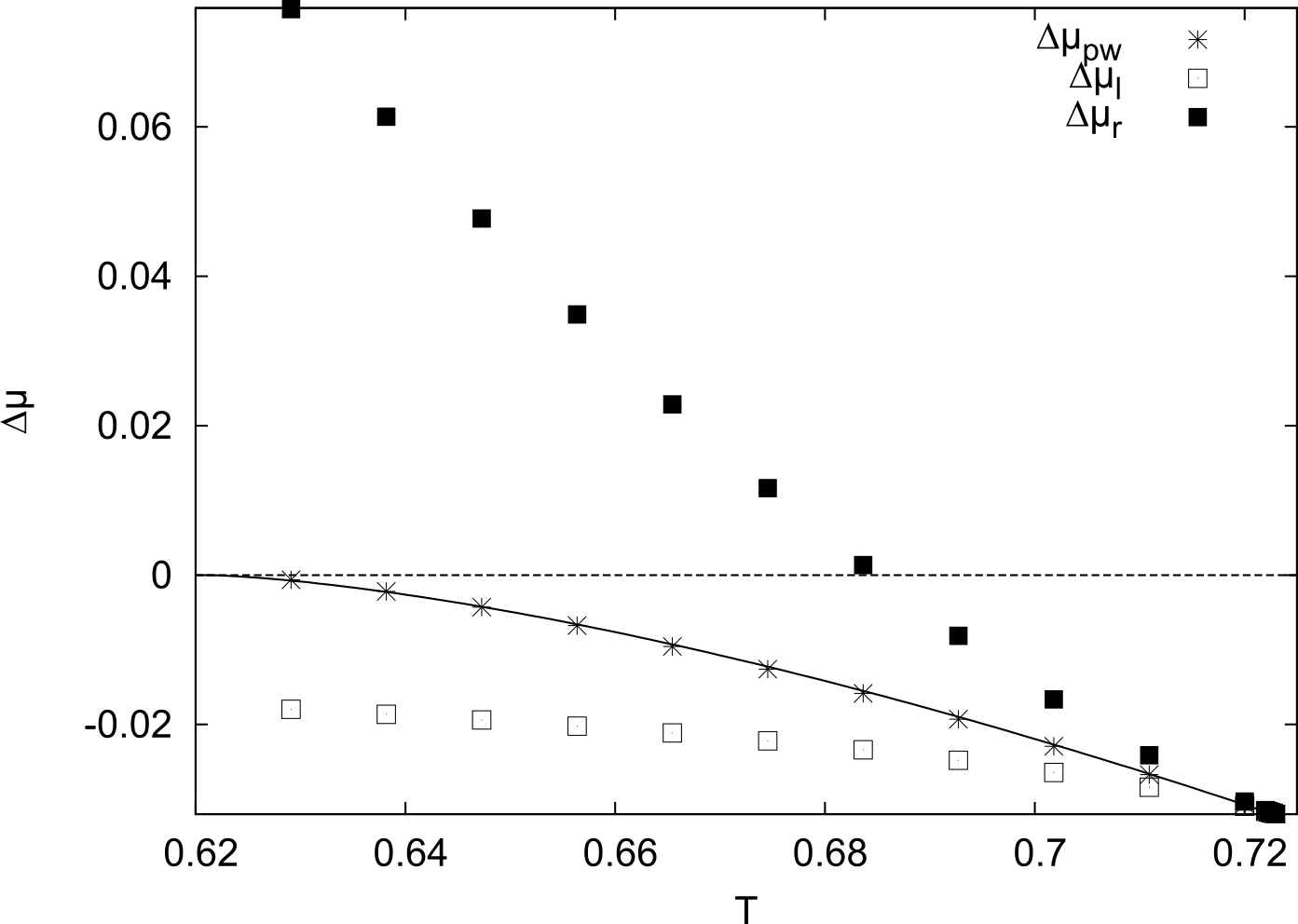}
\caption{Deviations of the chemical
potential from saturation at prewetting, $\Delta\mu_{pw}$,
and at the left and right saddle nodes as a function of
temperature. The wall parameters are $\sigma_w = 1.25$ and $\varepsilon_w = 0.8$. The solid line is the fit to $\Delta
\mu_{pw}(T) = -C (T - T_w)^{3/2}$ with $T_w = 0.62079$. The resulting coefficient is $C = 0.9839$. The dashed line marks the
locus of the chemical potential at saturation for the given temperature, $\Delta \mu = 0$. 
\label{fig:MuNodd_Temp}
}
\end{center}
\end{figure}

\begin{figure}[p]
\begin{center}
\includegraphics[width=12cm]{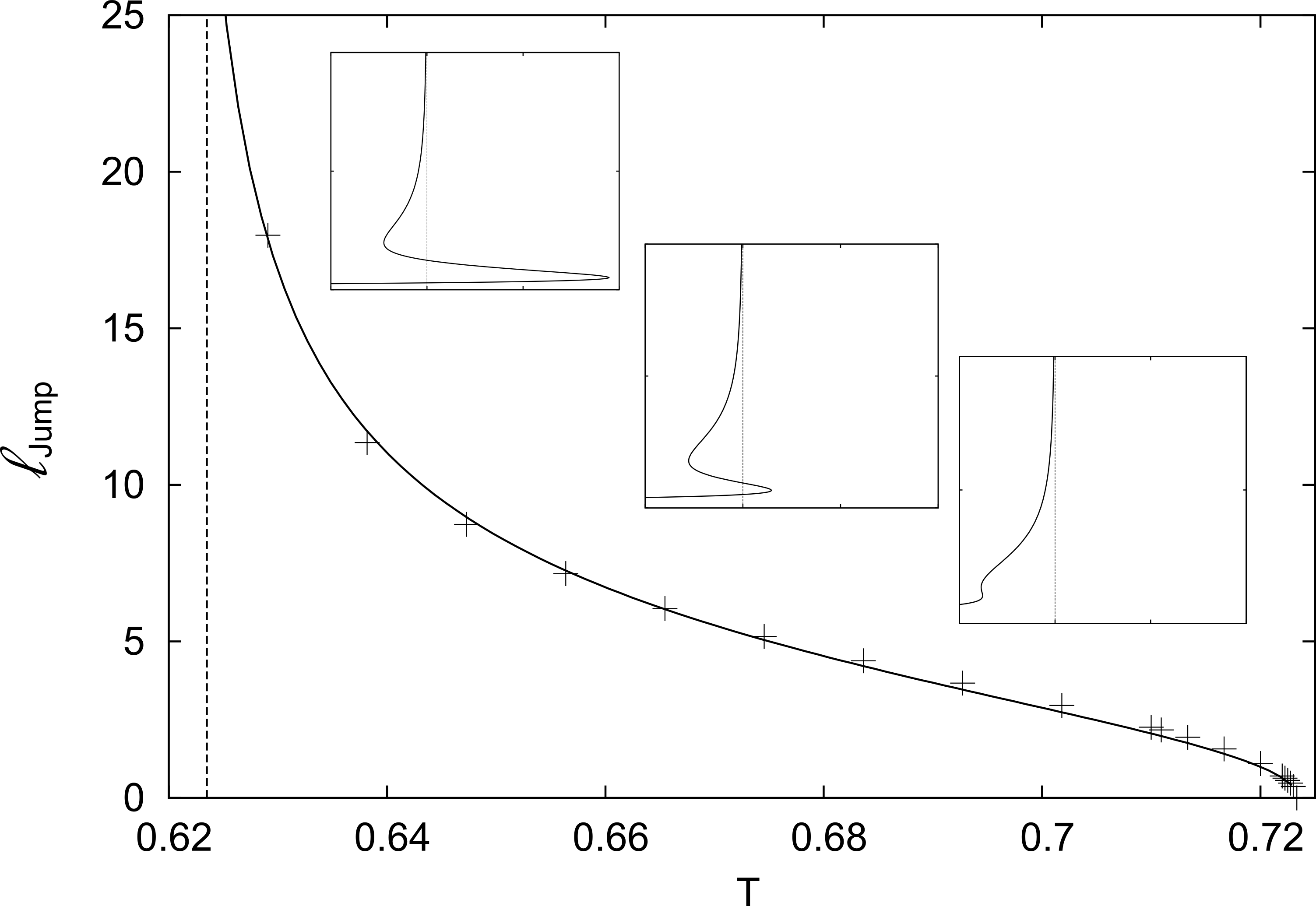}
\caption{Plot of the jump of the film thickness at the prewetting
transition vs. the temperature. The solid  line is a fit to the
equation $\ell_{jump}(T) = C \klamm{\frac{T_{cw} -
T}{T-T_w}}^{1/2}$ with $T_{cw} = 0.7235$ and $T_w = 0.62079$. The resulting coefficient is $C=5.2876$. The small
figures are $\Gamma - \Delta\mu$ diagrams for the temperatures
$T=0.629,0.674$ and $0.72$ on a range of  $[0,20\sigma]$ over
$[-0.04\varepsilon,0.08\varepsilon]$.
 }
\end{center}
\end{figure}

\begin{figure}[p]
\centering
\includegraphics[width=12cm]{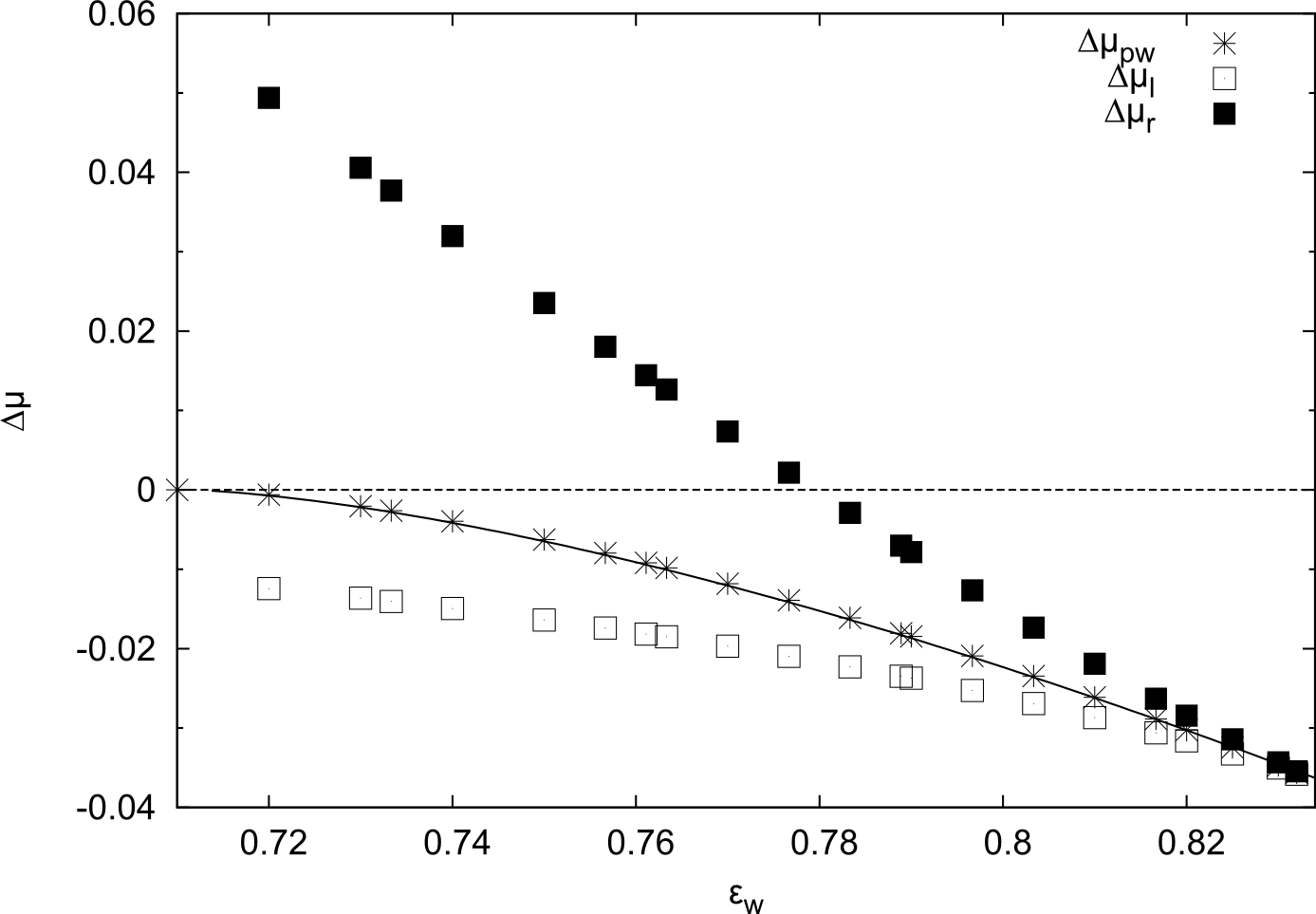}
\caption{Deviations of the chemical potential from the saturation one at prewetting, $\Delta\mu_{pw}$, and at the left and right saddle nodes as a function of the wall parameter $\varepsilon_w$ at temperature $T=0.7$ and for $\sigma_w=1.25$. The solid line is a fit to equation $\Delta \mu_{pw}\klamm{\varepsilon_w} = - C (\varepsilon_w- \varepsilon_{w,w})^{1.5}$, where $\varepsilon_{w,w} = 0.7124$. The resulting coefficient is $C =0.8589$.}
\label{fig:MuNodd_Epsilonw}
\end{figure}

\begin{figure}[p]
\centering
\includegraphics[width=12cm]{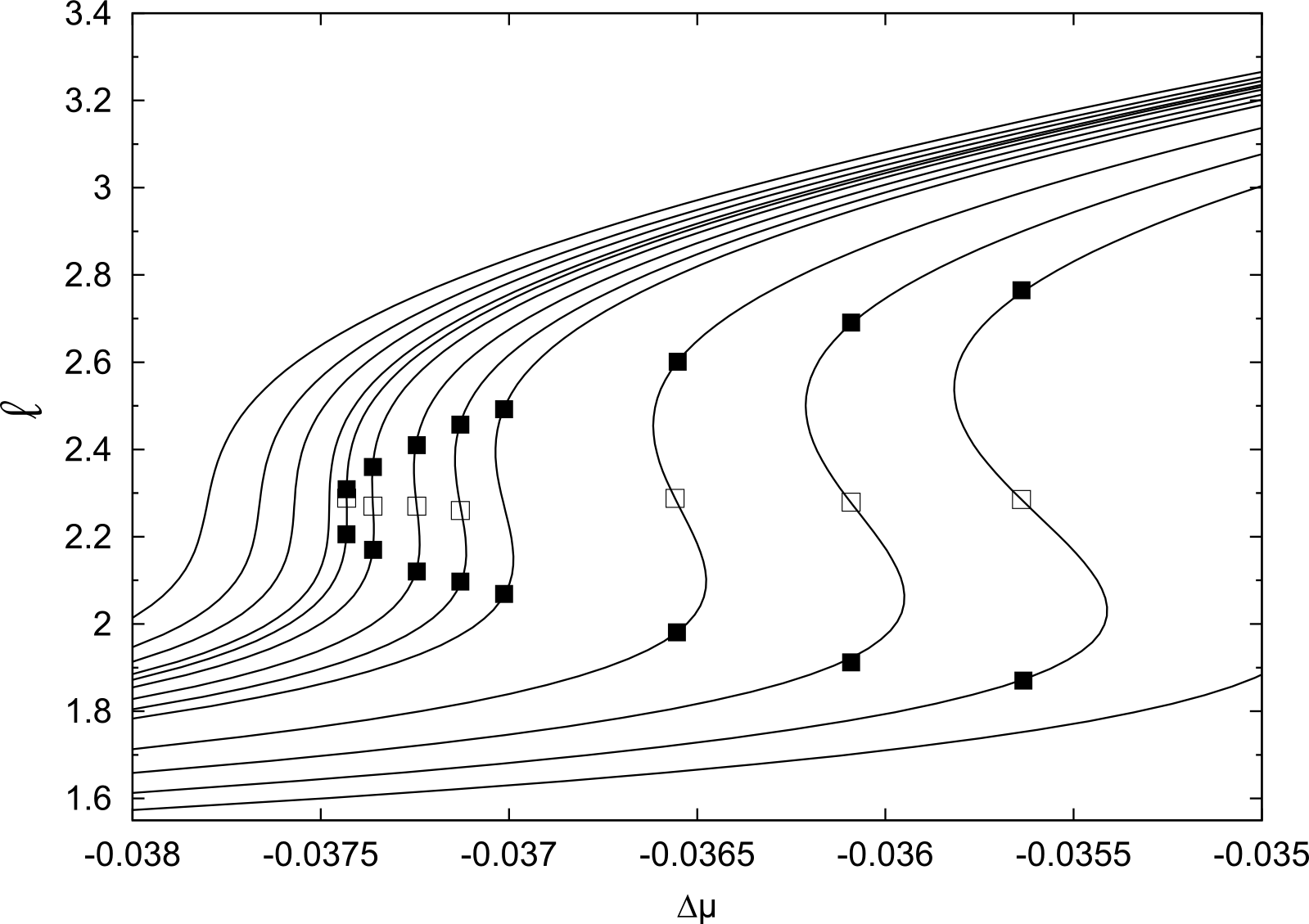}
\caption{Isotherms of the film thickness $\ell$ over the deviation of the chemical potential from saturation for a varying wall parameter $\varepsilon_w$ at the vicinity of the transition from a complete wetting scenario to a prewetting scenario at temperature $T=0.7$ and for $\sigma_w = 1.25$. From left to right, the isotherms correspond to $\varepsilon_w =\{0.8367,~0.8364,~0.8362,~0.8360, 0.8359,~0.8357,~0.8355,~0.8353,~0.8350,~0.8340,~0.8330,~0.8320,~0.8310\}$. The squares are at the prewetting transitions of the respective isotherms, where the black square are at the stable lower and upper branch and the white square is at the unstable branch.}
\label{fig:TransitionPrewettingCompleteWetting}
\end{figure}

\begin{figure}[ht] 
\centering 
\includegraphics{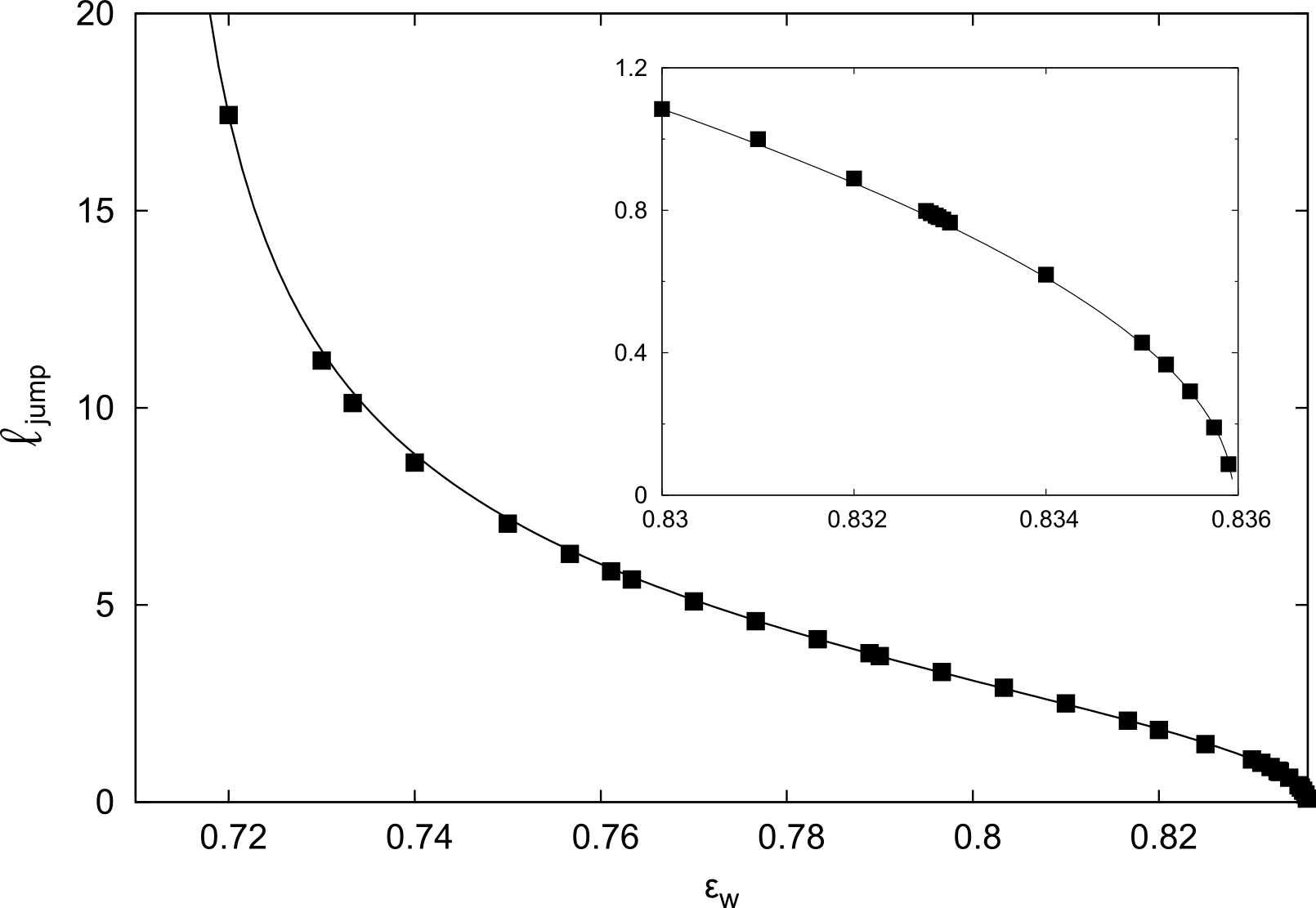}
\caption{Plot of the jump of the film thickness at the prewetting transition vs. the wall parameter $\varepsilon_w$ at temperature $T = 0.7$ and for $\sigma_w = 1.25$. The black squares are the result of numerical computations, whereas the solid line is a fit to the equation $\ell_{jump}(\varepsilon_w) = C\klamm{\frac{\varepsilon_{w,cw}-\varepsilon_w}{\varepsilon_w - \varepsilon_{w,w}}}^{1/2}$ with $\varepsilon_{w,cw} = 0.8359$ and $\varepsilon_{w,w} = 0.7110$. The resulting coefficient is $C = 4.85$.}
\label{fig:JumpPrewettingEpsilonW}
\end{figure}

\subsection{Analytic Prediction}

The SIA (\ref{eq:SharpKink_OmegaEx_4}) is applied on the case of a planar wall $W = \mathbb{R}^2\times \mathbb{R}^-$. The film volume is given by $V_f = \mathbb{R}^2\times [\delta,\ell)$ and the volume of the bulk gas is given by $V_B = \mathbb{R}^2\times [\ell,\infty)$. In this case, it is useful to introduce an excess grand potential per unit area $\gamma_{wall,\infty}$, which corresponds to the surface energy of the wall, as a function of the film thickness. (\ref{eq:SharpKink_OmegaEx_4}) yields
\begin{align}
\gamma_{wall,\infty}^{SIA}(\ell) = 
 - \Delta \mu \Delta n \ell 
+ \gamma_{wl,\infty}^{SIA} + \gamma_{lg,\infty}^{SIA} + \gamma_{B,\infty}^{SIA}(\ell), \label{eq:SharpKink_1D}
\end{align}
where (\ref{eq:SIA_OmegaWL}),(\ref{eq:SIA_OmegaLG}) and (\ref{eq:SharpKing_defBinding}) yield
\begin{align*}
\gamma_{wl,\infty}^{SIA} &= - \frac{n_l^2}{2} \int_{-\infty}^\delta \int_{\delta}^\infty \Phi_{\text{Pla}}\klamm{|z-z'|} dz' dz + n_l \int_\delta^\infty V_{\text{Pla}}(z) dz\\
&= \frac{3}{4} \pi n_l^2 + n_l \int_\delta^\infty V_{\text{Pla}}(z) dz,\\
\gamma_{lg,\infty}^{SIA} &= -\frac{\Delta n^2}{2} \int_{-\infty}^\ell \int_{\ell}^\infty \Phi_{\text{Pla}}\klamm{|z-z'|} dz' dz\\
&= \frac{3}{4} \pi \Delta n^2,\\
\text{and} \qquad
\gamma_{B,\infty}^{SIA}(\ell) &= n_l \Delta n \int_{-\infty}^\delta \int_\ell^\infty \Phi_{\text{Pla}}(|z-z'|) dz' dz 
- \Delta n \int_\ell^\infty V_{\text{Pla}}(z) dz\\
&= \Delta n{ \int_{\ell}^{\infty} \klamm{ n_l \Psi_{\text{Pla}}(z-\delta) - V_{\text{Pla}}(z)} dz} 
\end{align*} 
The only term carrying a $\ell$-dependence is the binding potential $\gamma_{B,\infty}^{SIA}(\ell)$. Hence, minimizing $\gamma_{wall,\infty}^{SIA}$ with respect to $\ell$ yields
\begin{align*}
\Delta n\klamm{- \Delta\mu - n_l \Psi_{\text{Pla}}(\ell-\delta) + V_{\text{Pla}}(\ell)} = 0
\end{align*}
Now, assume that $\ell \gg 1$. The repulsive part of the wall potential (\ref{eq:WallPotential}) as well as the repulsive part of the interaction potential (\ref{eq:1D_Numerics_Vn}) is neglected, as they are of order $\ell^{-9}$ whereas the attractive part is of order $\ell^{-3}$. This leads to 
\begin{align}
- \Delta\mu &=
\frac{2}{3} \pi \klamm{ \frac{\varepsilon_w \sigma_w^6 }{\ell^3} -\frac{n_l}{(\ell - \delta)^3}} + O\klamm{\ell^{-9}}
\label{eq:1D_DMu_SharpInterface}
\end{align}
In order to compare this analytical prediction with the numerical results obtained from the continuation method, the film thickness has to be written as a functional of the density profile $n\klamm{\cdot}$. This is done such that the adsorption $\Gamma$ of $n(\cdot)$ corresponds with the adsorption of the sharp-interface profile $n_\ell^{SIA}$. (\ref{eq:PlanarAdsorption}) yields
\begin{align*}
\Gamma[n\klamm{\cdot}] =& \Gamma[n_\ell^{SIP}\klamm{\cdot}]\\
\Rightarrow \qquad \ell[n\klamm{\cdot}] \defi& \delta + \frac{1}{\Delta n} \Gamma[n\klamm{\cdot}],
\end{align*}
where $\delta$ is the thickness of the liquid-gas interface, usually a value around $0.9$. In Fig.~\ref{fig:1DAsymptotic}, the prediction (\ref{eq:1D_DMu_SharpInterface}) is compared with numerical results showing a very good agreement for large film thicknesses.
\begin{figure}[ht]
\centering
\includegraphics{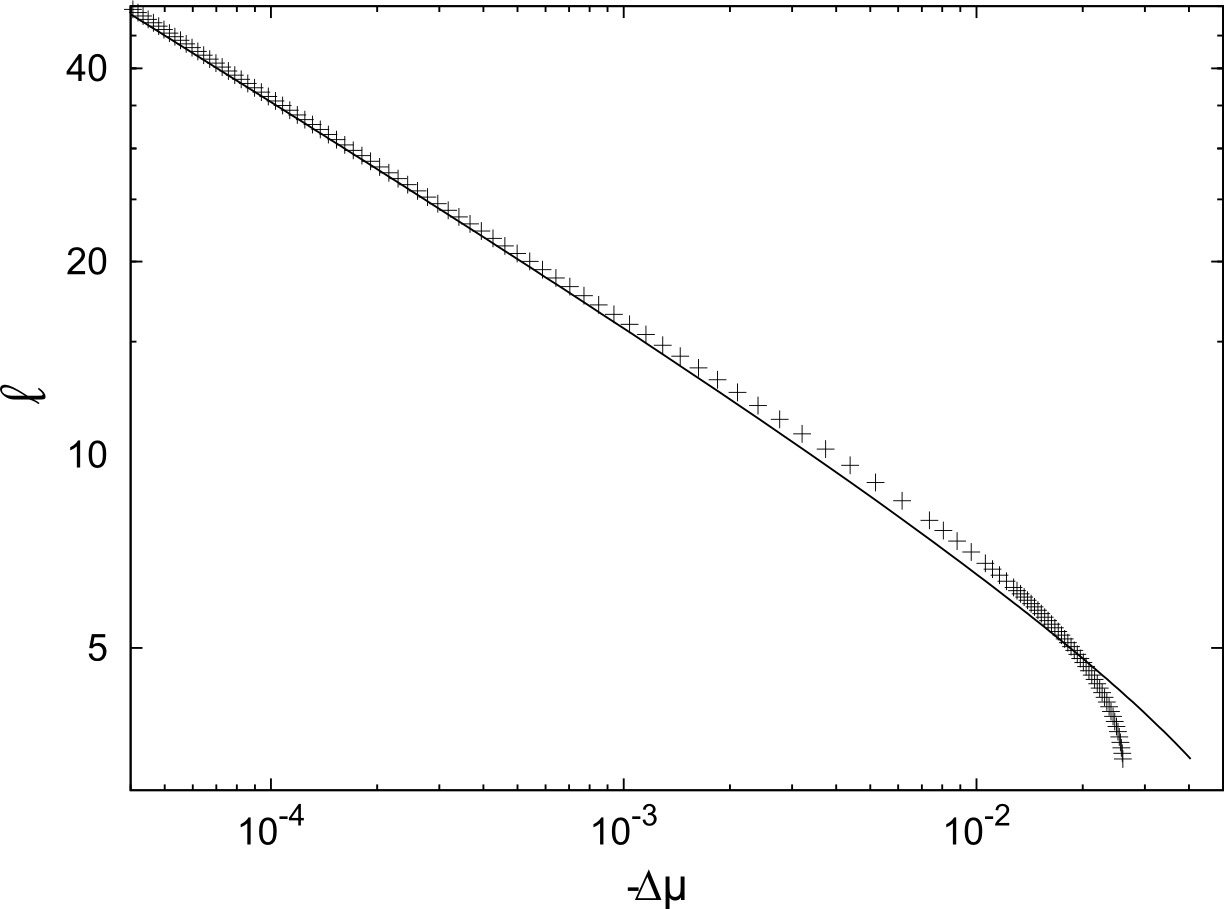}
\caption{Log-log plot of the film thickness $\ell$ as a function of deviation of the chemical potential from saturation $\Delta \mu$ for $T = 0.7$ and wall parameters $\varepsilon_w = 0.8$ and $\sigma_w = 1.25$. The solid line is the analytical prediction in Eq.(\ref{eq:1D_DMu_SharpInterface}) obtained from a SIA, demonstrating that adsorption divergence upon approaching coexistence satisfies a power law $\ell \sim \Delta \mu^{-1/3}$, which is characteristic for the complete wetting regime of long range potentials.}
\label{fig:1DAsymptotic}
\end{figure}

\section{Wetting on a Curved Substrate \label{sec:Wetting_Sphere}}

We now examine the influence of curved substrates on wetting. As
pointed out in Sec.~\ref{sec:DenProf_Sphere} we use a spherical wall
as a model system. Unlike the planar case, \NewStuff{the liquid-gas surface tension now influences} significantly the wetting behavior. This leads
to inaccuracies of the SIA. Hence, we will employ the PFA as an analytical method to obtain equilibrium film thicknesses. Here, we show that based on some simple assumptions for the density profile at the wall-liquid and the liquid-gas interface, one can obtain a simple and exact equation relating the film thickness $\ell$ and the chemical potential $\mu$ with the radius of the substrate $R$.

\subsection{Isotherms for a Spherical Wall \label{sec:SphericalIsotherm}}

\begin{figure}[phtb]
\centering
\includegraphics[width=12cm]{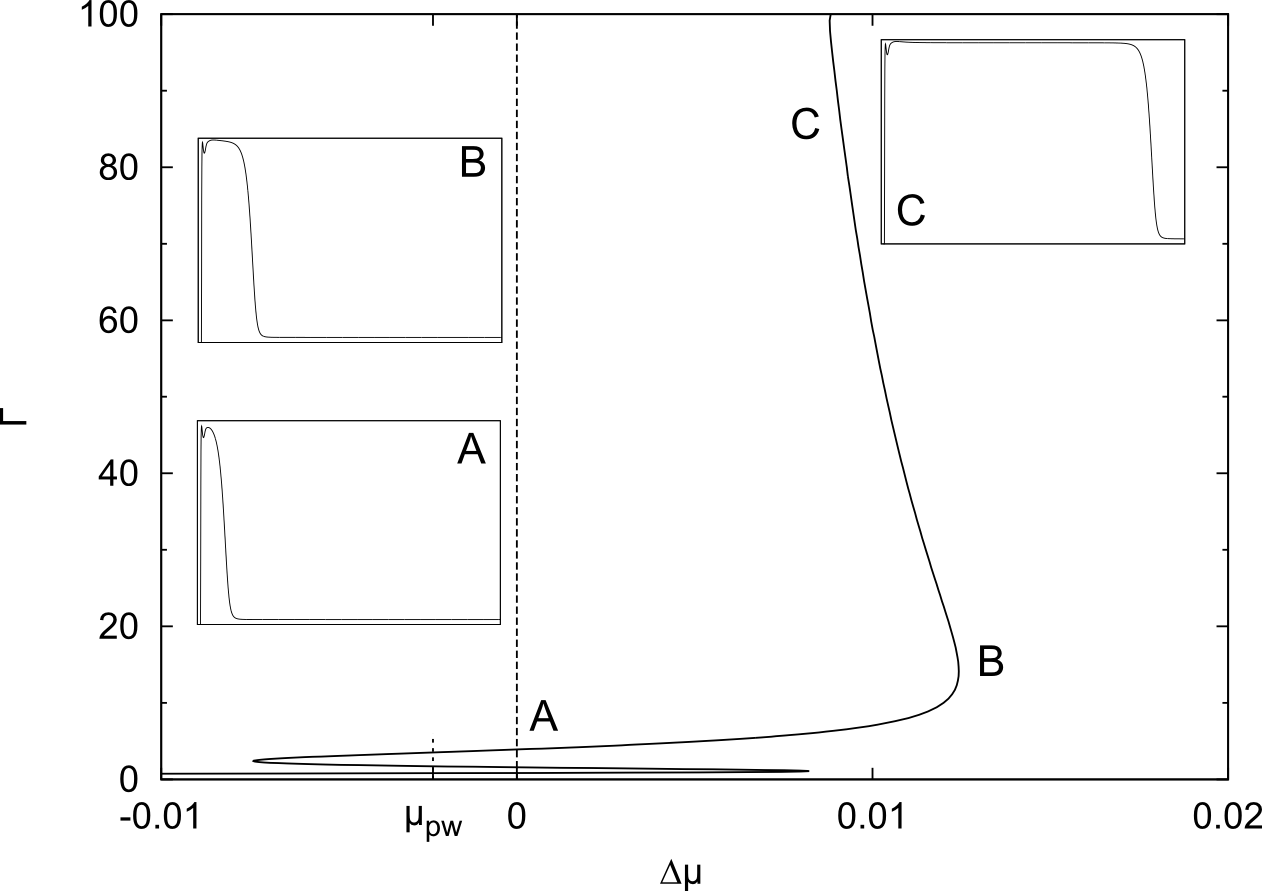}
\caption{$\Gamma-\Delta\mu$ bifurcation diagram for $T = 0.7$ for a sphere with radius $R=100$ and with parameters $\varepsilon_w = 0.8$ and $\sigma_w = 1.25$. The inset subplots show the density $\rho$ over the distance from the wall $(r-R)$ in the ranges $[0,0.7]$ over $[0,80]$. The point A is at saturation $\Delta \mu = 0$. It separates the stable branch to the left from the metastable branch to the right. Point B is at the right turning point, whereas point $C$ is at the unstable branch. $\mu_{pw}$ is at the first prewetting transition.}
\label{fig:ShowOnlySphericalIsotherm}
\end{figure}

\begin{figure}[p]
\centering
\includegraphics[width=12cm]{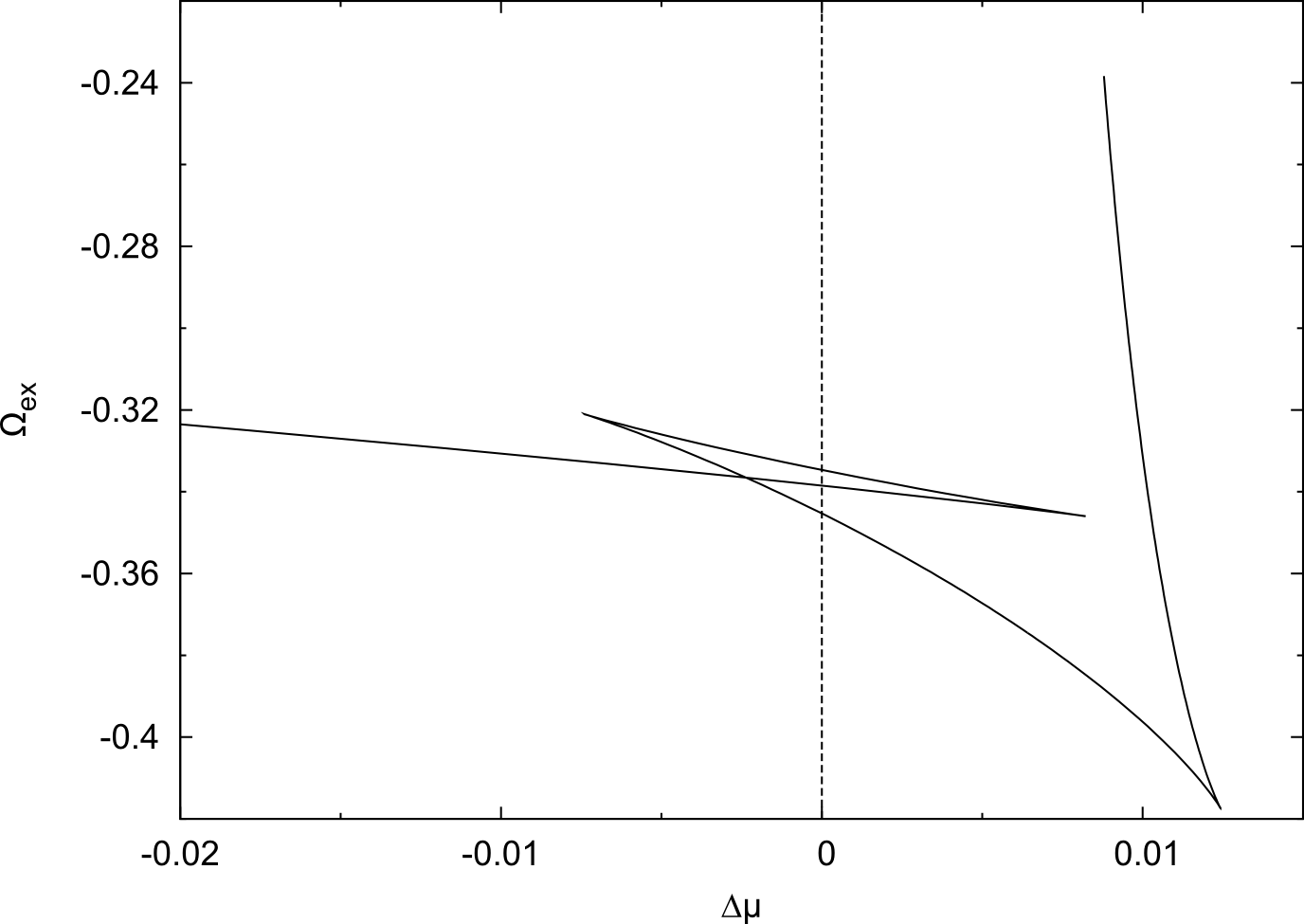}
\caption{Plot of the excess grand potential $\Omega_{ex}$ over the deviation of the chemical potential from saturation, $\Delta \mu$, for the wetting transition in Fig.~\ref{fig:ShowOnlySphericalIsotherm}}
\label{fig:OmegaMuSpherical}
\end{figure}
\NewStuff{
Fig. \ref{fig:ShowOnlySphericalIsotherm} depicts a typical isotherm for a prewetting situation on a spherical wall. For low values of the adsorption $\Gamma$, there is a multi-valued S-type curve similar to the planar case in Fig.~\ref{fig:IsothermProfiles}. This is the signature of a first-order wetting transition at $\mu_{pw}$ which separates the lower stable branch for a very thin film at $\mu < \mu_{pw}$ from the upper stable branch at $\mu_{pw} < \mu < \mu_{sat}$. 

The most striking difference to the isotherm for a planar wall is, that the isotherm crosses the saturation line in A. For positive values of $\Delta \mu$, the bulk liquid phase is more stable than a thin liquid film in a bulk gas phase, which means that the isotherm to the right of the saturation line is not stable and consequently, a second
first-order wetting transition has to take place in A. We conclude that, unlike the planar prewetting case where the film thickness goes smoothly to infinity as saturation is approached, the spherical case isotherm exhibits a maximal film thickness $\ell^\ast$. 
}Physically, this can be
explained by the fact that the surface of the liquid-gas interface grows with
increasing film thickness. Hence, the surface energy of the liquid
gas interface competes with the energy necessary to \NewStuff{increase} the
liquid film. As a result, one obtains a maximal film thickness $l^\ast$. 

As we shall demonstrate in Sec.~\ref{sec:SphericalAsymptoticFT}, the
unstable branch above point $B$ in
Fig.~\ref{fig:ShowOnlySphericalIsotherm} approaches saturation
as slowly as $\Delta \mu^{-1}$. Remark that this branch approaches saturation from the right. This means that the excess grand potential increases (see also Sec.~\ref{sec:MaxwellConstruction}), which leads to a structural change in the plot of the excess grand potential over the chemical potential as in Fig.~\ref{fig:OmegaMuSpherical}, when compared to the planar case in Fig.~\ref{fig:OmegaMuIsotherm}.

The appearance of a maximal film thickness has a direct impact on the nature of  wetting transitions as shown in Fig.~\ref{fig:CompletePartialPre} for a planar wall. There, we have compared the isotherms for a complete wetting scenario with and without a prewetting transition and a partial wetting scenario. The corresponding isotherms on a spherical substrate for a complete wetting and a pre-wetting scenario have an additional first order wetting transition at $\mu_{sat}$ ( see Fig.~\ref{fig:CompletePrePartialSph}). In the pre-wetting scenario, there are two first order wetting transitions (Fig. \ref{plot:PreWettingSph}). In Fig.~\ref{plot:CompleteWettingSph}, the analogous scenario to a complete wetting scenario in the planar case is shown. Here, the adsorption does not go to infinity as $\mu \to \mu_{sat}^-$, but instead is limited by a maximal film thickness $\ell^\ast$. 
\begin{figure}[p]
\centering
\subfigure[Partial Wetting]{
\includegraphics[width=8cm]{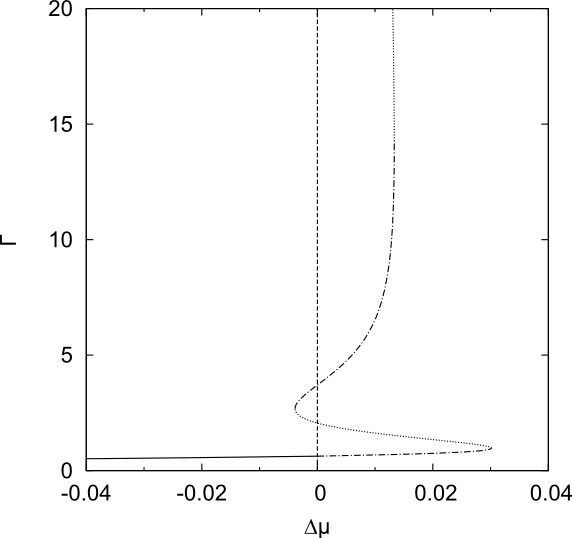}
\label{plot:PartialWettingSph}
}
\subfigure[Prewetting Transition]{
\includegraphics[width=8cm]{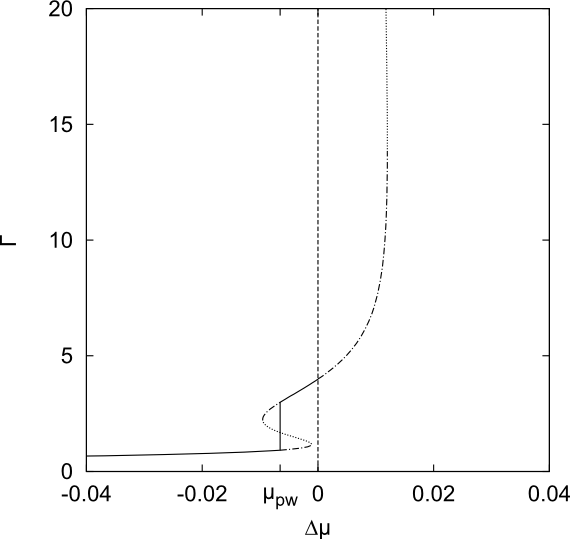}
\label{plot:PreWettingSph}
}
\subfigure[Pseudo-Complete Wetting]{
\includegraphics[width=8cm]{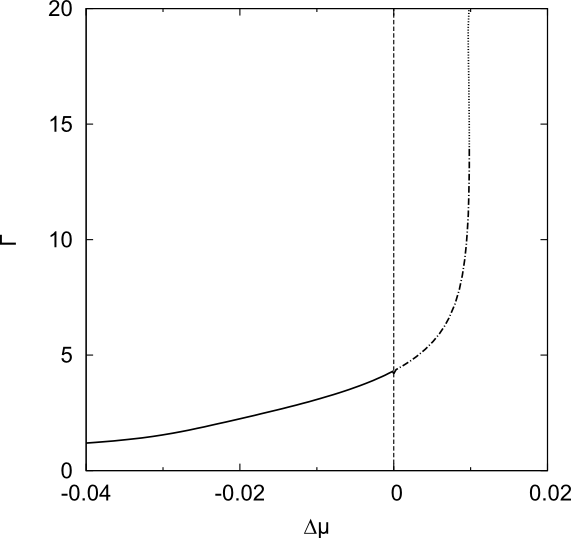}
\label{plot:CompleteWettingSph}
}
\caption{Plots of the isotherms of adsorption $\Gamma$ over the deviation of the chemical potential from saturation $\Delta \mu$ for a spherical substrate with radius $R = 100$ for (a) a partial wetting scenario at $T = 0.68$, (b) a prewetting scenario at $T = 0.71$ and (c) a pseudo-complete wetting scenario at $T = 0.76$. The wall parameters are $\varepsilon_w = 0.8$ and $\sigma_w = 1.25$. The solid lines are the stable branches of the isotherms. The metastable branches are the dot-dash lines whereas unstable branches are drawn as dotted lines.}
\label{fig:CompletePrePartialSph}
\end{figure}

\NewStuff{In Figure~\ref{fig:SphericalIsothermProfiles} typical isotherms
for a prewetting situation in the planar and the spherical case are compared}. The most striking property of the spherical case isotherm is that it is shifted to the right compared to its counterpart in the planar case. In
Sec.~\ref{sec:SphAsymptoticLargeR} we shall demonstrate analytically
that for a large radius of the sphere $R$ and large film thickness
$\ell$ such that $R \gg \ell \gg 1$, this shift corresponds
precisely to the Laplace pressure $2\gamma_{lg,\infty}/R$. 

\NewStuff{ In the subplots of Fig. \ref{fig:SphericalIsothermProfiles}, the density profiles of the spherical and the planar case are compared. For the points B and C, which are at the same film thicknesses as B' and C', respectively, the density profiles are practically indistinguishable. The subplot $A,A'$ shows the density profiles at the upper branch of the prewetting transition at $\mu_{pw}^+$. The differences in the profiles indicate that the film thickness at the prewetting transition changes with the curvature of the substrate.}

\begin{figure}[hbt]
\centering
\includegraphics[width=10cm]{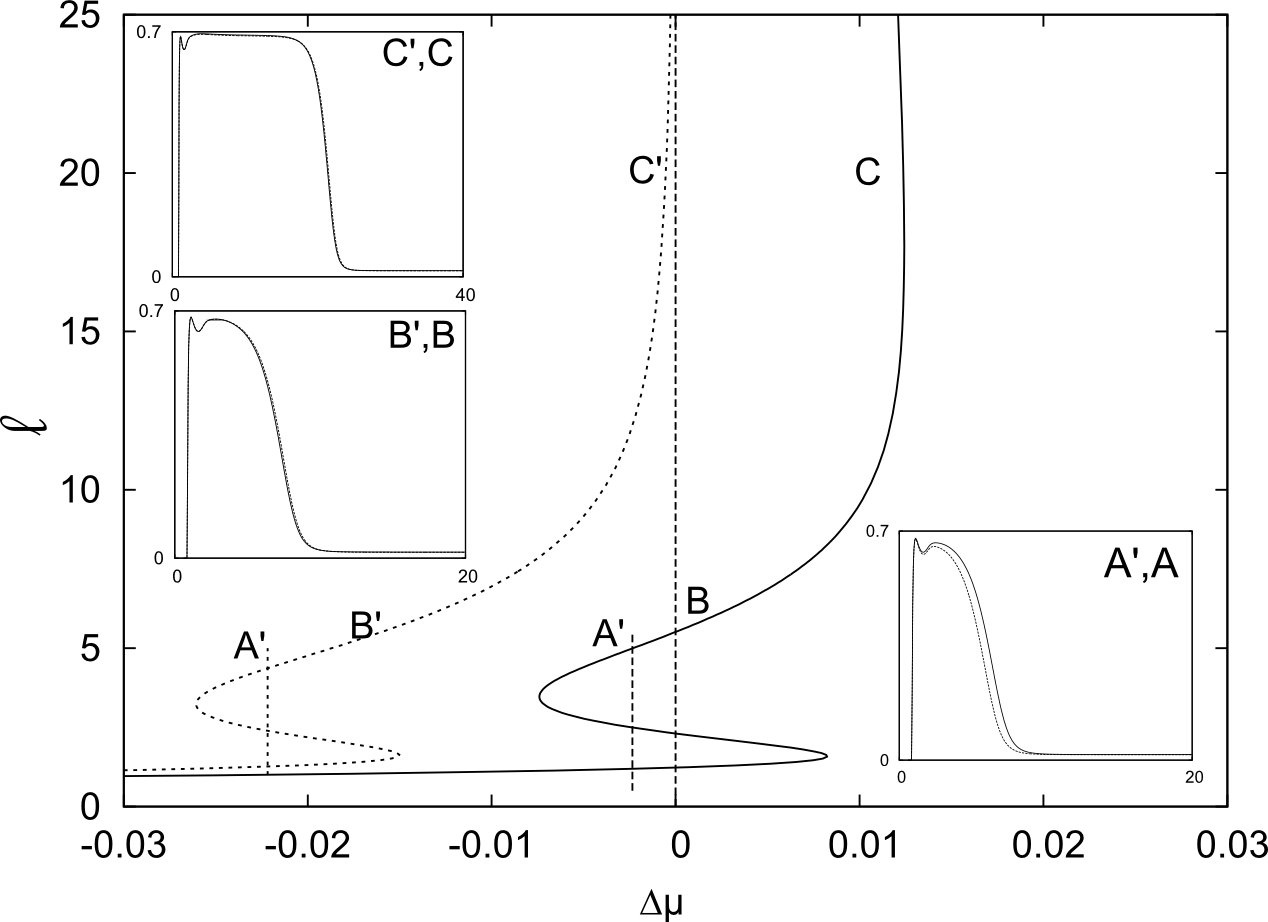}
\caption{Isotherms and density profiles for a planar
wall (dashed lines) and a sphere with $R = 100$ (solid lines)
at $T = 0.7$ and with wall parameters $\varepsilon_w
= 0.8$ and $\sigma_w = 1.25$. In order to
compare the planar to the spherical case, the film thickness
instead of adsorption is used as a measure. The subplots in the inset show
the dimensionless density $n$ as a function of the distance from the
wall $z$ and $(r-R)$ for the planar and the spherical
cases, respectively. The points $A$ and $A'$ are at the prewetting
transitions. Points $B,B'$ and $C,C'$ correspond to the same film
thickness; for these points, the planar and the spherical density
profiles are practically indistinguishable. $B'$ is at saturation whereas
$C$ is chosen such that the film thickness $\ell$ is $20$.
\label{fig:SphericalIsothermProfiles} }
\end{figure}

\subsection{Analytic Prediction}

The SIA has been one of the most used techniques to analyze wetting on substrates analytically. For chemical potentials close to saturation, this approximation leads to a simplified formulation of the excess grand potential as a function of the wall-liquid and the liquid-gas surface tension, as well a binding potential $\Omega_B$ (see also Eq. (\ref{eq:SharpKink_OmegaEx_4})). In order to show that the SIA is an accurate method to predict the asymptotic behavior of the isotherm as it approaches saturation $\mu\to \mu_{sat}^-$, Napi\'{o}rkowski and Dietrich~\cite{DietrichNapiorkowski_BulkCorrelation} calculated the excess grand potential of a profile for which the density is assumed to be everywhere constant except in the liquid-gas interface: 
\begin{align}
n(z) =  \left\{\begin{array}{ll}
0 & \text{ if } z \leq \delta \\
n_{l} & \text{ if } \delta < z < \ell -\kappa/2 \\
n_{lg}(z-\ell) & \text{ if } \ell-\kappa/2 \leq z \leq \ell + \kappa/2 \\
n_{g} & \text{ if } \ell + \kappa/2 < z
\end{array}\right.,
\label{eq:SoftInterfaceApproximation}
\end{align}
where $\kappa$ is the width of the liquid-gas interface,
$n_{lg}(z)$ is its shape and $\delta$ is the width of the
wall-liquid interface. This test function is similar to the PFA introduced in Sec.~\ref{sec:PFA}. Similar to (\ref{eq:OmegaEx_PFA}), Napi\'{o}rkowski and Dietrich wrote the excess grand potential as the sum of one term considering the deviation of the chemical potential from saturation, the liquid-gas and the wall-liquid surface tension as well as a binding potential. In the planar case, the liquid-gas surface tension does not depend on the film thickness $\ell$, which leads to the conclusion:
\begin{quote}
"At coexistence, the binding potential [..] carries the only $\ell$ dependence of [$\Omega_{ex}$] and contains the information about wetting transitions."\cite{DietrichNapiorkowski_BulkCorrelation}
\end{quote}
In the limit of a large film thickness $\ell \to \infty$, the binding potential of the PFA corresponds with the binding potential of the SIA in the planar case. Consequently, the wetting behavior close to saturation in the planar case can be predicted by the SIA.

However, this is not true for the spherical case. As we shall demonstrate shortly, the main reason for this is that the term in the excess grand potential including the liquid-gas surface tension depends on the film thickness $\ell$. But, as it is shown in Sec.~\ref{sec:DensityProfiles_PlanarWall}, there is a huge discrepancy between the exact surface tension and the sharp-interface surface tension. Thus leads to significant errors of the SIA in the spherical case.

In order to investigate the influence of the smooth interface on the wetting behavior, we make use of the PFA for the spherical case. The test function (\ref{eq:FiniteDifference_Splitting}), with a spherical wall $W = \{{\bf r}\in\mathbb{R}^3: |{\bf r}| < R\}$, the wall-liquid interface volume $V_{wl} = \{{\bf r}\in\mathbb{R}^3: R\leq |{\bf r}| < R+\delta\}$, the film volume $V_f = \{{\bf r}\in\mathbb{R}^3:  R+\delta \leq |{\bf r}| < R + \ell - \kappa/2\}$, the volume of the liquid-gas interface  $V_f = \{{\bf r}\in\mathbb{R}^3:   R + \ell - \kappa/2 \leq |{\bf r}| < R + \ell + \kappa/2\}$ and the bulk volume  $V_B = \{{\bf r}\in\mathbb{R}^3:   R + \ell + \kappa/2 \leq |{\bf r}|\}$ yields:
\begin{align*}
n^{PFA}(r) =
\left\{
\begin{array}{lll}
0 & \text{ if } & r \in [0,R)\\
n_{wl}(r-R) & \text{ if } & r \in [R,R+\delta)\\
n_{l} & \text{ if } & r \in [R+\delta,R+\ell-\frac{\kappa}{2})\\
n_{lg}(r-R-\ell) & \text{ if } & r \in [R+\ell-\frac{\kappa}{2},R+\ell+\frac{\kappa}{2})\\
n_{g} & \text{ if } & r \in  [R+\ell+\frac{\kappa}{2},\infty) 
\end{array}
\right. .
\end{align*}
We assume that the density at the wall-liquid interface is bounded from above by the liquid bulk density: $n_{wl}(r) < n_l$. Furthermore, we assume that $n_{lg}(r)$ is a monotonically decreasing function. $\ell$ is such that it defines the position of the Gibbs-dividing surface of the liquid-gas interface:
\begin{align*}
4 \pi \int_{R+\ell -\kappa/2}^{R+\ell} r^2\klamm{ n^{PFA}(r) - n_l } dr
+
4 \pi \int_{R+\ell}^{R+\ell + \kappa/2} r^2\klamm{ n^{PFA}(r) - n_g } dr
=0.
\end{align*}
In this case, the excess grand potential (\ref{eq:OmegaEx_PFA}) yields
\begin{align*}
\Omega_{ex,sph}^{PFA}(R,\ell,\delta,\kappa,\{n_{lg}(r-R-\ell)\},\{n_{wl}(r-R)\}) =& 
- \frac{4}{3} \pi \Delta \mu \Delta n \klamm{ \klamm{R+\ell}^3 - R^3 } + \notag\\
&+ \Omega_{lg,sph}^{PFA}(R + \ell,\kappa,\{n_{lg}(r-R-\ell)\})+\\
&+ \Omega_{wl,sph}^{PFA}(R,\delta,\{n_{wl}(r-R)\}) + \\
& + \Omega_{B,sph}^{PFA}\klamm{R,\ell,\delta,\kappa,\{n_{wl}(r-R)\},\{n_{lg}(r-R-\ell)\}} ,
\end{align*}
where the liquid-gas excess grand potential (\ref{eq:PFA_OmegaLG}) is 
\begin{align*}
\Omega_{lg,sph}^{PFA}(R,\kappa,\{n_{lg}(r)\}) \defi& - 4\pi\int_{R- \kappa/2}^{R} r^2\klamm{ p\klamm{n_{lg}(r)} - p\klamm{n_{l}}} dr
 -4\pi\int_R^{R+ \kappa/2} r^2\klamm{p\klamm{n_{lg}(r)} - p\klamm{n_{g}}} dr-\\
&- 2\pi\klamm{ \int_{R-\kappa/2}^{R+\kappa/2} \klamm{ \klamm{ n_{lg}(r) - n_{l} }^2 \Psi_{in,R-\kappa/2}\klamm{r} + \klamm{ n_{lg}(r) - n_{g} }^2 \Psi_{out,R+\kappa/2}\klamm{r}}r^2 dr
}+\\
&+ 2\pi 
\int_{R-\kappa/2}^{R+\kappa/2} 
\int_{R-\kappa/2}^{R+\kappa/2}
r^2 n_{lg}(r)\klamm{ n_{lg}(r') - n_{lg}(r) } \Phi_{Sph}\klamm{r,r'} 
dr' dr-\\
&- 2\pi(n_{l} - n_{g})^2
\int_{0}^{R-\kappa/2} 
\int_{R+\kappa/2}^{\infty}
r^2\Phi_{sph}\klamm{r,r'} dr' dr,
\end{align*}
and the wall-liquid excess grand potential (\ref{eq:PFA_OmegaWL}) is 
\begin{align*}
\Omega_{wl,sph}^{PFA}(R,\delta,\{n_{wl}(r)\}) \defi& 
- 4\pi \int_{R}^{R+\delta} r^2 \klamm{p\klamm{n_{wl}(r)} - p\klamm{n_{l}}} dr-\\
&- 2\pi \int_R^{R+\delta} \klamm{\klamm{ n_{wl}(r) - n_{l} }^2 \Psi_{out,R+\delta}\klamm{r} + n_{wl}(r)^2 \Psi_{in,R}\klamm{r}}r^2 dr+\\
&+ 2\pi 
\int_{R}^{R+\delta} 
\int_{R}^{R+\delta}
n_{wl}(r)r^2\klamm{ n_{wl}(r') - n_{wl}(r) } \Phi_{Sph}\klamm{r,r'} 
dr' dr-\\
&- 2\pi n_{l}^2
\int_{0}^{R} 
\int_{R+\delta}^{\infty}
r^2 \Phi_{Sph}\klamm{r,r'} dr' dr+\\
&+
4\pi \int_R^{R+\delta} V_{sph,R}(r) r^2 n_{wl}(r) dr +  4\pi n_{l} \int_{R+\delta}^\infty r^2 V_{sph,R}(r) dr.
\end{align*}
The binding potential (\ref{eq:OmegaB_PFA}) in the spherical case yields
\begin{align}
\Omega_{B,sph}^{PFA}\klamm{R,\ell,\delta,\kappa,\{n_{wl}(r)\},\{n_{lg}(r)\}} \defi&
 4 \pi \int_{0}^{R+\delta} \int_{R+\ell-\kappa/2}^\infty \klamm{n_{l} - n(r') }\klamm{n_{l} - n(r)} \Phi_{Sph}\klamm{r,r'} r^2 dr'dr-\notag\\
&- 4\pi \int_{R+\ell-\kappa / 2}^{\infty} \klamm{ n_{l} - n(r)}V_{sph,R}(r) r^2 dr.
\label{eq:OmegaBsph_PFA}
\end{align}

For the liquid-gas excess grand potential $\Omega_{lg,sph}^{PFA}(R,\kappa,\{n_{lg}(\cdot)\})$ it is assumed that the density-profile of the liquid-gas interface $n_{lg}(r)$ corresponds to the density profile of a drop of radius $R$. This allows to approximate the excess grand potential by the liquid-gas surface tension $\gamma_{lg,R}$ of the drop times its surface:
\begin{align}
\Omega_{lg,ex}(R,\kappa,\{n_{lg}(r)\})
=
4\pi R^2 \gamma_{lg,R}.
\end{align}
We now minimize the excess grand potential $\Omega_{ex,sph}^{PFA}$ with respect to the film thickness $\ell$. Doing this, we assume that the shape of the wall-liquid interface, given by $\delta$ and $n_{wl}(r)$, is constant. This yields
\begin{align}
\dif{\Omega_{ex,sph}^{PFA}}{\ell} = & - 4\pi \Delta \mu \Delta n \klamm{ R  +  \ell}^2
+ 8\pi ( R+ \ell) \gamma_{lg, R+ \ell} + 4\pi ( R+ \ell)^2 \dif{\gamma_{lg, R+ \ell}}{\ell}+\label{eq:DerivativeOmegaExcessFilmThicknessSpherical}\\
& + 4\pi \int_0^{ R+\delta} \klamm{n_{l} -
n( r')} \cdot h( r') d r'
-4\pi \int_{ R +  \ell -\kappa/2}^{ R + 
l+\kappa/2}  n_{lg} '( r- R- \ell) 
V_{sph,R}( r)  r^2 d r, \notag
\end{align}
where $h( r')$ is defined by:
\begin{align*}
h( r') \defi \int_{ R+ \ell -\kappa /2}^{ R+
\ell +\kappa /2} n_{lg} '( r- R-
\ell)\Phi_{Sph}( r, r')  r^2 d r.
\end{align*}
We first consider how expression
(\ref{eq:DerivativeOmegaExcessFilmThicknessSpherical}) depends on
the profile of the wall-liquid interface. Indeed, the fourth term,
which is the derivative of the binding potential $\Omega_{B,sph}^{PFA}$ defined in
Eq.~(\ref{eq:OmegaBsph_PFA}), is the only term having
such a dependence. Notice that the fourth term is an integral over
the product of the positive terms $( n_{l} -
n( r))$ and $h( r')$. It can hence be
simplified by replacing $n_{l} - n( r)$ by
its upper limit, $n_{l}$, and by reducing the domain of
integration from $[0, R+\delta]$ to $[0,
R+\delta^\ast]$ for some $\delta^\ast \in
(0,\delta)$ such that:
\begin{align}
\int_0^{ R+\delta} \klamm{n_{l} -
n( r)}h( r') d r' = n_{l}
\int_0^{ R + \delta^\ast} h( r') d r'.
\label{eq:DefiningEquation_DeltaStar}
\end{align}
Effectively, with this expression we replaced the wall-liquid
interface $n_{wl}( r)$ by the auxiliary parameter
$\delta^\ast$. The fourth term of Eq.~(\ref{eq:DerivativeOmegaExcessFilmThicknessSpherical}) can be written as
\begin{align*}
4 \pi n_{l}
\int_0^{ R + \delta^\ast} 
\int_{ R+ \ell -\kappa /2}^{ R+
\ell +\kappa /2} n_{lg} '( r- R-
\ell)\Phi_{Sph}( r, r')  r^2 d r
 d r' = 
4 \pi n_{l}
\int_{ R+ \ell -\kappa /2}^{ R+
\ell +\kappa /2} n_{lg} '( r- R-
\ell)\Psi_{in,R+\delta^\ast}( r, r')  r^2 d r,
\label{eq:WettingSphere_mid1}
\end{align*}
where $\Psi_{in, R+\delta^\ast}( r)$ is defined in
Eq.~(\ref{eq:Sph_PsiIn}).
The remaining terms of
Eq.~(\ref{eq:DerivativeOmegaExcessFilmThicknessSpherical}), in
particular (\ref{eq:WettingSphere_mid1}) and the fifth term, also involve the density profile of
the liquid-gas interface, $n_{lg}( r)$. They are both
of the following form:
\begin{align}
\int_{-\kappa/2}^{\kappa/2} n_{lg}'( r)( r)
f_{I,V}( R + \ell + r) d r,
\end{align}
where $f_I \defi n_{l}  r^2
\Psi_{in, R+\delta^\ast}( r)$ and $f_V \defi 
V_{sph, R}( r) r^2$ (see Eq.~(\ref{eq:Sph_WallPot})). Both quantities $f_{I,V}$ are
negative. Consequently, the mean value theorem for integrals can be
employed such that
\begin{align}
\int_{-\kappa/2}^{\kappa/2} n_{lg}'( r)
f_{I,V}( R +  \ell+ r) d r
= - \Delta n f_{I,V}( R +  \ell + \xi_{I,V}),
\label{eq:WettingSphere_IntermediateResult}
\end{align}
for some  $\xi_{I,V} \in (-\kappa/2,\kappa/2)$, where we made
use of the fact that $\int  n_{lg}'( r)  d r =
- \Delta  n$. Here, the shape of the liquid-gas interface
was replaced by the auxiliary parameters $\xi_I$ and $\xi_V$.

Now, insert (\ref{eq:WettingSphere_IntermediateResult}) into
expression (\ref{eq:DerivativeOmegaExcessFilmThicknessSpherical})
and set it to zero. Division by $4\pi \Delta n (R+ l)^2$ leads to the
following relation between the chemical potential $\Delta \mu$
and the film thickness $\ell$:
\begin{align}
0 =& -\Delta \mu  + \frac{2 \gamma_{lg, R+ \ell}}{\Delta n ( R+ \ell)} + \frac{1}{\Delta n} \dif{\gamma_{lg, R+ \ell}}{\ell}-\notag\\
 &\qquad - n_{l} \Psi_{in,
R+\delta^\ast}( R +  \ell + \xi_I)
\klamm{1+\frac{\xi_I}{ R+ \ell}}^2 +  V_{sph, R}( R +  \ell +
\xi_V)\klamm{1+\frac{\xi_V}{ R+ \ell}}^2.
\label{eq:WettingSph_ExactEq}
\end{align}
In order to compare this analytical prediction with numerical results, the film thickness $\ell$ is written as a functional of the density distribution $n(r)$ such that the adsorption corresponds with the adsorption of a film in the SIA with film thickness $\ell$:
\begin{align*}
\Gamma[n(\cdot)] &= \Delta n \int_{R+\delta}^{R+\ell} \klamm{\frac{r}{R}}^2 dr\\
\Rightarrow \qquad
\ell[n(\cdot)] &= \klamm{ \frac{3R^2 \Gamma[n(\cdot)]}{\pi \Delta n} + \klamm{R + \delta}^3 }^{1/3} - R
\end{align*}

\subsubsection{Asymptotic Behavior for Substrates with Large Radii \label{sec:SphAsymptoticLargeR}}

By fixing the film thickness $\ell$ and expanding expression (\ref{eq:WettingSph_ExactEq}) for large radius of the
wall, $R \to \infty$, we expect to regain the planar $\ell^{-3}$- law (\ref{eq:1D_DMu_SharpInterface}) 
from SIA. The first and
second term of (\ref{eq:WettingSph_ExactEq}) are expanded
for $R \to \infty$ to yield,
\begin{align}
\frac{2\gamma_{lg,R+l}}{R + \ell} + \dif{\gamma_{lg,R+\ell}}{\ell} =
\frac{2\gamma_{lg,\infty}}{R} + {O} \klamm{\frac{1}{R^2}},
\label{eq:ExpandSurfaceTension}
\end{align}
where we made use of the expansion in Eq.~(\ref{eq:SurfaceTensionSpherical_Tolman})
for the spherical surface tension, $\gamma_{lg,R+l}$. Subsequently,
the third and fourth term of
Eq.~(\ref{eq:WettingSph_ExactEq}) are expanded for
$\ell \to \infty$. We will do the computations for the fourth term involving the external potential (\ref{eq:Sph_WallPot}) first. We assume that $R^{-1} \ll 1$ and approximate 
\begin{align*}
\klamm{1+\frac{\xi_V}{ R+ \ell}}^2 \approx 1.
\end{align*}
For $\tilde r = \ell + \xi_V$, we get
\begin{align}
V_{sph,R}(R + \tilde r) = 
\varepsilon_w \sigma_w^3 \pi 
\frac{\sigma_w}{3(\tilde r + R)}\klamm{
\frac{\sigma_w^8}{30}\klamm{
\frac{\tilde r+ 10R}{(\tilde r+ 2 R)^9} - \frac{\tilde r-8R}{\tilde r^9}
}+ 
\sigma_w^2\klamm{
\frac{\tilde r-2R}{\tilde r^3} - \frac{\tilde r+4R}{(\tilde r+2 R)^3}
}} 
\notag
\end{align}
The first terms is of order $O\klamm{R^{-9}}$ and is neglected in the following. Rearranging the equation above yields
\begin{align}
V_{sph,R}(R + \tilde r) &\approx 
\frac{\varepsilon_w \sigma_w^4 \pi}{ 3 \tilde r^3}
\klamm{
\frac{\sigma_w^8}{30\tilde r^6}
\frac{8-\frac{\tilde r}{R}}{1 + \frac{\tilde r}{R}}
- \frac{2\sigma_w^2}{\klamm{1+\frac{\tilde r}{2R}}^3}
}.
\notag
\end{align}
Now, we assume that the fraction $\frac{\tilde r}{R}$ is small and expand around zero. This means that we assume that the radius of the substrate is significantly larger than the film thickness. This yields
\begin{align*}
\frac{8-\frac{\tilde r}{R}}{1 + \frac{\tilde r}{R}} = 8 + O\klamm{\frac{\tilde r}{R}}\\
\frac{1}{\klamm{1+\frac{\tilde r}{2R}}^3} = 1 + O\klamm{\frac{\tilde r}{R}}
\end{align*}
Hence, we obtain
\begin{align}
V_{sph,R}(R + \tilde r) &\approx 
\frac{\varepsilon_w \sigma_w^6 \pi}{ 3 \tilde r^3}
\klamm{
-2 + \frac{8\sigma_w^6}{\tilde r^6}  + O\klamm{\frac{\tilde r}{R}}
}.
\notag
\end{align}
We neglect terms of order
$\frac{\tilde r}{R}$. In one further step, $\tilde r$ is supposed to be large, such that $\tilde r^{-6} \ll 1$. This leads to
\begin{align}
V_{sph,R}(R + \tilde r) &\approx 
- \frac{2\varepsilon_w \sigma_w^6 \pi}{ 3 \tilde r^3},
\notag
\end{align}
which finally yields
\begin{align*}
 V_{sph, R}( R +  \ell +
\xi_V)
\approx
- \frac{2}{3} \varepsilon_w \sigma_w^6 \pi
\frac{1}{\ell^3}
\end{align*}
where we made the assumption $\xi_V \ll \ell$ (recall that $\xi_V \in \klamm{-\kappa/2,\kappa/2}$). The third term of Eq.~(\ref{eq:WettingSph_ExactEq}) can be simplified analogously (compare (\ref{eq:Sph_WallPot}) with (\ref{eq:Sph_PsiIn})) such that we obtain
\begin{align}
&- n_{l} \Psi_{in,R+\delta^\ast}(R + \ell +
\xi_I)\klamm{1+\frac{\xi_I}{R+\ell}}^2
+ V_{sph,R}(R + \ell + \xi_V)\klamm{1+\frac{\xi_V}{R+\ell}}^2
\approx \frac{2\pi}{3\ell^3}\klamm{ n_{l} - \varepsilon_w \sigma_w^6}. 
\label{eq:ExpandNonLocalTermsAsym}
\end{align}
for $\xi_{I,V}\ll \ell$, $R^{-1} \ll 1$ and $\frac{\ell}{R} \ll 1$.
Inserting (\ref{eq:ExpandNonLocalTermsAsym}) and
(\ref{eq:ExpandSurfaceTension}) into
(\ref{eq:WettingSph_ExactEq}) yields,
\begin{align}
\Delta \mu  - \frac{2 \gamma_{lg,\infty}}{R\Delta n} \approx
 \frac{2\pi}{3\ell^3}\klamm{ n_l - \varepsilon_w \sigma_w^6}.\label{eq:SphericalAsymLargeR}
\end{align}
From the left-hand-side of Eq.~(\ref{eq:SphericalAsymLargeR}) it is evident
that the deviation of the chemical potential from saturation times
the density difference, $\Delta \mu \Delta n$, equilibrates
the Laplace pressure $2 \gamma_{lg,\infty}/R$.

This property manifests itself if density profiles for a planar wall
at $\Delta \mu < 0$ are compared with density profiles at saturation
for spherical walls. Here, we choose $R$ such that $2
\gamma_{lg,\infty}/R$ is equal to $\Delta n|\Delta\mu|$.
In Fig. \ref{fig:ComparePlanarSphericalDensityProfiles} we see a
very good agreement between the two density profiles. This result is rather surprising, as we only expected to have an equivalence of the film thicknesses, but obtained equal density profiles for a planar and a spherical substrate. A
similar result was obtained by Stewart and Evans for drying on a
hard spherical wall~\cite{StewartEvans}. 

The approximations made here are also valid for the critical film
thickness $l^\ast$ at which the isotherm crosses the saturation line
as in Fig.~\ref{fig:SphericalIsothermProfiles}. Setting $\Delta \mu$
to zero in (\ref{eq:SphericalAsymLargeR}) leads to an $R^{1/3}$
dependence of the critical film thickness. In
Fig.~\ref{fig:FilmThicknessAtSaturation}, this approximation is
favorably compared with numerical results. Furthermore, it is shown in Fig.~\ref{fig:MuNoddVaryingR} that the chemical potential at the prewetting transition and at the saddle nodes of bifurcation approaches its value in the limit of zero curvature as $\sim 1/R$.

\begin{figure}[hbt]
\centering
\includegraphics[width=10cm]{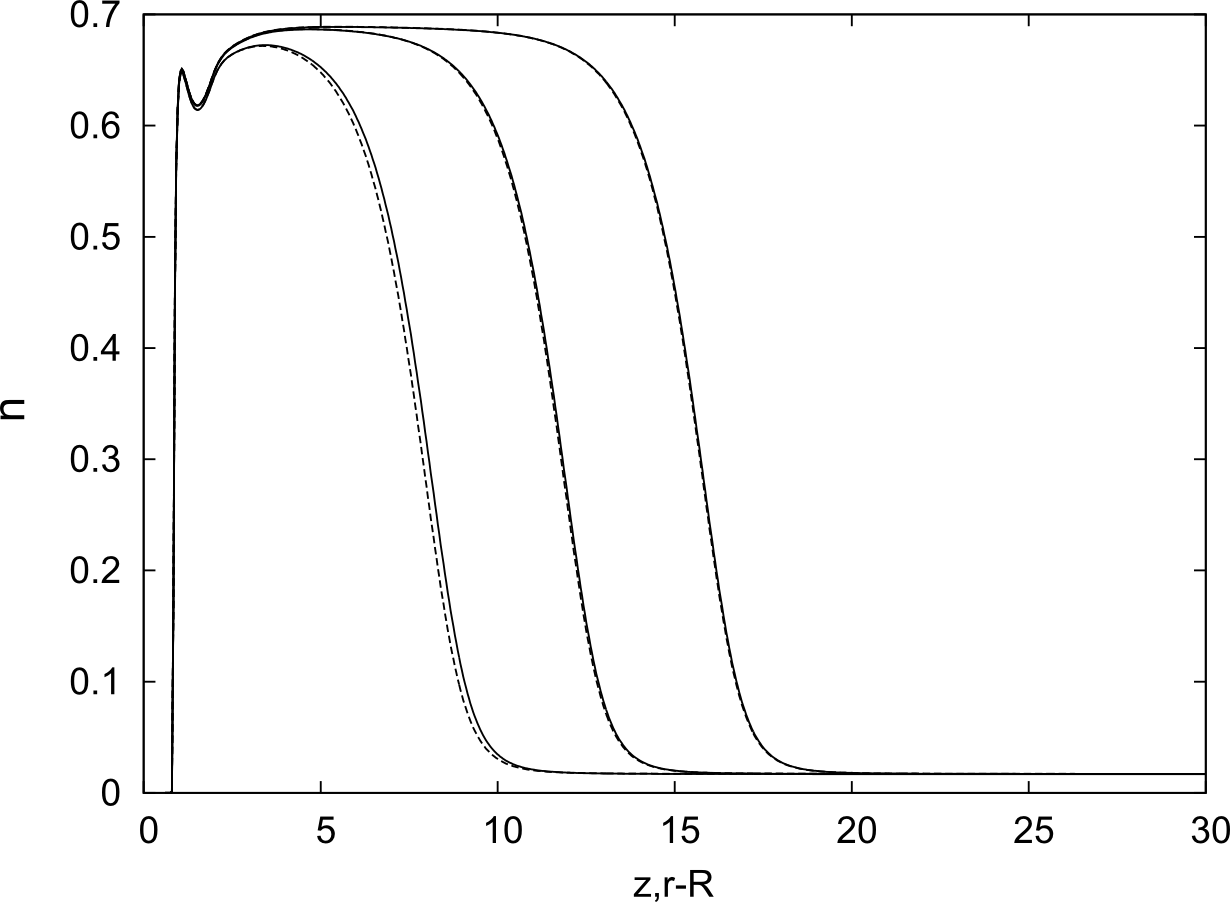}
\caption{Plots of density as a function of the distance from the wall for three
pairs of planar and spherical walls for a temperature $T = 0.7$ and wall parameters
$\varepsilon_w = 0.8$ and $\sigma_w = 1.2$. The planar density profiles
(solid lines) were computed for $\Delta \mu =
-0.00580,-0.00227,-0.00095$, from the left to the right. The spherical density profiles (dashed lines) were
computed at saturation $\Delta \mu = 0$ with radii $R=
264.808,680.057,1616.089$, such that $2
\gamma_{lg,\infty}/R$ in the spherical case
equals to $\Delta n|\Delta \mu|$ in the planar case as in Eq.~(\ref{eq:SphericalAsymLargeR}). As
$\Delta \mu \rightarrow 0^{+}$, the planar/spherical profiles are practically indistinguishable.
\label{fig:ComparePlanarSphericalDensityProfiles}
 }
\end{figure}

\begin{figure}[hbt]
\centering
\includegraphics[width=10cm]{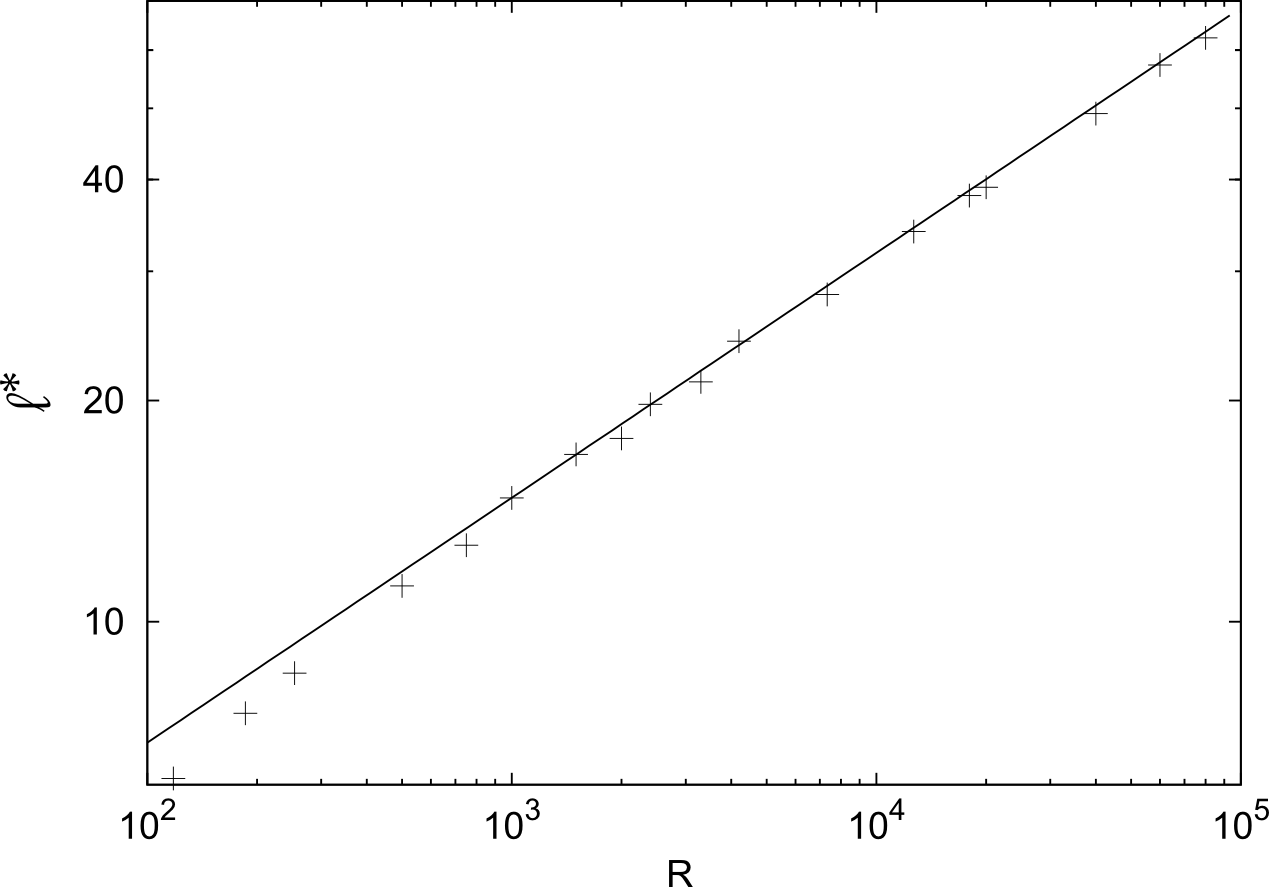}
\caption{Plot of the film thickness at saturation as a function the radius of
a spherical wall for $\varepsilon_w
= 0.8$ and $\sigma_w = 1.25$ at $T
= 0.7$. The solid line is the analytical prediction
(\ref{eq:SphericalAsymLargeR}) for $\Delta \mu = 0$.
\label{fig:FilmThicknessAtSaturation} }
\end{figure}
\begin{figure}[htbp]
\centering
\includegraphics{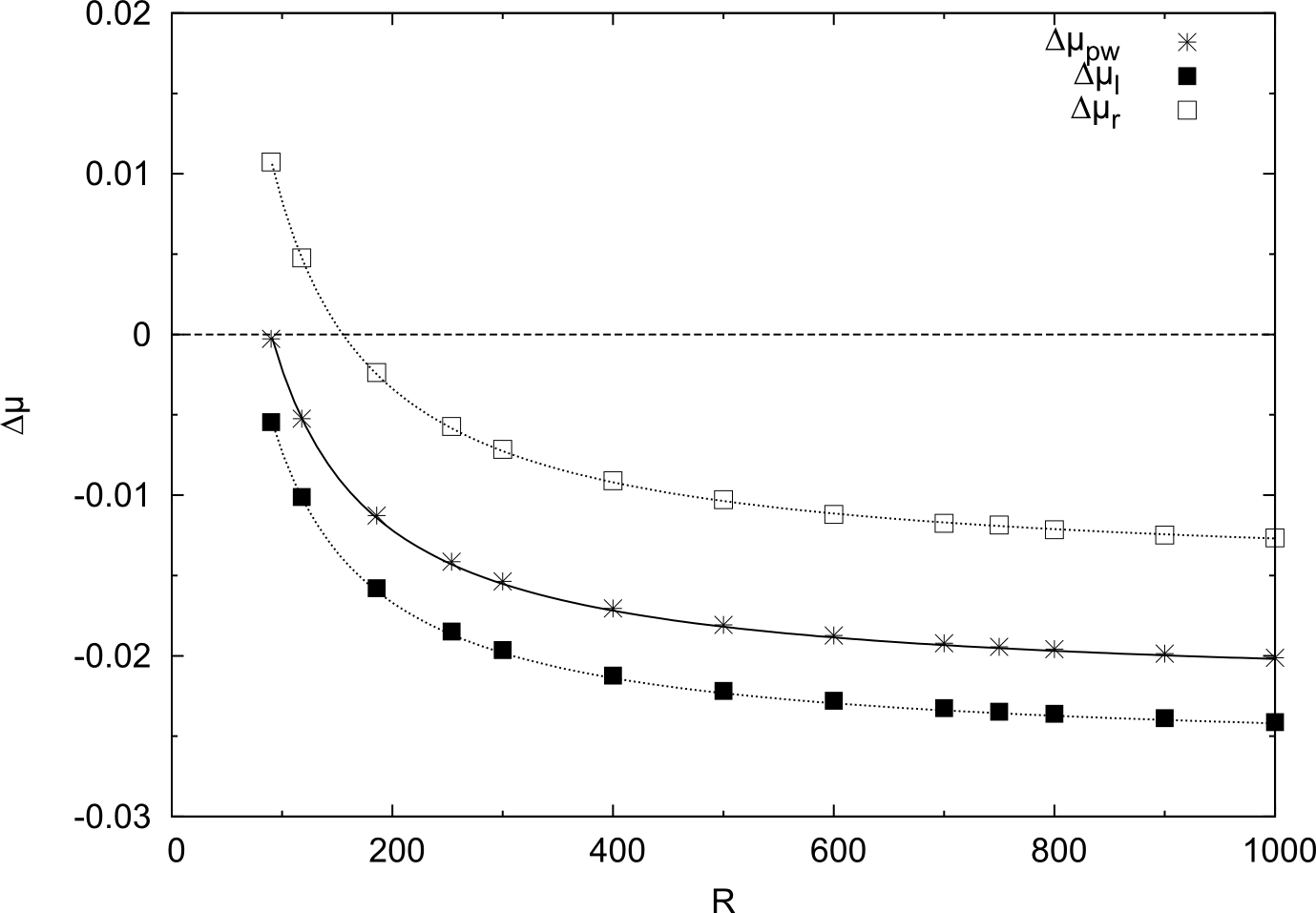}
\caption{Deviations of the chemical potential from the saturation one at prewetting, $\Delta \mu_{pw}$ and at the left and right saddle nodes as a function of the radius of the spherical substrate. The wall parameters are $\sigma_w = 1.25$ and $\varepsilon_w = 0.8$. In the limit of zero curvature, the deviations of the chemical potential converge to the planar values $\Delta\mu_{pw,R=\infty} = -0.0222$, $\Delta\mu_{l,R=\infty} = -0.0261$ and $\Delta\mu_{r,R=\infty} = -0.0150$. The solid and dotted lines are fits to the equation $\Delta\mu_{\{pw,l,r\}}(R) = \Delta\mu_{\{pw,l,r\},R=\infty} + \frac{C_{\{pw,l,r\}}}{R}$. The resulting coefficients are $C_{pw}=2.00$, $C_l=1.8765$ and $C_r = 2.33$.}
\label{fig:MuNoddVaryingR}
\end{figure}

\subsubsection{Asymptotic Behavior for Large Film Thickness \label{sec:SphericalAsymptoticFT}}

Assume now that the radius of the wall $R$ is fixed, whereas the
film thickness $\ell$ is increased. It can be shown that for large
$\ell$, the external potential $V_{sph,R}(R+l+\xi_V)$ as well as the
contribution of the binding potential
$\Psi_{in,R+\delta^\ast}\klamm{R + l + \xi_I }$ have an
$\ell^{-3}$-leading-order behavior. 
For $\tilde r = \ell + \xi_V$, we get
\begin{align*}
V_{sph,R}(R + \tilde r) &= 
\varepsilon_w \sigma_w^3 \pi 
\frac{\sigma_w}{3(\tilde r + R)}\klamm{
\frac{\sigma_w^8}{30}\klamm{
\frac{\tilde r+ 10R}{(\tilde r+ 2 R)^9} - \frac{\tilde r-8R}{\tilde r^9}
}+ 
\sigma_w^2\klamm{
\frac{\tilde r-2R}{\tilde r^3} - \frac{\tilde r+4R}{(\tilde r+2 R)^3}
}} \\
&= O\klamm{\tilde r^{-3}}\\
\Rightarrow V_{sph,R}(R + \ell + \xi_V) =&  O\klamm{ \ell^{-3}} .
\end{align*}
We resume that in (\ref{eq:WettingSph_ExactEq}), in the limit $\ell \to \infty$, the deviation of the chemical
potential $\Delta \mu$ has to balance the second term like,
\begin{align}
\Delta \mu = \frac{2\gamma_{lg,\infty}}{\Delta n (R+\ell)} +
{O}\klamm{\frac{1}{\ell^2}}, \label{eq:SphericalPrediction_FT}
\end{align}
where we replaced the spherical surface tension by the planar one
from Eq.~(\ref{eq:SurfaceTensionSpherical_Tolman}). This analytical result is in very
good agreement with numerical results obtained from the full system
as shown in Fig.~\ref{fig:SphericalAsymptoticBehaviour}.
Consequently, the unstable branch of the spherical isotherms
approaches saturation asymptotically as $\sim \Delta \mu^{-1}$.

\begin{figure}
\centering
\includegraphics[width=10cm]{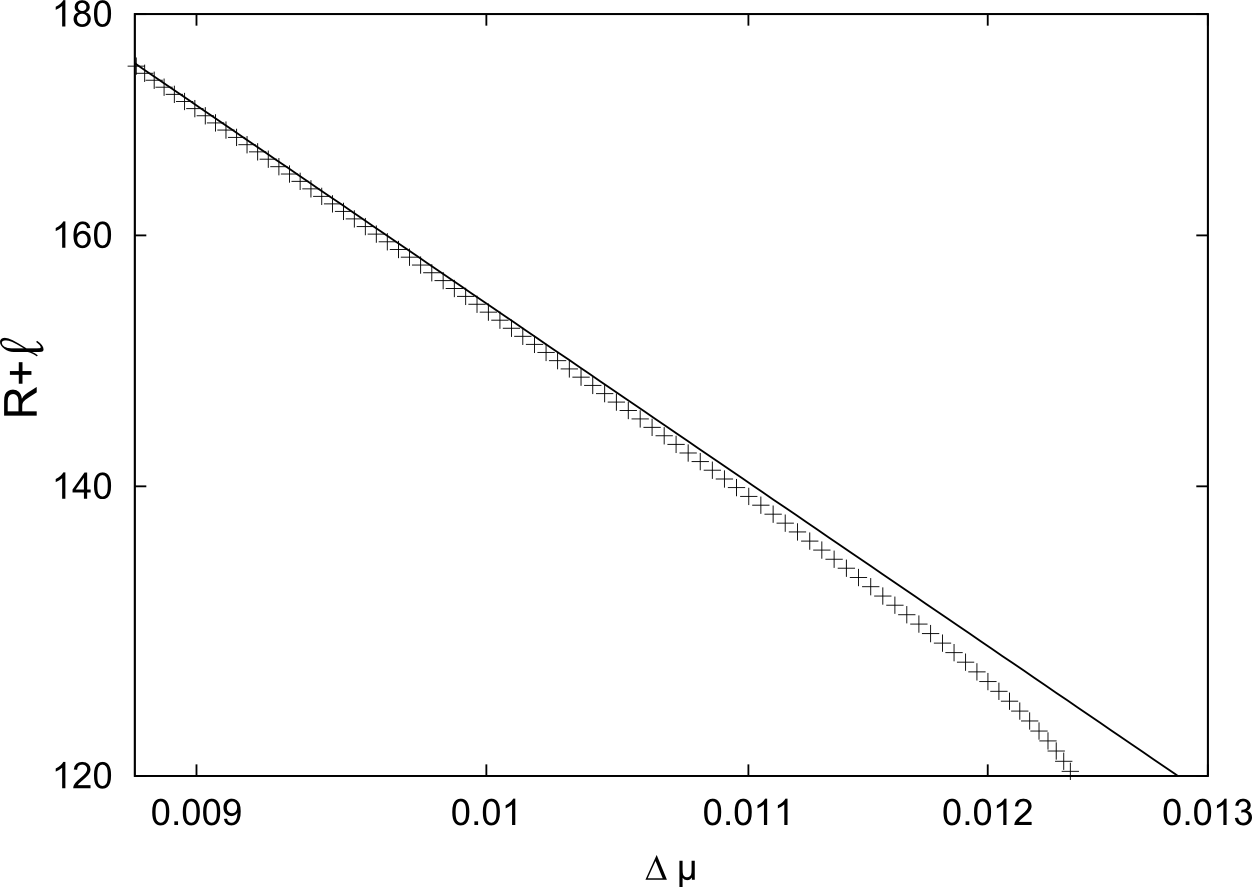}
\caption{Log-log plot of the film thickness $(R+\ell)$ as
a function of the deviation of the chemical potential from saturation for a spherical wall with
radius $R = 100$ and parameters $\varepsilon_w =
0.8$ and $\sigma_w = 1.25$ at temperature
$T = 0.7$. The solid line is the prediction in Eq.~(\ref{eq:SphericalPrediction_FT}) with $\gamma_{lg,\infty}=0.517$
and the crosses are calculations for a planar liquid-gas interface and $\Delta n
= 0.6709$. \label{fig:SphericalAsymptoticBehaviour} }
\end{figure}

Finally, we note that the good agreement between the analytical and
numerical results suggests that two of the basic underlying
assumptions in our analysis need not be improved. These are: (i) the
fact that the influence of the tails of the liquid-gas interface has
not been addressed; (ii) we have assumed that the shape of the
liquid-gas density profile of a droplet does not change with its
radius.

\chapter{Conclusion}

We have analyzed wetting of a simple fluid on one-dimensional
substrates, namely a planar wall and a sphere, which are
characterized by translational and rotational invariances,
respectively. The fluid was modeled through density-functional
theory coupled with a simple mean-field approach. 

We have evaluated the surface tension of a liquid-gas interface for a fixed temperature as a function of the steepness of a $\tanh$-density profile. We conclude that due to big differences between the global minimum and the sharp interface limit, the sharp interface approximation (SIA) is no appropriate method for the prediction of the liquid-gas surface tension. This property leads to inconsistencies if the SIA is used for the prediction of the wetting behavior on curved substrates. 

A method which avoids these inconsistencies is the piecewise-function approximation (PFA), where the density is assumed to be everywhere constant except at the wall-liquid and the liquid-gas interface where it varies smoothly. We have introduced formalisms in order to apply the PFA and the SIA for general geometries. 

Additionally to these analytical techniques, we minimize the grand potential numerically. We have introduced a novel pseudo arc-length continuation procedure to trace the full bifurcation diagram of the adsorption over the chemical potential including unstable branches and saddle nodes as functions of the temperature. 

Our main results can be summarized as follows:
\begin{enumerate}
\item We have examined numerically the jump of the film thickness at
    the prewetting transition as a function of temperature and the attractiveness of the wall. We have also given phase diagrams for the deviation of the chemical potential from saturation at the prewetting transition as a function of the temperature and the wall parameter $\varepsilon_w$. In the vicinity of the transition from a complete wetting to a prewetting scenario, we have shown isotherms of the film thickness over the chemical potential .
\item In the case of a spherical wall, the numerical results show an additional 	first-order wetting transition at saturation. In contrast to the isotherm for a planar wall, the film thickness does not go to infinity as saturation is approached. Instead, we get a maximal film thickness $\ell^\ast$.
\item We have examined analytically the wetting behavior on
    curved substrates with the PFA. Unlike the planar case, the liquid-gas surface tension has an influence on the asymptotic behavior of the isotherm. A number of
    auxiliary parameters have been introduced as representatives
    for the effect of the exact shape of the liquid-gas and
    wall-liquid density profiles. This has allowed us to perform
    two separate expansions in the film thickness and in the
    radius of the wall. The SIA, quite
    popular for planar substrates, is shown to lead to
    inaccurate predictions for wetting on curved substrates. On
    the other hand, our PFA offers
    a relatively simple and self-consistent way to examine
    wetting on curved substrates.
\item We have shown analytically that the maximal film
    thickness $\ell^\ast$ has a
    leading-order behavior $\sim R^{1/3}$, $R \gg 1$, where $R$
    is the radius of the spherical wall. As a result, we can
    obtain the dependence of a (microscopic)
    critical film thickness on the radius of a (mesoscopic)
    wall.
\item We have shown analytically that for $R \gg \ell \gg 1$,
    $\Delta n \Delta \mu$ equilibrates the Laplace pressure
    $2\gamma_{lg,\infty}/R$. This is in agreement with numerical results, where we have compared the density
    profiles of a thin film on a planar substrate at
    $\Delta \mu < 0$ with the density profiles on a spherical
    substrate at saturation such that $\Delta n|\Delta \mu|  =
    2\gamma_{lg,\infty}/R$, for which we obtained a very good
    agreement. We have also shown that the shift of the chemical potential at the prewetting transition, $\Delta \mu_{pw}$ as a function of the radius of the substrate, approaches its planar limit as $~1/R$.
\item We have shown the appearance of a second unstable branch
    of the isotherm in the spherical case. This branch
    approaches saturation asymptotically from the right with
    $\sim \Delta \mu^{-1}$ as $\Delta \mu \to 0^+$. Again,
    comparison of the analytical with the numerical results gave
    a very good agreement.
\end{enumerate}

We believe that the model presented here allows for a qualitative
description of the microscopic behavior of thin films on solid
substrates. Nevertheless, a number of improvements can be made: (i)
by using a fundamental measure theory for the reference part of the
fluid which, in general, would lead to more pronounced oscillatory
effects close to the wall; (ii) use more accurate models for the
hard sphere diameter $d$; (iii) refine the attractive part of the
model which would lead to a more accurate prediction of the
homogeneous limit. The model used in this work allows drawing
qualitative conclusions about the wetting behavior on substrates.

\chapter{Acknowledgments}

I owe my deepest gratitude to my supervisor Prof. Serafim Kalliadasis, whose experienced guidance and support from the initial to the final level enabled me to develop a deep understanding of the subject. His encouragement and sound advice pushed me to think beyond the borders of what I was used to. 

My very special thanks goes to Prof. Martin Oberlack, who supervised this thesis in my home university and to whom I am deeply indebted for his support and for his guidance towards the principles of scientific thinking and working during the last years of my studies. I would also like to thank my tutor Prof. Klaus Keimel, whose support and advice helped me substantially to organize my studies and take important decisions such as the application for a Diploma thesis abroad.

This thesis would not have been possible without the excellent support and the many discussions with my colleagues and members of the Imperial College London. Particularly, I want to thank Alexandr Malijewski for fruitful discussions considering DFT and wetting phenomena on spherical substrates as well as for critically reading the thesis and making several valuable comments and suggestions.
I am very thankful to Antonio Pereira, who introduced us to the numerical principles of DFT computations during a stay in Nancy.
My visit to Nancy was funded by the \emph{Multiflow} Network". I also want to thank my colleague Peter Yatsyshin, with whom I visited Antonio Pereira in Nancy, for many discussions about DFT and statistical mechanics. I am grateful to Marc Pradas for many discussions about analytical approaches and to Rajagopal Vellingiri and Vlad Novak for discussions about physical interpretations of the results. I also thank Nikos Savva for his ready help concerning problems with the computations and the software. Finally, I thank my friends Betsy Voigt and Sandro Gorini for helping me during the application process and the correction of this thesis.

My stay in London was funded by the Rotary Clubs Darmstadt,
Darmstadt-Bergstra\ss e and Darmstadt-Kranichstein. I want to thank all members of these Rotary-Clubs and especially Dr. Heiner Diefenbach, who is the president of the Rotary-Club Darmstadt-Kranichstein which assigned the fellowship for students studying abroad in 2009. Throughout my studies, I was supported by the Studienstiftung des Deutschen Volkes. Their readiness to support me financially during my studies helped me a lot to plan and organize projects such as the stay at Imperial College London. In particular, I want to thank the Studienstiftung for their recommendation towards the Dr.-Jürgen-Ulderup foundation. I want to thank this foundation for offering a scholarship for my stay in London. I also want to thank Prof. Hans-Dieter Alber for his supervision during my time as a fellow at the Studienstiftung des Deutschen Volkes. 

Finally, I offer my regards to all those who supported me in any respect during the completion of the project. I enjoyed the work in the {\it Complex Fluid Flows} group at Imperial College London and in the {\it Multiflow} Network and hope to continue this fruitful cooperation in future projects.

\vspace{1cm}

This thesis has been reviewed by Alexandr Malijewsk\'y, whose corrections and suggestions were taken into consideration in the revised version of this thesis. I thank Alexandr for carefully reading and commenting on this thesis. 
\vspace{0.5cm}

I also thank my colleague David Sibley for carefully reading the thesis and for his useful comments, which were also taken into consideration in this revised version.

\bibliographystyle{ThesisStyle}
\bibliography{ref}

\chapter{Appendix}
\vspace{3cm}
\begin{table}[ht]
\begin{center}
\begin{tabular}[ht]{llll}
\toprule
$T$ & $\mu_{sat}$ & $n_{g}$ & $n_{l}$\\\midrule 
$0.25$ & $-4.718520757$ & $1.213644469\cdot 10^{-8}$ & $1.127211727$\\
$0.3$ & $-4.434056675$ & $7.278758914\cdot 10^{-7}$ & $1.069152485$ \\
$0.4$ & $-4.011264150$ & $8.446901093\cdot 10^{-5} $ & $0.9649884199$\\
$0.5$ & $-3.728205634$ & $1.126155592 \cdot 10^{-3}$ & $0.8700695119$\\
$0.55$ & $-3.626461863$ & $2.732555380 \cdot 10^{-3}$ & $0.8243739417$\\
$0.6$ & $-3.546705118$ & $5.609208593 \cdot 10^{-3}$ & $0.7791390130$\\
$0.65$ & $-3.486303570$ & $1.021152841 \cdot 10^{-2}$ & $0.7338582420$\\
$0.7$ & $-3.443050923$ & $1.703790008\cdot 10^{-2}$ & $0.6879908811$ \\
$0.75$ & $-3.415074850$ & $2.668164774\cdot 10^{-2}$ & $0.6408950758$ \\
$0.8$ & $-3.400773052$ & $3.994135433\cdot 10^{-2}$ & $0.5917145822$ \\
$0.85$ & $-3.398765436$ & $5.805717753\cdot 10^{-2}$ & $0.5391442686$ \\
$0.9$ & $-3.407856740$ & $8.332134885\cdot 10^{-2}$ & $0.4808226348$\\ 
$0.95$ & $-3.427006546$ & $0.1213047530$ & $0.4111080056$\\ 
$1.0$  & $-3.455304914$ & $0.2028151078$ & $0.2991224633$\\
$1.003$ & $-3.457275772$ & $0.2155452967$ & $0.2846022041$\\
$1.006172833$ & $-3.459392667$ &  $0.2491294675$ & $0.2491294675$\\\bottomrule
\end{tabular}
\caption{Values of the bulk gas and bulk liquid densities and the chemical potential at saturation for different temperatures. At saturation, the bulk liquid and gas phases are equally stable. All values are in dimensionless form (see Sec.~\ref{sec:StatMech_OurModel})}
\label{tab:UniformCoexistence}
\end{center}
\end{table}

\begin{figure}[ht]
\centering
\includegraphics[height=9cm]{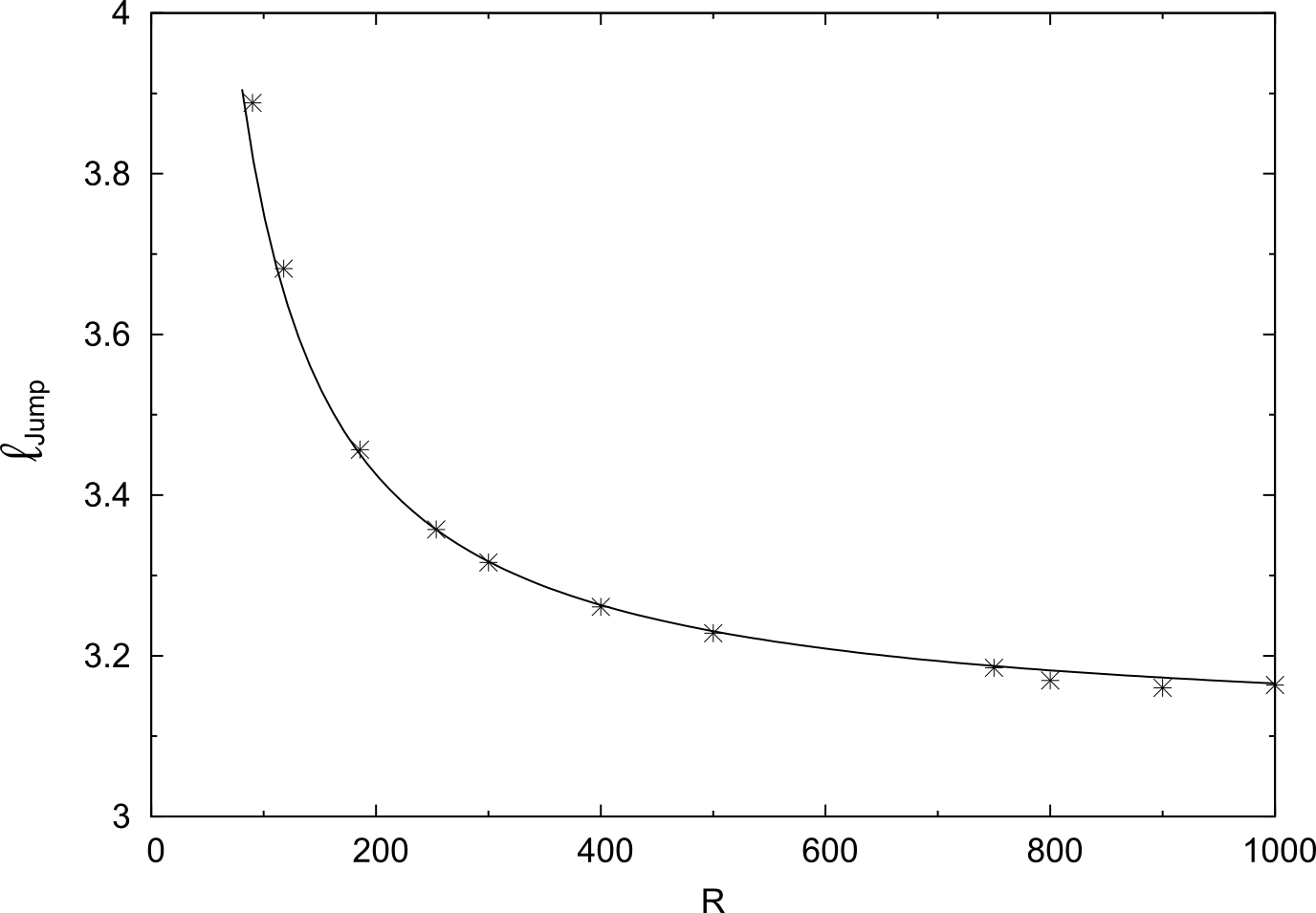}
\caption{Jump of the film thickness at the prewetting transition vs. the radius of the spherical substrate. In the limit of zero curvature, the jump is $\ell_{jump,\infty} = 3.10067$. The solid line is a fit to the equation $\ell_{jump}(R) = \ell_{jump,\infty} + \frac{C}{R}$, where the resulting coefficient is $C =65.0$.}
\label{fig:JumpVaryingR}
\end{figure}

\begin{figure}[ht]
\begin{center}
\subfigure[$T=0.40$]{
\includegraphics[width=4.5cm]{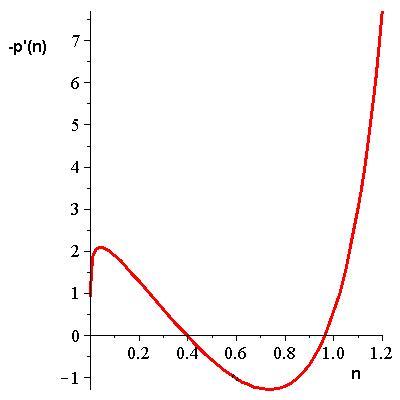}
}
\subfigure[$T=0.40$]{
\includegraphics[width=4.5cm]{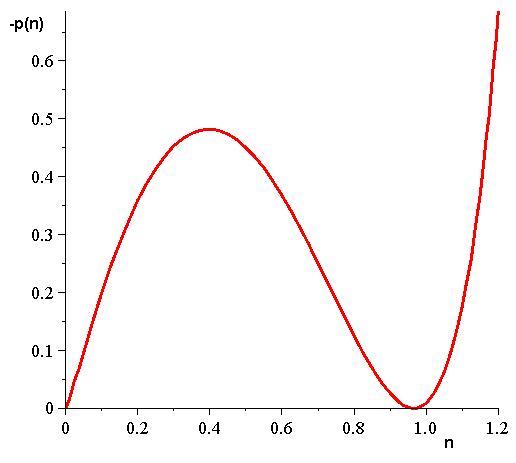}
}\\
\subfigure[$T=0.70$]{
\includegraphics[width=4.5cm]{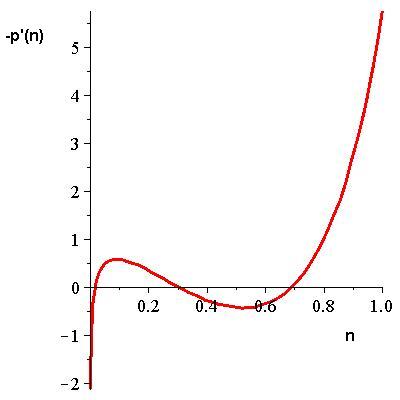}
}
\subfigure[$T=0.70$]{
\includegraphics[width=4.5cm]{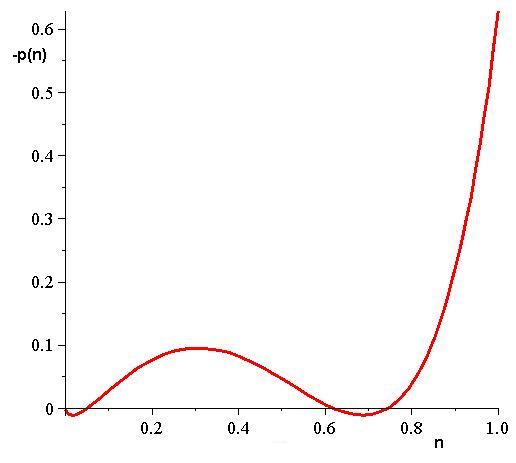}
}\\
\subfigure[$T=0.90$]{
\includegraphics[width=4.5cm]{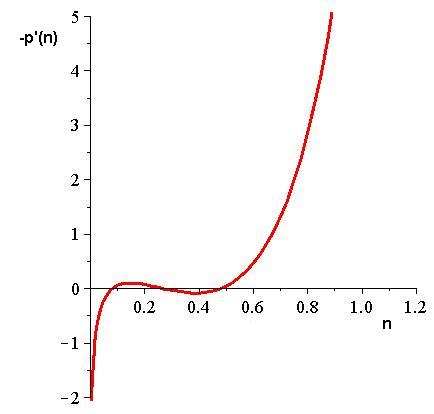}
}
\subfigure[$T=0.90$]{
\includegraphics[width=4.5cm]{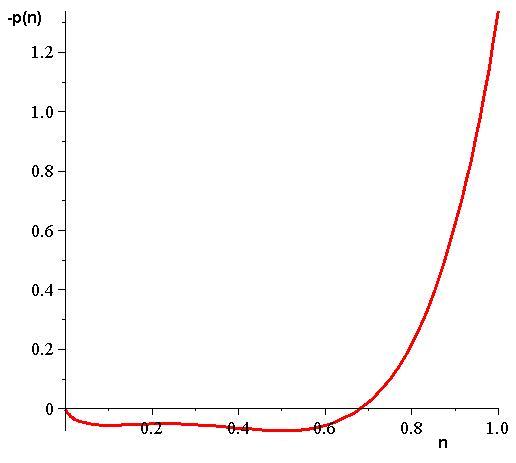}
}\\
\subfigure[$T=1.003$]{
\includegraphics[width=4.5cm]{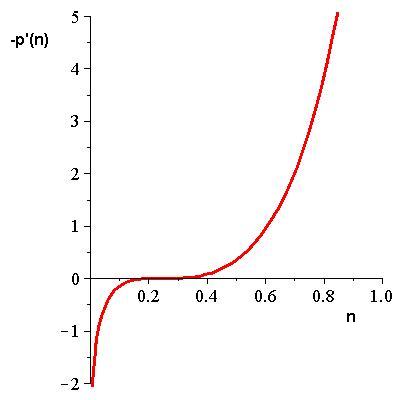}
}
\subfigure[$T=1.003$]{
\includegraphics[width=4.5cm]{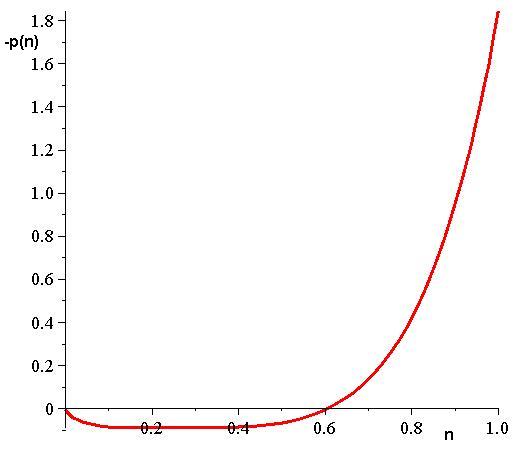}
}
\caption{Plots of the negative pressure (\ref{eq:Uniform_Omega_def}) and its derivative as a function of the uniform density for different temperatures.}
\end{center}
\label{fig:OmegaOverN_AnnexDiagrams}
\end{figure}

\section{Barker and Henderson Approach \label{sec:BarkerHendersonApproach}}

This expansion was first presented by Barker and Henderson in 1967~\cite{BarkerHenderson1} who developed an expression for the Helmholtz free energy $F$ of a homogeneous canonical system. In a canonical system, the free energy is given by
\begin{align}
F = - \beta^{-1} \ln Z_C, \label{eq:StatMech_BarkerHenderson_FreeEnergyZC}
\end{align}
where the partition function $Z_C$ is defined in (\ref{eq:StatMech_Can_PartitionFunction}). This expression is directly linked with the Hamiltonian of the system, which is the sum of the kinetic energy $E_k$, the particle interaction energy $U$ and the external energy $V_{ext}$. Neglecting the external potential, the canonical partition function can be rewritten as follows:
\begin{align}
Z_C &= \frac{1}{h^{3N}N!} 
\iint e^{-\beta U_{attr}} e^{-\beta\klamm{E_k + U_{HS}}} d{\bf p}^N d{\bf r}^N, \label{eq:StatMech_BarkerHenderson_ZC}
\end{align}
where the perturbation $U_{attr}$ can be written as a sum of pair potentials 
\begin{align}
U_{attr}\klamm{{\bf r}^N} = \frac{1}{2}\sum_{i\neq j} \phi_{attr}\klamm{|{\bf r}_i - {\bf r}_j|}.  \label{eq:StatMech_DefUAttr}
\end{align}
Now, Barker and Henderson assumed that $N_i$ is the number of pairs of particles which have a distance $r \in [R_i,R_{i+1})$, where $(R_i)_i$ is a sequence of increasing positive real numbers. In this case, the perturbation can be written in discrete form as 
\begin{align}
U_{attr}\klamm{{\bf r}^N} = \sum_i N_i \phi_{attr}^i,
\end{align}
where $\phi_{attr}^i$ is an approximation for the attractive interaction energy of particles with distance $r \in [R_i,R_{i+1})$. We say that for a reference hard-sphere fluid, the number of pairs of particles at distance $[R_i,R_{i+1})$ is $\avar{N_i}_{HS}$. This leads to the following split of  $N_i$ into the average with respect to the reference fluid plus a respective deviation:
\begin{align}
U_{attr} = \sum_{i} \avar{N_i}_{HS} \phi_{attr}^i + \sum_{i} \klamm{ N_i - \avar{N_i}_{HS}} \phi_{attr}^i. \label{eq:StatMech_BarkerHenderson_Uattr}
\end{align}
The average $\avar{\cdot}_{HS}$ with respect to the reference system is defined as
\begin{align*}
\avar{\cdot}_{HS} \defi \frac{1}{Z_{C,HS}} \frac{1}{h^{3N}N!} \iint \cdot \; e^{-\beta \klamm{ E_k + U_{HS} } } d{\bf p}^N d{\bf r}^N,
\end{align*}
where $Z_{C,HS}$ is the canonical partition function of the reference system:
\begin{align*}
Z_{C,HS} \defi \frac{1}{h^{3N}N!} \iint e^{-\beta \klamm{ E_k + U_{HS} } } d{\bf p}^N d{\bf r}^N.
\end{align*}

The result obtained for (\ref{eq:StatMech_BarkerHenderson_Uattr}) is inserted in
(\ref{eq:StatMech_BarkerHenderson_ZC}):
\begin{align*}
Z_C = Z_{C,HS} \exp\klamm{ -\beta \sum_{i} \avar{N_i}_{HS} \phi_{attr}^i} \avar{ \exp\klamm{-\beta \sum_{i} \klamm{ N_i - \avar{ N_i}_{HS}} \phi_{attr}^i} }_{HS}.
\end{align*}
This equation is inserted in (\ref{eq:StatMech_BarkerHenderson_FreeEnergyZC}), which establishes a link between the partition function $Z_C$ and the Helmholtz free energy. The first term on the right hand side of the equation above results in the hard-sphere free energy $F_{HS} = - \beta^{-1}\ln Z_{C,HS}$ of the reference system. Hence, we get:
\begin{align*}
F - F_{HS} = \sum_{i} \avar{N_i}_{HS} \phi_{attr}^i  - \beta^{-1} 
\ln 
 \avar{ \exp\klamm{-\beta \sum_{i} \klamm{ N_i - \left\langle N_i \right\rangle} \phi_{attr}^i} }_{HS}
\end{align*}
Remark that in the last term, the argument of the exponential function is zero if $N_i$ equals $\avar{N_i}_{HS}$, i.e. if the fluid corresponds to the reference fluid. Hence, we expand the exponential function around zero and the logarithm around one. This yields
\begin{align}
F - F_{HS} = \sum_{i} \avar{N_i}_{HS} \phi_{attr}^i  - \frac{\beta}{2} \sum_{i,j} \klamm{ \avar{ N_i N_j}_{HS} - \avar{ N_i}_{HS} \avar{ N_j}_{HS} } \phi_{attr}^i\phi_{attr}^j + O\klamm{\beta^2}. \label{eq:Pert_FreeEnergy1}
\end{align}

We now have to find expressions for the average of $N_i$ as well as for the covariance of $N_i$ and $N_j$, given by $\avar{ N_i N_j}_{HS}-\avar{ N_i}_{HS}\avar{N_j}_{HS}$. 
$\avar{N_i}_{HS}$ is the average number of pairs of particles in a homogeneous hard-sphere fluid which have the distance $|{\bf r}_i - {\bf r}_j| \in [R_i,R_{i+1})$. This can be written in terms of the two particle distribution $n^{(2)}_{HS}\klamm{{\bf r_1},{\bf r_2}}$, i.e. in terms of the average probability that there will be two particles at the positions ${\bf r}_1$ and ${\bf r}_2$ simultaneously~\cite{Plischke}.
\begin{align}
n^{(2)}_{HS}\klamm{{\bf r_1},{\bf r_2}} \defi \avar{ \delta\klamm{{\bf r_1}-{\bf r}}  \delta\klamm{{\bf r_2}-{\bf r}}
}_{HS}. \label{eq:HStwoparticleDistr}
\end{align}
This yields
\begin{align}
\avar{N_i}_{HS} &= \frac{1}{2} \iint_{R_i < |{\bf r}-{\bf r}'| < R_{i+1} }  n^{(2)}_{HS}\klamm{{\bf r},{\bf r}'} d{\bf r}'d{\bf r}. \label{eq:StatMech_DefNi}
\end{align}
We can rewrite the first term in (\ref{eq:Pert_FreeEnergy1}) in continuum description and get
\begin{align}
\sum_{i} \avar{ N_i}_{HS} \phi_{attr}^i &=
\frac{1}{2} \sum_i \iint_{R_i < |{\bf r}-{\bf r}'| < R_{i+1} }  n^{(2)}_{HS}\klamm{{\bf r},{\bf r}'}  \phi_{attr}^i  d{\bf r}'d{\bf r} \notag\\
&\approx \frac{1}{2} \iint n^{(2)}_{HS}\klamm{{\bf r},{\bf r}'} \phi_{attr} \klamm{|{\bf r}-{\bf r}'|} d{\bf r}' d{\bf r}. \label{eq:StatMech_BH_1stTerm}
\end{align}

Now, an approximation for the second term of expansion (\ref{eq:Pert_FreeEnergy1}) has to be found. In a homogeneous system, $N_i$ represents the number of particles in a spherical shell surrounding one molecule.  For large macroscopic shells, Barker and Henderson assumed that the number of molecules in different shells is uncorrelated. This means that we can set $\left\langle N_i N_j \right\rangle -\left\langle N_i \right\rangle\left\langle N_j \right\rangle = 0$ for $i \neq j$. 
For a known average density $n$, the fluctuation of the number of particles can be written as 
\begin{align}
\avar{N_i^2} - \avar{N_i}^2
 = \beta^{-1} \avar{N_i}\diff{n}{p}, \label{eq:Pertur_GlobalCompress}
\end{align}
where $p$ is the pressure of the system. For more details, see also Plischke and Bergersen~\cite[p.42]{Plischke}. Consequently, the second term on the right hand side of (\ref{eq:Pert_FreeEnergy1}) can be written as:
\begin{align}
- \frac{\beta}{2}\sum_{i,j} \klamm{ \avar{ N_i N_j}_{HS} - \avar{ N_i}_{HS} \avar{ N_j}_{HS} } \phi_{attr}^i\phi_{attr}^j
&=
- \frac{1}{2} \sum_i \avar{N_i}_{HS} \klamm{\phi_{attr}^i}^2 \diff{n}{p} \notag \\
&
{\overset{(\ref{eq:StatMech_DefNi})}\approx}
- \frac{1}{4} \iint n_{HS}^{(2)}\klamm{{\bf r},{\bf r}'} \phi_{attr}^2\klamm{|{\bf r} -{\bf r}' |} \diff{n}{p} d{\bf r}' d{\bf r}
\label{eq:StatMech_BH_SecondTerm}
\end{align}

Inserting (\ref{eq:StatMech_BH_SecondTerm}) and (\ref{eq:StatMech_BH_1stTerm}) into (\ref{eq:Pert_FreeEnergy1}) gives the following equation for the Helmholtz free energy
\begin{align*}
F-F_{HS} =  \frac{1}{2} \iint n_{HS}^{(2)}\klamm{{\bf r},{\bf r}'} \klamm{
\phi_{attr}\klamm{|{\bf r} -{\bf r}'|}
-
\frac{1}{2}\phi_{attr}^2\klamm{|{\bf r} -{\bf r}' |}  \diff{n}{p}
  } d{\bf r}' d{\bf r}
+ O\klamm{\beta^2}
\end{align*}

In the homogeneous case, the two particle distribution $n_{HS}^{(2)}$ can be written in terms of the pair distribution function $g_{HS}(r)$, which is often also referred to as radial distribution function. It is defined by
\begin{align}
n_{HS}^{(2)}\klamm{r} = n^2 g_{HS}(r). \label{eq:StatMech_PairDistributionFunction}
\end{align}
We then get
\begin{align*}
F-F_{HS} =  \frac{n^2}{2} \iint g_{HS}\klamm{|{\bf r}-{\bf r}'|} \klamm{
\phi_{attr}\klamm{|{\bf r} -{\bf r}'|}
-
\frac{1}{2}\phi_{attr}^2\klamm{|{\bf r} -{\bf r}' |}  \diff{n}{p}
  } d{\bf r}' d{\bf r}
+ O\klamm{\beta^2}
\end{align*}

This approximation is based on relation (\ref{eq:Pertur_GlobalCompress}), which takes into account the pressure-derivative of the global density. Including the local density at a certain distance from a given molecule into the pressure derivative leads to an expansion using a "local compressibility" term $\diff{\klamm{n\cdot g_{HS}}}{p}$:
\begin{align*}
F-F_{HS} =  \frac{n^2}{2} \iint
 g_{HS}\klamm{|{\bf r}-{\bf r}'|} \phi_{attr}\klamm{|{\bf r} -{\bf r}'|}
-
\frac{1}{2}\phi_{attr}^2\klamm{|{\bf r} -{\bf r}' |}  \diff{\klamm{n \cdot g_{HS}\klamm{|{\bf r}-{\bf r}'|} }}{p}
 d{\bf r}' d{\bf r}
+ O\klamm{\beta^2}
\end{align*}

Johannessen, Gross and Bedeaux~\cite{NonequilibriumJohannsen}, extend this approach to inhomogeneous systems by evaluating the pair distribution function at the average density $\bar n_{{\bf r},{\bf r}'}\defi\frac{1}{2}\klamm{n\ofR + n\ofRD}$:
\begin{align*}
F[n\ofR]-F_{HS}[n\ofR] = &\frac{1}{2} \iint n\ofR n\ofRD g_{HS}\klamm{{\bf r},{\bf r}',\bar n_{{\bf r},{\bf r}'}} \phi_{attr}\klamm{|{\bf r}-{\bf r}'|}
d{\bf r}' d{\bf r}
\\
&- \frac{1}{4}  \iint n\ofR n\ofRD \diff{\klamm{\bar n_{{\bf r},{\bf r}'} \cdot g_{HS}\klamm{{\bf r},{\bf r}',\bar n_{{\bf r},{\bf r}'}}}}{p}  \phi_{attr}^2\klamm{|{\bf r}-{\bf r}'|} d{\bf r}' d{\bf r}
+ O\klamm{\beta^2}
\end{align*}

Remark that in this expansion, the second term on the right hand side does not involve many-body correlation functions.

\end{document}